\documentclass[11pt,a4paper]{article}

\usepackage{amsmath,amsthm,amsfonts,amssymb}
\usepackage[a4paper,margin=2cm]{geometry}
\usepackage{colonequals}
\usepackage{graphicx,tikz}
\usepackage{ifthen}
\usepackage{authblk}

\usepackage{algorithm2e}

\newboolean{ElectronicVersion}
\setboolean{ElectronicVersion}{true}

\ifthenelse{\boolean{ElectronicVersion}}{
    \usepackage[a4paper=true,pdftex,bookmarks,pagebackref,
	plainpages=false, 
        pdfpagelabels=true 
        ]{hyperref}}{}

\makeatletter
\newtheorem*{rep@theorem}{\rep@title}
\newcommand{\newreptheorem}[2]{%
\newenvironment{rep#1}[1]{%
 \def\rep@title{#2 \ref*{##1}}%
 \begin{rep@theorem}}%
 {\end{rep@theorem}}}
\makeatother

\usepackage{hyperref}
\hypersetup{
    bookmarksnumbered=true, 
    unicode=false, 
    pdfstartview={FitH}, 
    pdftitle={Multidimensional Quantum Walks, Recursion, and Quantum Divide and Conquer}, 
    pdfauthor={Stacey Jeffery and Galina Pass}, 
    pdfsubject={}, 
    pdfcreator={}, 
    pdfproducer={}, 
    pdfkeywords={}, 
    pdfnewwindow=true, 
    colorlinks=true, 
    linkcolor=blue, 
    citecolor=blue, 
    filecolor=blue, 
    urlcolor=blue 
}

\usepackage{color}
\definecolor{darkgreen}{rgb}{0,.5,0}
\definecolor{darkred}{rgb}{.7,.3,.3}
\definecolor{deepblue}{rgb}{0,.1,.7}

\newcommand{\eq}[1]{\hyperref[eq:#1]{(\ref*{eq:#1})}}
\renewcommand{\sec}[1]{\hyperref[sec:#1]{Section~\ref*{sec:#1}}}
\newcommand{\thm}[1]{\hyperref[thm:#1]{Theorem~\ref*{thm:#1}}}
\newcommand{\lem}[1]{\hyperref[lem:#1]{Lemma~\ref*{lem:#1}}}
\newcommand{\cor}[1]{\hyperref[cor:#1]{Corollary~\ref*{cor:#1}}}
\newcommand{\app}[1]{\hyperref[app:#1]{Appendix~\ref*{app:#1}}}
\newcommand{\tabl}[1]{\hyperref[tab:#1]{Table~\ref*{tab:#1}}}
\newcommand{\defin}[1]{\hyperref[def:#1]{Definition~\ref*{def:#1}}}
\newcommand{\fig}[1]{\hyperref[fig:#1]{Figure~\ref*{fig:#1}}}
\newcommand{\clm}[1]{\hyperref[clm:#1]{Claim~\ref*{clm:#1}}}
\newcommand{\conj}[1]{\hyperref[conj:#1]{Conjecture~\ref*{conj:#1}}}
\newcommand{\rem}[1]{\hyperref[rem:#1]{Remark~\ref*{rem:#1}}}
\newcommand{\examp}[1]{\hyperref[ex:#1]{Example~\ref*{ex:#1}}}

\newcommand{\thmthm}[2]{\hyperref[thm:#1]{Theorem~\ref*{thm:#1}} and~\hyperref[thm:#2]{\ref*{thm:#2}}}
\newcommand{\lemlem}[2]{\hyperref[lem:#1]{Lemma~\ref*{lem:#1}} and~\hyperref[lem:#2]{\ref*{lem:#2}}}

\newtheorem{theorem}{Theorem}[section]
\newtheorem{lemma}[theorem]{Lemma}

\newtheorem{corollary}[theorem]{Corollary}

\newtheorem{claim}[theorem]{Claim}
\newtheorem{remark}[theorem]{Remark}
\newtheorem{definition}[theorem]{Definition}

\newtheorem{problem}[theorem]{Problem}

\theoremstyle{definition} 
\newtheorem{example}[theorem]{Example}

\newenvironment{subroutine}[1][]
  {\begin{algorithm}[#1]}
  {\end{algorithm}}

\def\ket#1{{\lvert}#1\rangle}
\def\bra#1{{\langle}#1\rvert}

\def\braket#1#2{{{\langle}#1\vert}#2\rangle}
\def\abs#1{\left| #1 \right|}

\def\norm#1{\left\| #1 \right\|}

\def\w{{\sf w}}
\def\r{{\sf r}}

\RestyleAlgo{boxruled}

\title{Multidimensional Quantum Walks, Recursion,\\ and Quantum Divide \& Conquer}
\author[1,2]{Stacey Jeffery}
\author[2]{Galina Pass}
\affil[1]{CWI, Amsterdam} 
\affil[2]{QuSoft \& University of Amsterdam}

\begin{document}

\maketitle

\begin{abstract}
We introduce an object called a \emph{subspace graph} that formalizes the technique of multidimensional quantum walks. Composing subspace graphs allows one to seamlessly combine quantum and classical reasoning, keeping a classical structure in mind, while abstracting quantum parts into subgraphs with simple boundaries as needed. As an example, we show how to combine a \emph{switching network} with arbitrary quantum subroutines, to compute a composed function. 
As another application, we give a time-efficient implementation of quantum Divide \& Conquer when the sub-problems are combined via a Boolean formula. We use this to quadratically speed up Savitch's algorithm for directed $st$-connectivity.
\end{abstract}

\section{Introduction}\label{sec:intro}

There are a number of graphical ways of reasoning about how the steps or subroutines of a classical algorithm fit together. 
For example, it is natural to think of a (randomized) classical algorithm as a (randomized) decision tree (or branching program), where different paths are chosen depending on the input, as well as random choices made by the algorithm. A deterministic algorithm gives rise to a computation \emph{path}, a randomized algorithm to a computation \emph{tree}. The edges of a path or tree, representing steps of computation, might be implemented by some subroutine that is also realized by a path (or tree) -- we can abstract the subroutine's details by viewing it as an edge, or zoom in and see those details, as convenient. 
More generally, we often think of a classical randomized algorithm as a random walk on a (possibly directed) graph, where there may be multiple parallel paths from point $a$ to point $b$, with the cost of getting from $a$ to $b$ being derived from the expected length of these paths. 

This picture appears to break down for quantum algorithms, at least in the standard circuit model. A quantum circuit can be thought of as a path, with edges representing its steps, but it is not immediately clear how to augment this reasoning with subroutines. Consider calling subroutines with varying time complexities $\{T_i\}_i$ in superposition.
Even if the subroutines are all classical deterministic, in the standard quantum circuit model, we tend to incur a cost of $\max_i T_i$ if we call a superposition of subroutines, since we must wait for the slowest subroutine to finish before we can apply the next step of the computation. This problem was addressed in \cite{jeffery2022subroutines}, where the technique of multidimensional quantum walks~\cite{jeffery2022kDist} was used to show how to get an average in place of a max in several settings where a quantum algorithm calls subroutines in superposition: a general setting, as well as the setting of quantum walks. The intuition behind this work is that a quantum walk does keep the classical intuition of parallel paths representing a superposition of possible computations, and \emph{any} quantum algorithm can be viewed as some sort of a quantum walk on a simple underlying graph (something like a path), but with some additional structure associated with it. 

Multidimensional quantum walks, which we study more formally in this paper as an object than has been done previously, are valuable as a way of combining quantum and classical reasoning. A quantum algorithm can be abstracted as a graph with perhaps complicated internal structure, but a simple \emph{boundary} with an ``in'' and an ``out'' terminal, that can be seamlessly hooked into other graph-like structures, perhaps representing simple classical reasoning, such as a quantum random walk, or perhaps with their own complicated very quantum parts.

\paragraph{Subspace Graphs} While \cite{jeffery2022kDist} and \cite{jeffery2022subroutines} use similar techniques, it is not formally defined what a \emph{multidimensional quantum walk} is. We formally define an object called a \emph{subspace graph} (\defin{multid-QW}) that abstracts the structures in \cite{jeffery2022kDist} and \cite{jeffery2022subroutines}, as well as some other previous algorithmic techniques. A subspace graph is simply a graph with some subspaces associated with each edge and vertex, where the structure of the graph constrains how the spaces can overlap.  
Defining what we mean, precisely, by a multidimensional quantum walk (i.e. subspace graph) is the first step to developing a general theory of recursive constructions of subspace graphs. 

The recursive structure of subspace graphs is useful for composing quantum algorithms, but as a design tool, it is also convenient to be able to view a subspace graph in varying levels of abstraction. We can ``zoom out'' and view a complicated process as just a special ``edge'', or zoom in on that edge and understand its structure as an involved graph with additional structure. 

We cannot hope to be able to understand all quantum algorithms using purely classical ideas -- quantum computing is \emph{not} classical computing. But perhaps the next best thing is a way to seamlessly combine classical and quantum ideas, extending the classical intuition to its limits, and then employing quantum reasoning when needed, but with the possibility of abstracting out from it when needed as well, using a fully quantum form of abstraction.

In this work, we consider one specific kind of composition of subspace graphs called \emph{switch composition} -- another type is implicit in \cite{jeffery2022subroutines} -- but we would like to emphasize the potential for more general types of recursion, which we leave for future work. 

\paragraph{Time-Efficient Quantum Divide \& Conquer} A particular type of recursive algorithm is \emph{divide \& conquer}, in which a problem is broken into multiple smaller sub-problems, whose solutions, obtained by recursive calls, are combined into a solution for the original problem. As a motivating example, consider the recursively defined \emph{nand-tree function}. Let $f_{k,d}:\{0,1\}^{d^k}\rightarrow\{0,1\}$ be defined $f_{0,d}(x)=x$, and for $k\geq 1$,
$$f_{k,d}(x)=1-f_{k-1,d}(x^{(1)})\dots f_{k-1,d}(x^{(d)}),$$ 
where each $x^{(j)}\in\{0,1\}^{d^{k-1}}$, and $x=(x^{(1)},\dots,x^{(d)})$. There is a natural way to break an instance $x$ of $f_{k,d}$ into $d$ sub-problems $x^{(1)},\dots,x^{(d)}$ of $f_{k-1,d}$, and combine the solutions by taking the NAND (negated AND) of the $d$ sub-problem solutions. Grover's algorithm computes this NAND in $O(\sqrt{d})$ queries, so we might hope for a speedup by recursive calls to this quantum algorithm. Unfortunately, since we recurse to depth $k$, the constant in front of $\sqrt{d}$ is raised to the $k$-th power. This kills the quantum speedup completely when $d$ is constant (for example, the most common setting of $d=2$), and that is not even touching on the fact that we would seem to need to amplify the success probability of the subroutine, turning those constants into log factors. On the other hand, it is known \cite{reichardt2009GameTree} that $f_{k,d}$ can be evaluated in $O(\sqrt{d^k})$ quantum queries, even though our attempt to use classical divide-\&-conquer reasoning combined with the basic Grover speedup failed.

More recently, \cite{childs2022Divide} showed how to employ divide-\&-conquer reasoning in the study of quantum query complexity, in which one only counts the number of queries to the input.  They obtained their query upper bounds by composing \emph{dual adversary solutions}. The key to their results is that dual adversary solutions exhibit \emph{perfect} composition: no error, no log factors, not even constant overhead. However, their results were not constructive, as dual adversary solutions do not fully specify algorithms, and in particular, the time complexity analysis of their results was unknown.
In this work, we use the framework of subspace graphs to give a constructive time-complexity version of some of the query complexity results obtained in \cite{childs2022Divide}. In particular, we show (see \thm{strategy_1} for formal statement):
\begin{theorem}[Informal]\label{thm:divide-and-conquer-intro}
Let $\{f_{\ell,n}:D_{\ell,n}\rightarrow\{0,1\}\}_{\ell,n}$ be a family of functions.
Let $\varphi$ be a symmetric Boolean formula on $a$ variables, and suppose $f_{\ell,n} = \varphi(f_{\ell/b,n},\dots,f_{\ell/b,n})\vee f_{aux,\ell,n}$, for some $b>1$ and some auxiliary function $f_{aux,\ell,n}$ with quantum time complexity $T_{\mathrm{aux}}(\ell,n)$. Then the quantum time complexity of $f_{\ell,n}$ is $\widetilde{O}(T(\ell,n))$ for $T(\ell,n)$ satisfying:
$$T(\ell,n) \leq \sqrt{a}T(\ell/b,n)+T_{\mathrm{aux}}(\ell,n).$$
\end{theorem}
Our framework also handles the case where $f_{\ell,n}=\varphi(f_{\ell/b,n},\dots,f_{\ell/b,n},f_{\mathrm{aux},\ell,n})$ for any formula $\varphi$ on $a+1$ variables, but then some extra, somewhat complicated looking costs need to be accounted for, and there is also a scaling in the depth, although this is not an issue if the formula has been preprocessed to be balanced~\cite{bonet1994balancing}.  

Comparing this with the analogous classical statement, which would have $a$ instead of $\sqrt{a}$, we get an up to quadratic speedup over a large class of classical divide-\&-conquer algorithms. As an application, we show a quadratic speedup of Savitch's divide-\&-conquer algorithm for directed $st$-connectivity \cite{savitch1970relationships}. 

To achieve these results, it is essential that we compose subspace graphs, rather than algorithms. When we convert a subspace graph to a quantum algorithm, we get constant factors in the complexity, and these seem necessary without at least specifying what gateset we are working in. By first composing in the more abstract model of subspace graphs, and then only converting to a quantum algorithm at the end, we ensure these factors only come into the complexity once. This is similar to compositions done with other abstract models, such as span programs (of which dual adversary solutions are a special case)~\cite{reichardt2009span}, and transducers~\cite{belovs2023LasVegasTime}. 

\paragraph{Switching Networks} A switching network is a graph with Boolean variables associated with the edges that can switch the edges on or off. Originally used to model certain hardware systems, including automatic telephone exchanges, and industrial control equipment~\cite{shannon1938switchingNetworks}, 
a switching network has an associated function $f$ that is 1 if and only if two special vertices, $s$ and $t$, are connected by a path of ``on'' edges. 
Shannon~\cite{shannon1938switchingNetworks,shannon1949switchingNetworks} showed that series-parallel switching networks are equivalent to Boolean formulas, and 
Lee~\cite{lee1959switchingnetworks} showed that switching networks can model branching programs. These theoretical results have given this model a place in classical complexity theory, where they can be used to study classical space complexity and circuit depth (see~\cite{potechin2015thesis}).

Quantum algorithms for evaluating switching networks\footnote{Switching networks have never been directly referred to in prior work on quantum algorithms, as far as we are aware, but the ``$st$-connectivity problems'' referred to in \cite{jeffery2017stConnFormula} are, in fact, switching networks.} were given in~\cite{jeffery2017stConnFormula}, using a span program construction based heavily on~\cite{belovs2012stConn}. Tight analysis of these span programs for any switching network was completed in~\cite{jarret2018connectivity}.

From~\cite{jeffery2017stConnFormula,jarret2018connectivity}, we get a quantum algorithm for evaluating any switching network when we have query access to the edge variables. Let ${\cal R}_{s,t}(G(x))$ be the effective resistance (see \defin{resistance}) in the subgraph of ``on'' edges whenever $f(x)=1$, and whenever $f(x)=0$, let $F_x\subseteq E$ be the minimum weight $st$-cut-set (see \defin{cut-set}) consisting of only edges that are ``off''.
Then the time complexity of evaluating the switching network is
$$\sqrt{\max_{x\in f^{-1}(1)}{\cal R}_{s,t}(G(x)) \max_{x\in f^{-1}(0)}\sum_{e\in F_x}\w_e},$$
assuming we can implement a certain reflection related to the particular switching network in unit time. From this it follows that if we can query the variable associated with an edge $e$ in time $T_e$, we can evaluate the switching network in time 
$$\sqrt{\max_{x\in f^{-1}(1)}{\cal R}_{s,t}(G(x)) \max_{x\in f^{-1}(0)}\sum_{e\in F_x}\w_e}.\max_{e\in E}T_e.$$
Here we improve this to (see \thm{switching-network-comp}):
$$\sqrt{\max_{x\in f^{-1}(1)}{\cal R}_{s,t}(G(x)) \max_{x\in f^{-1}(0)}\sum_{e\in F_x}\w_eT_e^2}.$$

This is analogous to results of~\cite{jeffery2022subroutines}, which showed a similar statement, but for quantum walks rather than switching networks\footnote{Both this result, and~\cite{jeffery2022subroutines}, include \emph{variable-time quantum search}~\cite{ambainis2010VTSearch} as a special case.}.
It is also similar in flavour to the results of \cite[Chapter 7.2]{cornelissen2023thesis}, where \emph{span programs} are used to evaluate the edge variables of a switching network.
Because switching networks perfectly model Boolean formulas, without the constant overhead we get when we convert to a quantum algorithm, they can be used as a building block for our divide-\&-conquer results. 

\paragraph{Application to DSTCON} Quantum algorithms for $st$-connectivity on \emph{undirected graphs}, which are closely related to evaluating switching networks, are well studied~\cite{durr2004QQueryCompGraph,belovs2012stConn,jeffery2017stConnFormula,jarret2018connectivity,apers2022ustcon}, including quantum algorithms that achieve optimal time- and space-complexity simultaneously, in both the edge-list access model, and the adjacency matrix access model.\footnote{We do not make a distinction between various access models, because they can simulate one another in poly$(n)$ time and $\log(n)$ space, so our result, which includes a $2^{O(\log n)}$ term, is the same in all models.}
In contrast, quantum algorithms for \emph{directed} $st$-connectivity (\textsc{dstcon}) -- the problem of deciding if there is a directed path from $s$ to $t$ in a directed graph -- is less well understood. 
The algorithm of~\cite{durr2004QQueryCompGraph} also applies to directed graphs, deciding connectivity in $\widetilde{O}(n)$ time and space\footnote{In the edge-list access model.}. This algorithm has optimal time complexity, whereas its space complexity is far from optimal.

Directed $st$-connectivity, also called \emph{reachability}, is a fundamental problem in classical space complexity.
In particular, understanding if this problem can be solved in  $\log(n)$ space by a quantum algorithm would resolve the relationship between quantum logspace complexity and ${\sf NL}$, as $\textsc{dstcon}$ is ${\sf NL}$-complete. 

The best known classical (deterministic) space complexity of $\textsc{dstcon}$ is $O(\log^2(n))$, using Savitch's algorithm. We apply quantum divide \& conquer (\thm{divide-and-conquer-intro}) to give a quadratic speedup to Savitch's algorithm, achieving $2^{\frac{1}{2}\log^2(n)+O(\log(n))}$ time, while still maintaining $O(\log^2(n))$ space (see \thm{Savitch_quantum}).

\paragraph{Model of Computation} We work in the same model as \cite{jeffery2022kDist}, where we allow not only arbitrary quantum gates, but also assume subroutines are given via access to a unitary that applies the subroutine's $t$-th gate controlled on the value $t$ in some time register. This is possible, for example, with quantum random access gates. See~\sec{model} for details. 

\paragraph{Related and Future Work} While writing this manuscript, we became aware of an independent work that also achieves time-efficient quantum divide \& conquer~\cite{allcock2023divide} when either (1) $\varphi$ is an OR (equivalently, an AND) or (2) $\varphi$ is a minimum or maximum. OR is a Boolean formula, while minimum/maximum is not. In that sense, our results are incomparable. 
The framework of~\cite{allcock2023divide} also differs from our work in that they explicitly treat the complexity of computing sub-instances (the cost of the ``create'' step), whereas we assume sub-instances can simply be queried in unit cost. In one of the applications of~\cite{allcock2023divide}, this cost is not negligible, and is even the dominating cost, so our framework, as stated, would not handle this application. This is not an inherent limitation of our techniques -- it would be possible to take this cost into account in our framework as well.
Ref.~\cite{allcock2023divide} applies their framework to problems that are distinct from our applications.

We also mention that Ref.~\cite{childs2022Divide} analyzes the quantum query complexity of divide \& conquer where $\varphi$ is an arbitrary Boolean formula, as well as in settings where the function combining the sub-problems is more general.
While our techniques do apply to composing quantum algorithms for arbitrary functions (already studied in \cite{jeffery2022subroutines}), an issue is a poor scaling in the error of subroutines. If we start with a bounded-error quantum algorithm for some function, we need to amplify the success probability, as it will be called many times, incurring logarithmic factors. This becomes a serious problem if the function is called recursively to depth more than constant.
We get around this in the case of Boolean formulas by using a switching network construction (from which a quantum algorithm could be derived) rather than a bounded-error quantum algorithm for evaluating the formula. Our techniques would thus also readily apply to functions for which there is an efficient quantum algorithm derived from a switching network.
For more general functions, the solution might lie in a recent framework, \emph{transducers}~\cite{belovs2023LasVegasTime}, that allows for the composition of quantum algorithms without the need for success probability amplification. We leave further investigation for future work.

We can compare the model of transducers~\cite{belovs2023LasVegasTime} to the framework of subspace graphs introduced here. 
These models are related, in that they are both 
good
for analyzing the time complexity of composed quantum algorithms. 
A subspace graph gives rise to a unitary that is the product of two reflections, and this is a transducer, as discussed in \cite[Section 3.2]{belovs2023LasVegasTime}. Compared with the model of subspace graphs, transducers are more general, and cleaner, in particular how they are implemented as quantum algorithms. 
It is not known whether subspace graphs can reproduce the composition results of transducers in terms of error scaling. 
However, we feel that subspace graphs, with their additional structure, are still useful for designing and reasoning about quantum algorithms, especially those in which there is some classical intuition to hold onto. It seems likely that, once designed in the framework of subspace graphs, an algorithmic construction could be converted to the transducer framework to achieve any advantage that is missing for subspace graphs. We leave concrete exploration of the connection for future work. 

Aside from the directions already mentioned, the most interesting future direction is to study more general compositions of subspace graphs, and see to what extent we can employ classical reasoning to quantum algorithms, and to what extent this (hopefully gracefully) breaks down.

\section{Preliminaries}\label{sec:prelim}

\subsection{Model of Computation}\label{sec:model}

We will be talking about combining quantum subroutines, and here we describe the model of access for the subroutines that we will assume. 
A quantum algorithm or subroutine is a sequence of unitaries $U_1,\dots,U_T$ acting on some common space $H$. 
Given a set of subroutines $\{(U_1^i,\dots,U_{T_i}^i)\}_{i\in {\cal I}}$, all acting on a space $H$, but possibly running in different times, $T_i\leq T_{\max}$, we will assume the operator $\sum_{i\in {\cal I}}\ket{i}\bra{i}\otimes\sum_{t=1}^{{T}_{\max}}\ket{t}\bra{t}\otimes U_t^i$ can be implemented in ``unit cost''. We use ``unit cost'' to describe a cost we are willing to accept as a multiplicative factor on all complexities (for example, $O(1)$  or polylogarithmic in some natural variable). This assumption is only reasonable if the unitaries each have unit cost. It does not hold for strict gate complexity when $U_1^i,\dots,U_{{T}_{\max}}^i$ are arbitrary gates, but it does hold in the \emph{fully quantum} QRAM model if $U_1^i,\dots,U_{{T}_{\max}}^i$ are stored as a list of gates in classical memory to which a quantum computer has read-only superposition access (sometimes referred to as ``QRAM'' in previous literature), or 
if they satisfy certain uniformity conditions (see~\cite[Section~2.2]{jeffery2022kDist}).

\subsection{Graph Theory}

In this paper, a \emph{graph} will be an undirected graph, and we will only consider directed graphs in \sec{dstcon}. A graph $G=(V,E)$ has vertex set $V=V(G)$ and edge set $E=E(G)$. For $u\in V$, we let $E(u)$ be the subset of edges of $E$ that are incident to $u$. We assume for convenience that each edge has an orientation, and $E^\rightarrow(u)\subseteq E(u)$ are the edges oriented outwards from $u$, and $E^\leftarrow(u) = E(u)\setminus E^{\rightarrow}(u)$ those oriented inwards. Then for any $e\in E$, there is a unique $u\in V$ such that $e\in E^\rightarrow(u)$, and a unique $v\in V$ such that $e\in E^\leftarrow(v)$, and we can write $e=(u,v)$. 
We can associate non-negative \emph{weights} $\{\w_e\}_{e\in E}$ to the edges of $G$. 

\begin{definition}\label{def:resistance}
    A \emph{flow} on a weighted graph $G$ is a real-valued function $\theta$ on $E$. Define 
    $$\theta(u)\colonequals\sum_{e\in E^{\rightarrow}(u)}\theta(e)-\sum_{e\in E^\leftarrow(u)}\theta(e).$$
    We call $\theta$ a \emph{unit $st$-flow} if $\theta(s)=-\theta(t)=1$, and for all $u\in V\setminus\{s,t\}$, $\theta(u)=0$. The \emph{effective resistance} is defined
    $${\cal R}_{s,t}(G)\colonequals\min_{\mbox{unit flows }\theta}\sum_{e\in E}\frac{\theta(e)^2}{\w_e}.$$
\end{definition}

\begin{definition}\label{def:cut-set}
    An \emph{$st$-cut-set} is a set $F\subseteq E$ whose removal from $G$ would leave $s$ and $t$ in different connected components.
\end{definition}

\subsection{Function Composition and Boolean Formulas}\label{sec:formulas}

A Boolean formula, $\varphi$, on $N$ variables is a rooted tree with $N$ leaves, where each internal vertex is labelled by either $\vee$ (OR) or $\wedge$ (AND), and each leaf is labelled by a unique variable or its negation. The depth of a formula is the length of the longest path from the root to a leaf. We assume that for each internal vertex with $d$ children, the outgoing edges are labelled by some $d$-element set, without loss of generality, $[d]$, so that every vertex can be labelled by a string of the labels of its path from the root. That is: the root is labelled by the empty string, denoted $\emptyset$, and any other vertex is labelled by the label of its parent, concatenated with the label of the edge from its parent to it. We call the set of all strings labelling leaves $\Sigma$, so $|\Sigma|=N$. We let $\Sigma_i$ denote the set of strings in $\Sigma$ whose first letter is $i$. We let $\overline{\Sigma}$ denote the set of labels of all vertices, so it is the set of prefixes of strings in $\Sigma$. 

We label the variable associated with a leaf $\sigma\in \Sigma$ by $\sigma$ as well. A fixed setting of the variables $x\in\{0,1\}^\Sigma$ induces a value at each vertex of $\varphi$, as follows. If a leaf $\sigma$ is labelled by an unnegated variable, then its value is $x_\sigma$, and otherwise its value is $\neg x_\sigma$. If a vertex is not a leaf, if it is labelled by $\vee$, then its value is the OR of all the values of its children, whereas if it is labelled by $\wedge$, its value is the AND of all the values of its children. We let $\varphi(x)$ denote the value of the root of the tree. 

\begin{definition}\label{def:balanced}
We say a formula is \emph{symmetric} if for every internal node, the sub-trees of its children are identical aside from their leaves. 
\end{definition}

Note that symmetric formulas do not necessarily compute symmetric functions -- functions that only depend on the Hamming weight of the input.

\begin{definition}\label{def:balanced2}
    We say a (family of) formulas is balanced if there is a constant $c$ such that for every node, if its subtree has $N$ leaves, and it has $d$ children, then the sub-tree of each child has at most $cN/d$ leaves. 
\end{definition}

For any $f:\{0,1\}^\Sigma\rightarrow\{0,1\}$ and $\{f_\sigma:\{0,1\}^{m_\sigma}\rightarrow\{0,1\}\}_{\sigma\in \Sigma}$, we define the composed function $f\circ (f_\sigma)_{\sigma\in \Sigma}:\{0,1\}^{\sum_{\sigma\in\Sigma}m_\sigma}\rightarrow\{0,1\}$ by 
$$f\circ (f_\sigma)_{\sigma\in\Sigma} (x) = f(f_\sigma(x^\sigma))_{\sigma\in\Sigma}$$
where $x=(x^\sigma)_{\sigma\in\Sigma}$ is the $\sum_\sigma m_\sigma$-bit string obtained by concatenating the strings $\{x^\sigma\}_\sigma$, and $(f_\sigma(x^\sigma))_{\sigma\in\Sigma}$ is the $|\Sigma|$-bit string whose $\sigma$-th bit is $f_\sigma(x^\sigma)$. 
We can also write $f\circ (f_\sigma)_{\sigma\in F}$ for some $F\subset\Sigma$, in which case, we implicitly assume that $f_\sigma$ is the identity on $\{0,1\}$ for all $\sigma\in\Sigma\setminus F$.
For a Boolean formula $\varphi$ on $\Sigma$, we will also write $\varphi\circ (f_\sigma)_{\sigma\in \Sigma}$ as a short-hand for $f_\varphi\circ (f_\sigma)_{\sigma\in \Sigma}$, where $f_\varphi$ is the function computed by the formula. 

\subsection{Phase Estimation Algorithms}\label{sec:phase-estimation}

The technique of phase estimation dates back to~\cite{kitaev1996PhaseEst}, but we will use a specific application of this technique precisely defined in~\cite{jeffery2022kDist}, following a blueprint that had already been used many times in the quantum algorithms literature. 

\begin{definition}[Parameters of a Phase Estimation Algorithm]\label{def:phase-est-alg}
For an implicit input $x\in\{0,1\}^*$, fix a finite-dimensional complex inner product space $H$, a unit vector $\ket{\psi_0}\in H$, and sets of vectors
$\Psi^{\cal A},\Psi^{\cal B}\subset H$. We further assume that $\ket{\psi_0}$ is orthogonal to every vector in $\Psi^{\cal B}$.  Let $\Pi_{\cal A}$ be the orthogonal projector onto ${\cal A}=\mathrm{span}\{\Psi^{\cal A}\}$, and similarly for $\Pi_{\cal B}$.
\end{definition}
Let 
$U_{\cal AB}=(2\Pi_{\cal A}-I)(2\Pi_{\cal B}-I)$.
The algorithm defined by $(H,\ket{\psi_0},\Psi^{\cal A},\Psi^{\cal B})$ performs phase estimation of $U_{\cal AB}$ on initial state $\ket{\psi_0}$, to sufficient precision that by measuring the phase register and checking if the output is 0, we can distinguish between a \emph{negative case} and a \emph{positive case}.

\begin{definition}[Negative Witness]\label{def:neg-witness}
A \emph{negative witness} for $(H,\ket{\psi_0},\Psi^{\cal A},\Psi^{\cal B})$ is a vector $\ket{w_{\cal A}}\in {\cal A}$ such that $\ket{w_{\cal B}}\colonequals \ket{\psi_0}-\ket{w_{\cal A}}\in {\cal B}$. 
\end{definition}

\begin{definition}[Positive Witness]\label{def:pos-witness}
A \emph{positive witness} for $(H,\ket{\psi_0},\Psi^{\cal A},\Psi^{\cal B})$ is a vector $\ket{w}\in {\cal A}^{\bot}\cap {\cal B}^\bot$ such that $\braket{\psi_0}{w}\neq 0$.
\end{definition}

Note that a negative witness exists if and only if $\ket{\psi_0}\in {\cal A}+{\cal B}$, whereas a positive witness exists if and only if $\ket{\psi_0}$ has a non-zero component in $({\cal A}+{\cal B})^\bot = {\cal A}^\bot\cap {\cal B}^\bot$, which is if and only if $\ket{\psi_0}\not\in {\cal A}+{\cal B}$. The following theorem describes an algorithm for distinguishing these two cases.

\begin{theorem}[\cite{jeffery2022kDist}]\label{thm:phase-est-fwk}
Fix $(H,\ket{\psi_0},\Psi^{\cal A},\Psi^{\cal B})$ as in \defin{phase-est-alg}.
Suppose we can generate the state $\ket{\psi_0}$ in cost ${\sf S}$, and implement $U_{\cal AB}=(2\Pi_{\cal A}-I)(2\Pi_{\cal B}-I)$ in cost ${\sf A}$.

\noindent Let $c_+\in [1,50]$ be some constant, and let ${\cal C}_-\geq 1$ be a positive real number that may scale with $|x|$, such that we are guaranteed that one of the following holds:
\begin{description}
\item[Positive Condition:] There is a positive witness $\ket{w}$ s.t.~$\frac{|\braket{w}{\psi_0}|^2}{\norm{\ket{w}}^2}\geq \frac{1}{c_+}$. 
\item[Negative Condition:] There is a negative witness $\ket{w_{\cal A}}$ s.t.~$\norm{\ket{w_{\cal A}}}^2 \leq {\cal C}_-$.
\end{description}
Then there is a quantum algorithm that distinguishes these two cases with bounded error in time
$$O\left({\sf S}+\sqrt{{\cal C}_-}{\sf A}\right)$$
and space $O(\log\dim H + \log {\cal C}_-)$.
\end{theorem}

The reason we assume that ${\cal A}$ and ${\cal B}$ are given by bases is that this extra structure allows us to implement the reflections around these spaces, which we describe shortly. For this reason, we refer to these given bases as \emph{working bases}.

\begin{definition}[Working Basis Generation]\label{def:working-basis}
    We say an orthonormal basis $\Psi=\{\ket{b_{\ell}}\}_{\ell\in L}\subset H_G$ can be generated in time $T$ if 
    \begin{enumerate}
        \item The reflection around the subspace $\mathrm{span}\{\ket{\ell}:\ell\in L\}$ of $H_G$ can be implemented in time $T$.
        \item There is a map that acts as $\ket{\ell}\mapsto\ket{b_{\ell}}$ for all $\ell\in L$ that can be implemented in time $T$.  
    \end{enumerate}
\end{definition}

\begin{claim}\label{clm:basis-implement-U}
    If $\Psi_{\cal A}$ and $\Psi_{\cal B}$ are working bases for ${\cal A}$ and ${\cal B}$, respectively, that can each be generated in time $T$, then $U=(2\Pi_{\cal A}-I)(2\Pi_{\cal B}-I)$ can be implemented in time $T$. 
\end{claim}
\begin{proof}
    Run the map $\ket{\ell}\mapsto\ket{b_{\ell}}$ in reverse, reflect around $\mathrm{span}\{\ket{\ell}:ell\in L\}$, and then run the map $\ket{\ell}\mapsto\ket{b_{\ell}}$. 
\end{proof}

\noindent We end by remarking that it is enough to have a working basis for the complement of ${\cal A}$ or ${\cal B}$.

\begin{corollary}\label{cor:dual-basis-gen}
If $\Psi$ is a working basis that can each be generated in time $T$, then there is a basis $\Psi'$ for $\mathrm{span}\{\Psi\}^\bot$ that can be generated in time $T$.    
\end{corollary}
\begin{proof}
Suppose $\Psi=\{\ket{b_{\ell}}\}_{\ell\in L}$, and $H_G\equiv\mathrm{span}\{\ket{\ell}:\ell\in Z\}$ where $L\subseteq Z$.
Let $U_\Psi$ be the given map that acts as $\ket{\ell}\mapsto \ket{b_\ell}$ for $\ell\in L$, in time $T$. Then define
$\Psi'=\{\ket{b_\ell'}=U_\Psi\ket{\ell}\}_{\ell\in Z\setminus L}$. That is, $\Psi'$ is generated by $U_\Psi$. Since both bases are generated by the same map, their generation has the same time complexity.
\end{proof}

\section{Subspace Graphs for Multidimensional Quantum Walks}

\subsection{Subspace Graphs}

Multidimensional quantum walks were introduced as such in \cite{jeffery2022kDist}, although they generalize various quantum algorithms that have appeared previously, including \cite{szegedy2004QMarkovChainSearch,reichardt2011spanformulas,belovs2013ElectricWalks,belovs2013TimeEfficientQW3Distintness,ito2015approxSpan}.
We wish to consider very general kinds of composition of multidimensional quantum walks, and in order to be clear about the precise types of objects we are composing, we give a more formal definition than has appeared previously.
\begin{definition}[Subspace Graph]\label{def:multid-QW}
A \emph{subspace graph} consists of a (undirected) graph $G=(V,E)$, a boundary $B\subseteq V$, and the following subspaces of a space $H=H_{G}$:
\begin{description}
    \item[Edge and Boundary Spaces] We assume $H$ can be decomposed into a direct sum of spaces as follows: $H=\bigoplus_{e\in E}\Xi_e \oplus \bigoplus_{u\in B}\Xi_u$. 
    \item[Edge and Boundary Subspaces] For each $e\in E\cup B$, let $\Xi_e^{\cal A}$ and $\Xi_e^{\cal B}$ be subspaces of $\Xi_e$. These need not be orthogonal, and they may each be $\{0\}$, all of $\Xi_e$, or something in between.
    \item[Vertex Spaces and Boundary Space] For each $u\in V$, let ${\cal V}_u\subseteq \bigoplus_{e\in E(u)}\Xi_e$ be pairwise orthogonal spaces. Let ${\cal V}_B\subseteq\bigoplus_{u\in B}\Xi_u$. 
\end{description}
Then we define 
$${\cal A}_G=\bigoplus_{e\in E\cup B}\Xi_e^{\cal A}
\mbox{ and }
{\cal B}_G=\bigoplus_{u\in V}{\cal V}_u+ {\cal V}_B+\bigoplus_{e\in E\cup B}\Xi_e^{\cal B}.$$
\end{definition}
We will refer to a subspace graph as simply $G$, with the associated spaces implicit. 
Let us give some intuition for the meaning of the graph structure in a subspace graph. We will shortly see (in \sec{subspace-phase-est}) that there is a quantum algorithm associated with a subspace graph, which alternatively applies reflections around ${\cal A}_G$ and ${\cal B}_G$ to some initial state $\ket{\psi_0}$ to distinguish between the cases:
\begin{description}
    \item[Negative Case:] $\ket{\psi_0}\in {\cal A}_G+{\cal B}_G$
    \item[Positive Case:] $\ket{\psi_0}$ has a (large) component, called a \emph{positive witness} in $({\cal A}_G+{\cal B}_G)^\bot = {\cal A}_G^\bot\cap{\cal B}_G^\bot$
\end{description}
Then we can see all the vectors in ${\cal A}_G$ and ${\cal B}_G$ as linear constraints on the initial state -- it must have a large component orthogonal to all of them in the positive case. 
${\cal A}_G$ and ${\cal B}_G$ are each decomposed into sums of subspaces, representing subsets of constraints. The graph structure of $G$ restricts how the different sets of constraints are allowed to overlap -- the constraints associated with vertex $v$ only overlap constraints associated with edges that are incident to $v$. This graph structure can then be used to help analyse the algorithm. For example, a positive witness is related to a \emph{flow} on $G$.

\begin{remark}
In \cite{jeffery2022kDist} and \cite{jeffery2022subroutines}, spaces are only defined for each vertex, and the spaces associated with a pair of vertices must be orthogonal unless $u$ and $v$ are adjacent in $G$
(in the words of \cite{jeffery2022subroutines}, $G$ must be an \emph{overlap graph} of the spaces).
Here we have found it convenient to assume that $G$ has been augmented by putting a vertex in the middle of each edge, which is why we have a space for each vertex of $G$, and for each edge of $G$. This ensures, for example, that the graph is bipartite, which is required in \cite{jeffery2022kDist} and \cite{jeffery2022subroutines}. 
However, the important intuition about subspace graphs is that the locality structure of the graph constrains how the spaces are allowed to overlap.
\end{remark}

We fix some notation for talking about subspace graphs. For any $e\in E\cup B$, we let $\Pi_e$, $\Pi_e^{\cal A}$ and $\Pi_e^{\cal B}$ denote orthogonal projectors onto $\Xi_e$, $\Xi_e^{\cal A}$, and $\Xi_e^{\cal B}$, respectively. For any set $F\subseteq E\cup B$, we let $\Xi_F=\bigoplus_{e\in F}\Xi_e$, $\Xi_F^{\cal A}=\bigoplus_{e\in F}\Xi_e^{\cal A}$, and similarly for $\Xi_F^{\cal B}$. We let $\Pi_F$, $\Pi_F^{\cal A}$ and $\Pi_F^{\cal B}$ denote the orthogonal projectors onto $\Xi_F$, $\Xi_F^{\cal A}$ and $\Xi_F^{\cal B}$, respectively.

We will allow subspace graphs to implicitly depend on some input, and associate them with a computation. Specifically, if we fix some initial state $\ket{\psi_0}\in H_G$, then this, along with ${\cal A}_G$ and ${\cal B}_G$, define a phase estimation algorithm (\sec{phase-estimation}) that decides if $\ket{\psi_0}\in {\cal A}_G+{\cal B}_G$. With this in mind, we say a subspace graph that depends on some input $x$, \emph{computes a function $f$} (with respect to $\ket{\psi_0}$, which should be clear from context) if $f(x)=0$ if and only if $\ket{\psi_0}\in {\cal A}_G+{\cal B}_G$. 
We make this formal in \sec{subspace-phase-est}.

One motivating example from which we derive intuition is the special case of a quantum walk on a graph. We will see other examples in later sections, and our results will mainly be based on these, but we describe how quantum walks fit into this picture as an illustration. First, we define a special case for ${\cal V}_u$. While the definition of  subspace graph is independent of any weights of the graph $G$, we can always assume $G$ has associated edge weights from which we derive extra structure. 
\begin{definition}[Simple Vertex]\label{def:simple-vertex}
    Fix a subspace graph $G$ with associated edge weights $\{\w_e\}_{e\in E}$. For $u\in V$, define
    $$\ket{\psi_\star(u)}\colonequals \sum_{e\in E^\rightarrow(u)}\sqrt{\w_e}\ket{\rightarrow,e}+\sum_{e\in E^\leftarrow(u)}\sqrt{\w_e}\ket{\leftarrow,e}.$$
    A vertex $u\in V$ is \emph{simple} if for all $e\in E^\rightarrow(u)$, $\ket{\rightarrow,e}\in\Xi_e$, for all $e\in E^\leftarrow(u)$, $\ket{\leftarrow,e}\in \Xi_e$, and if $u\in V\setminus B$
    ${\cal V}_u =  \mathrm{span}\{\ket{\psi_\star(u)}\};$
    and if $u\in B$, either $\ket{\rightarrow,u}\in \Xi_u$ and 
    ${\cal V}_u =  \mathrm{span}\{\ket{\rightarrow,u}+\ket{\psi_\star(u)}\}$
    or $\ket{\leftarrow,u}\in \Xi_u$ and 
    ${\cal V}_u =  \mathrm{span}\{\ket{\leftarrow,u}+\ket{\psi_\star(u)}\}.$
\end{definition}

In our applications, we generally have two degrees of freedom associated with each edge, which we can think of as having subdivided each edge into two edges (each associated with a single degree of freedom). When we subdivide edge $e$ by putting a vertex in the middle of it - call that vertex $e$ - it now has two incident edges going into it: $\ket{\rightarrow,e}$ and $\ket{\leftarrow,e}$. To make the simple vertex definition agree with the orientation on the edges, we let a star state $\ket{\psi_\star(u)}$ overlap $\ket{\rightarrow,e}$ if $e$ is an outgoing edge of $u$, and $\ket{\leftarrow,e}$ if $e$ is an incoming edge of $u$.

\noindent To motivate this definition, consider the following example. 

\begin{example}[Quantum Walk on a Graph]\label{ex:multiQW-QW}
A subspace graph $G$, with associated edge weights, in which 
\begin{enumerate}
    \item all vertices are simple, 
    \item $\Xi_{E\cup B}^{\cal B}=\{0\}$, 
    \item for all $e\in E$, $\Xi_e^{\cal A} = \mathrm{span}\{\ket{\rightarrow,e}+\ket{\leftarrow,e}\}$,
    \item $B=\{s\}\sqcup M$, where for all $u\in M$, $\Xi_u^{\cal A}+\Xi_u^{\cal B}=\{0\}$,
\end{enumerate}
is a \emph{quantum walk} on the weighted graph $G=(V,E)$.
To motivate this terminology, the phase estimation algorithm in the next section implements a quantum walk on $G$, in the sense that it decides if $\exists t\in M$ in the same component as $s$\footnote{This is if we use $\ket{\psi_0}=\ket{s}$. More generally, we can let $B=S\sqcup M$ and let $\ket{\psi_0}=\sum_{s\in S}\sqrt{\sigma(s)}\ket{s}$ for \emph{initial distribution} $\sigma$ on $S$.}. 

Simple vertices are motivated by quantum walks, and related constructions, as follows. In a random walk on a weighted graph, in order to take a step from a vertex $u\in V$ to a neighbour, the random walker chooses an edge $e\in E(u)$ to traverse, with probability $\w_e/(\sum_{e'\in E(u)}\w_{e'})$. This is precisely the distribution obtained by measuring $\ket{\psi_\star(u)}/\norm{\ket{\psi_\star(u)}}$. The states $\ket{\psi_\star(u)}$ correspond to \emph{quantum walk states} -- alternating a reflection around these with a reflection around the states $\{\ket{\rightarrow,e}+\ket{\leftarrow,e}\}_{e\in E}$ implements a discrete-time quantum walk, as in \cite{szegedy2004QMarkovChainSearch,magniez2006SearchQuantumWalk,belovs2013ElectricWalks}.\footnote{In some previous works, it is assumed that the graph is bipartite, and then the walk is implemented by alternating reflections around the states $\ket{\psi_\star(u)}$ of the two parts of the bipartition. We are actually doing the same thing here, we have just ensured the graph is bipartite by inserting a vertex in the middle of each edge.} 
\end{example}

A related example is that of a \emph{switching network}, which, unlike quantum walks, will actually be one of the building blocks of our results in this paper. We first define what it means for an edge to be a switch.

\begin{definition}[Switch Edge]\label{def:switch-edge}
Fix a subspace graph $G$. We call an edge $e\in E$ a \emph{switch} (or \emph{switch edge}) if there is some value $\varphi(e)\in\{0,1\}$ associated with that edge (implicitly depending on the input), such that
\begin{align*}
\Xi_e &= \mathrm{span}\{\ket{\rightarrow,e},\ket{\leftarrow,e}\},\\
\Xi_e^{{\cal A}} &=\mathrm{span}\{\ket{\rightarrow,e}-(-1)^{\varphi(e)}\ket{\leftarrow,e}\}\quad\mbox{and}\quad
\Xi_e^{{\cal B}} =\mathrm{span}\{\ket{\rightarrow,e}+\ket{\leftarrow,e}\},
\end{align*}
and moreover, if $e\in E^\rightarrow(u)\cap E^\leftarrow(v)$, (that is, $e=(u,v)$), then
 ${\cal V}_u\cap \Xi_e = \mathrm{span}\{\ket{\rightarrow,e}\}$ and 
 ${\cal V}_v\cap \Xi_e = \mathrm{span}\{\ket{\leftarrow,e}\}$.
\end{definition}

The idea behind a switch edge is that if $\varphi(e)=0$, $\Xi_e^{\cal A}+\Xi_e^{\cal B}=\Xi_e$, and so $\Xi_e^\bot=\{0\}$. Recall that a \emph{positive witness} is some $\ket{w}\in {\cal A}_G^\bot\cap {\cal B}_G^\bot$ (\defin{pos-witness}). If $\Xi_e^{\bot}=\{0\}$, then $\ket{w}$ can have no overlap with $\Xi_e$, so $e$ is essentially blocked from use, so we say the edge is \emph{switched off}. In switching networks, defined shortly, the positive witness is an $st$-flow, and it is restricted to edges in the subgraph $G(x)$ of edges that are switched on (we discuss this more in \sec{switching-networks} -- in particular, see \eq{flow-witness}).

Next, we define some conditions that are satisfied by switching networks, as well as other examples considered in this paper.

\begin{definition}[Canonical $st$-boundary]\label{def:canonical-boundary}
We say a subspace graph $G$ has \emph{canonical $st$-boundary} if $B=\{s,t\}$, where:
\begin{itemize}
    \item $\Xi_s=\mathrm{span}\{\ket{s},\ket{\leftarrow,s}\}$, $\Xi_s^{\cal A}= \mathrm{span}\{\ket{s}+\ket{\leftarrow,s}\}$
     and $\Xi_s^{\cal B}=\{0\}$,
    where ${\cal V}_s$ only overlaps $\ket{\leftarrow,s}$.
    \item $\Xi_t=\mathrm{span}\{\ket{\rightarrow,t},\ket{t}\}$,  $\Xi_t^{\cal A}=\mathrm{span}\{\ket{\rightarrow,t}+\ket{t}\}$
     and $\Xi_t^{\cal B}=\{0\}$,
    where ${\cal V}_t$ only overlaps $\ket{\rightarrow,t}$.
\item $\ket{\leftarrow,s}+\ket{\rightarrow,t}\in {\cal B}_G$.
\end{itemize}
\end{definition}
\noindent Note that when $G$ has canonical $st$-boundary, we always have $\ket{s}+\ket{t}\in {\cal A}_G+{\cal B}_G$. 

As discussed in \sec{intro}, a switching network is actually a classical object, but we overload the term to describe a specific kind of subspace graph. This is justified by the fact that, evaluating such a subspace graph is equivalent to evaluating a switching network (see the discussion in \sec{switching-networks}).
\begin{definition}[Switching Network]\label{def:switching-network}
    A \emph{switching network} is a subspace graph with canonical $st$-boundary in which all edges are switches (\defin{switch-edge}), and all vertices are simple (\defin{simple-vertex}).
\end{definition}
Note that if a subspace graph is a switching network, then all its spaces are fixed by the graph structure $(V,E)$, and the associated weights. This is not true for subspace graphs more generally, and in \sec{alg-subspace-graph}, we will see an example of a subspace graph with other kinds of subspaces.

We discuss more properties of switching networks, including several examples, in \sec{switching-networks}. The examples we will detail in this paper will not be restricted to switching networks, but all, aside from some pedagogical examples, will have canonical boundary, and we will let some of the edges be switches.

\subsection{Phase Estimation Algorithm}\label{sec:subspace-phase-est}

A subspace graph, as in \defin{multid-QW}, gives rise to a phase estimation algorithm (see \sec{phase-estimation}) with spaces ${\cal A}_G$ and ${\cal B}_G$ as in \defin{multid-QW} and some initial state $\ket{\psi_0}\in\Xi_B\cap {\cal B}_G^\bot$. Though it is possible to consider more general states $\ket{\psi_0}$, throughout this paper, we will always assume $G$ has canonical $st$-boundary, and will let
\begin{equation}
    \ket{\psi_0}=\frac{1}{\sqrt{2}}(\ket{s}-\ket{t}).
\end{equation}

To fully specify a phase estimation algorithm, we need a pair of orthogonal bases for ${\cal A}_G$ and ${\cal B}_G$ that are easily generated, in order to implement their respective reflections. We therefore always assume we have such a basis pair associated with $G$, which we call the \emph{working bases}. For example, if $\Xi_e^{\cal B}=\{0\}$ for all $e$, then it suffices to have good bases for all the pairwise orthogonal spaces ${\cal V}_v$; and good bases for all the pairwise orthogonal spaces $\Xi_e^{\cal A}$. However, it is in our interest to put as much of $\Xi_e^{\cal A}+\Xi_e^{\cal B}$ into $\Xi_e^{\cal B}$ as we can manage while still maintaining a good working basis for ${\cal B}_G$. That is because of how the negative witness complexity is defined (\defin{neg-witness}, and also \defin{neg-graph-wit} below): only the part of $\ket{\psi_0}$ in ${\cal A}_G$ is counted in the complexity. We will see this in action in switching networks, in \sec{switching-networks}.

To facilitate our later constructions, we will make the following extra assumptions on the working bases.

\begin{definition}[Composable Bases]\label{def:composable-basis}
    Fix a subspace graph with canonical $st$-boundary, and some $\r\in\mathbb{R}_{>0}$. We say a pair of working bases $\Psi_{\cal A}$ and $\Psi_{\cal B}$ (see \defin{working-basis}) for ${\cal A}_G$ and ${\cal B}_G$ are $st$-composable, with scaling factor $\r$, if they satisfy the following assumptions: 
    \begin{enumerate}
        \item $\Psi_{{\cal A}}=\bigcup_{e\in E\cup B}\Psi_{{\cal A}}(e)$, where $\Psi_{{\cal A}}(e)$ is an orthonormal basis for $\Xi_e^{\cal A}$;
        \item $\Psi_{\cal B}=\Psi_{\cal B}^-\cup \bigcup_{e\in E\cup B}\Psi_{{\cal B}}(e)$, for some $\Psi_{\cal B}^-$, where $\Psi_{{\cal B}}(e)$ is an orthonormal basis for $\Xi_e^{\cal B}$; 
        \item $\ket{b_0}\in\Psi_{\cal B}^-$ satisfies $\ket{b_0}=\frac{1}{\sqrt{2}}(\ket{\leftarrow,s}+\ket{\rightarrow,t})$;
        \item $\ket{b_1}\in\Psi_{\cal B}^-$ satisfies $\ket{b_1}=\frac{\ket{\leftarrow,s}-\ket{\rightarrow,t}+\sqrt{\r}\ket{\bar{b}_1}}{\sqrt{2+\r}}$ for some unit vector $\ket{\bar{b}_1}$ that is orthogonal to $\ket{\leftarrow,s}$ and $\ket{\rightarrow,t}$;
        \item $\Psi_{\cal B}\setminus\{\ket{b_0},\ket{b_1}\}$ is orthogonal to $\ket{\leftarrow,s}$ and $\ket{\rightarrow,t}$. 
    \end{enumerate}
\end{definition}
\noindent We note that generally we will be able to choose the scaling factor $\r$ by applying a linear map to $H_G$ that scales $\Xi_E$ relative to $\Xi_B$. See also the remark after \thm{subspace-graph-to-alg}.

Now we can define the special form of subspace graph that will apply to all subspace graphs studied in this paper (aside from some examples mentioned informally):
\begin{definition}\label{def:st-composable-subspace-graph}
    We say a subspace graph $G$ is \emph{$st$-composable} if:
    \begin{enumerate}
        \item $G$ has canonical $st$-boundary (\defin{canonical-boundary});
        \item $G$ is equipped with composable working bases (\defin{composable-basis});
        \item for each $e\in E\setminus\overline{E}$, $\Xi_e^{\cal B}=\{0\}$, where $\overline{E}$ is the set of edges that are switches.
    \end{enumerate}
\end{definition}

To analyze a phase estimation algorithm, we need to exhibit positive and negative witnesses. 
A positive witness for a phase estimation algorithm (\defin{pos-witness}) is a vector $\ket{w}\in {\cal A}_G^\bot\cap {\cal B}^\bot_G$ such that $\braket{\psi_0}{w}\neq 0$. By scaling $\ket{w}$, we can make $\braket{\psi_0}{w}\neq 0$ take any scalar value, without impacting the complexity $c_+$ from \thm{phase-est-fwk}. This motivates the following definition.
\begin{definition}[Positive Witness for a Graph]\label{def:pos-graph-wit}
    We say $\ket{w}\in {\cal A}_G^\bot\cap {\cal B}^\bot_G$ is a positive witness for $G$ if it can be expressed as 
    $$\ket{w}=\ket{s}-\ket{\leftarrow,s}+\Pi_{E}\ket{w}+\ket{\rightarrow,t}-\ket{t}.$$
    We let $W_+(G)$ be an upper bound (over some implicit input) on the minimum $\norm{\ket{w}}^2$ of any positive witness for $G$.
    We call $\ket{\hat{w}}\colonequals \Pi_E\ket{w}$ the \emph{cropped witness}, and we let $\hat{W}_+(G)=W_+(G)-4$ be an upper bound on the minimum $\norm{\ket{\hat{w}}}^2$. 
\end{definition}
To parse the definition above, note that if $\braket{\psi_0}{w}=\sqrt{2}$, then we can always express $\ket{w}$ as 
$$\ket{w}=\sqrt{2}\ket{\psi_0}+(I-\ket{\psi_0}\bra{\psi_0})\ket{w}=\ket{s}-\ket{t}+(I-\ket{\psi_0}\bra{\psi_0})\ket{w},$$
but more than that, $\ket{w}$ is orthogonal to $\ket{s}+\ket{t}\in {\cal A}_G+{\cal B}_G$, so $(I-\ket{\psi_0}\bra{\psi_0})\ket{w}$ is orthogonal to both $\ket{s}-\ket{t}$ and $\ket{s}+\ket{t}$,
so we have:
$$\ket{w}=\ket{s}-\ket{t}+(I-\ket{s}\bra{s}-\ket{t}\bra{t})\ket{w}
=\ket{s}-\ket{\leftarrow,s}-\ket{t}+\ket{\rightarrow,t}+(I-\Pi_B)\ket{w},$$
since $\ket{w}$ must be orthogonal to $\ket{s}+\ket{\leftarrow,s}$ and $\ket{t}+\ket{\rightarrow,t}$. 
We consider $\ket{w}$ with its boundary $\Pi_B\ket{w}$ removed -- which we call a cropped witness -- because in \sec{composition}, when we compose subspace graphs, we will crop off their boundary spaces to save additive constants in the witness sizes, which become relevant when we have very deep recursive structures. Thus, we similarly define:
\begin{definition}[Negative Witness for a Graph]\label{def:neg-graph-wit}
    We say $\ket{w_{\cal A}}\in {\cal A}_G$ is a negative witness for $G$ if 
    $\ket{w_{\cal A}}-\ket{s}+\ket{t}\in {\cal B}_G$. We define the corresponding \emph{cropped negative witness} by $\ket{\hat{w}_{\cal A}}=\Pi_E\ket{w_{\cal A}}$.
\end{definition}
For the above, recall 
that a negative witness for a phase estimation algorithm is a vector $\ket{w_{\cal A}}\in {\cal A}$ such that $\ket{w_{\cal B}}\colonequals \ket{\psi_0}-\ket{w_{\cal A}}\in {\cal B}$. We have slightly modified the definition, by replacing $\ket{\psi_0}=\frac{1}{\sqrt{2}}(\ket{s}-\ket{t})$ with $\ket{s}-\ket{t}$, in order to save factors of $\sqrt{2}$ everywhere, but this changes very little. We will use the following later in the paper.
\begin{lemma}Assume $G$ has canonical $st$-boundary.
    Then $\ket{w_{\cal A}}\in {\cal A}_G$ is a negative witness for $G$ if and only if $\ket{\hat{w}_{\cal A}}+\ket{\leftarrow,s}-\ket{\rightarrow,t}\in {\cal B}_G$. 
\end{lemma}
\begin{proof}
    $\ket{w_{\cal A}}\in {\cal A}_G$ is a negative witness for $G$ if and only if $\ket{w_{\cal A}}=\ket{s}-\ket{t}+\ket{w_{\cal B}}$ for some $\ket{w_{\cal B}}\in {\cal B}_G$. Since the only vectors in ${\cal A}_G$ that overlap the boundary $\Xi_B$ are $\Xi_B^{\cal A}=\mathrm{span}\{\ket{s}+\ket{\leftarrow,s},\ket{t}+\ket{\rightarrow,t}\}$, it must be the case that 
    $$\ket{w_{\cal A}} = \ket{s}+\ket{\leftarrow,s} - \ket{t}-\ket{\rightarrow,t}+\underbrace{\Pi_E\ket{w_{\cal A}}}_{=\ket{\hat{w}_{\cal A}}}.$$
    Thus, $\ket{w_{\cal A}}$ is a negative witness if and only if 
    \begin{align*}
        &\ket{w_{\cal A}}-\ket{s}+\ket{t} \in {\cal B}_G\\
        \Leftrightarrow \;\; & \ket{\hat{w}_{\cal A}}+\ket{\leftarrow,s}-\ket{\rightarrow,t} \in {\cal B}_G.\qedhere
    \end{align*}
\end{proof}

In the next section, we describe what positive and negative witnesses look like for switching networks, and compare them with the witnesses for quantum walks on a graph (see \examp{multiQW-QW}).

Finally, we can define the complexity of a subspace graph by:
$${\cal C}(G)\colonequals \sqrt{{W}_+(G){W}_-(G)}.$$
In the case of subspaces graphs with canonical boundary, where we have defined cropped witnesses, we will often use
$$\hat{\cal C}(G)\colonequals \sqrt{\hat{W}_+(G)\hat{W}_-(G)},$$
which differs from ${\cal C}(G)$ by additive constants. 
Then it follows from \thm{phase-est-fwk} and \clm{basis-implement-U} that:
\begin{theorem}\label{thm:subspace-graph-to-alg}
    Let $G$ be a subspace graph that computes $f:\{0,1\}^n\rightarrow\{0,1\}$ with respect to $\ket{\psi_0}$, with working bases that can be generated in time $T$, and such that $\hat{W}_+(G)=O(1)$. Then there exists a quantum algorithm that decides $f$ in time complexity ${O}\left(T\sqrt{\hat{W}_-(G)}\right)$ and space $O(\log\dim H_G + \log \hat{W}_-(G))$.
\end{theorem}

Given any subspace graph $G$, we can apply the linear map 
$$\Pi_B+\frac{1}{\hat{W}_+(G)}\Pi_E$$
to $H_G$ to get a new subspace graph $G'$, which will still have canonical boundary if $G$ did, and such that
$$\hat{W}_+(G')=1
\quad\mbox{and}\quad
\hat{W}_-(G')=\hat{W}_+(G)\hat{W}_-(G).$$
This justifies thinking of $\hat{\cal C}(G)$ as a complexity.
Such a scaling preserves the basis properties in \defin{composable-basis}, and simply scales $\r$. We can see this formally in \cor{scaling}.

\subsection{Example: Switching Networks}\label{sec:switching-networks}

We now discuss switching networks (\defin{switching-network}) in more detail, and give several examples. In the classical model of switching networks, a switching network is a graph with a variable $\varphi(e)$ associated with each edge $e$, and two terminals $s,t\in V$. A switching network computes a function $f:\{0,1\}^E\rightarrow\{0,1\}$ if for any $x\in\{0,1\}^E$, $f(x)=1$ if and only if $s$ and $t$ are connected in $G(x)$, the subgraph consisting only of edges $e$ such that $x_e=1$ if $\varphi(e)$ is a positive literal, and $x_e=0$ otherwise -- that is, edges labelled by literals that are true when the variables are set according to $x$. Here, we let ``switching network'' refer to a special kind of subspace graph, but this subspace graph implements the switching network, in the sense that the algorithm referred to in \thm{subspace-graph-to-alg} decides if $s$ and $t$ are connected in $G(x)$.

We start by discussing some properties of switching networks. A switching network always has:
\begin{multline}\label{eq:switching-network-cal-A}
{\cal A}_G=\mathrm{span}\left\{\ket{a_s}\colonequals\frac{1}{\sqrt{2}}(\ket{s}+\ket{\leftarrow,s}),\ket{a_t}\colonequals\frac{1}{\sqrt{2}}(\ket{t}+\ket{\rightarrow,t})\right\}\\
\cup\left\{\ket{a_{e}}\colonequals\frac{1}{\sqrt{2}}(\ket{\rightarrow,e}-(-1)^{\varphi(e)}\ket{\leftarrow,e}):e\in E\right\},
\end{multline}
so there is a natural choice of working basis $\Psi_{\cal A}$, and we immediately get the following:
\begin{lemma}\label{lem:cal-A-switching}
    If $G$ is a switching network and $\Psi_{\cal A}$ is the basis in \eq{switching-network-cal-A}, then $\Psi_{\cal A}$ satisfies the conditions of \defin{composable-basis}, and can be generated in one query to the unitary $\ket{e}\mapsto (-1)^{\varphi(e)}\ket{e}$ and $O(1)$ additional operations. 
\end{lemma}

Turning to ${\cal B}_G$, a switching network always has:
\begin{multline}\label{eq:switching-network-cal-B}
{\cal B}_G=\mathrm{span}\{\ket{\psi_\star(u)}:u\in V\setminus\{s,t\}\}\cup\{\ket{\leftarrow,s}+\ket{\psi_\star(s)},\ket{\rightarrow,t}+\ket{\psi_\star(t)}\}\\
+\mathrm{span}\left\{\frac{1}{\sqrt{2}}(\ket{\rightarrow,e}+\ket{\leftarrow,e}):e\in E\right\}.
\end{multline}
The working basis for ${\cal B}_G$ is less obvious, since the first and second space are not orthogonal, but one choice is given by the following, which we state for intuition (the proof is a simple exercise).
\begin{lemma}
    If $G$ is a switching network,
    $$\Psi_{\cal B}=\Psi_{\cal B}^-\cup\left\{\ket{b_{e}}\colonequals\frac{1}{\sqrt{2}}(\ket{\rightarrow,e}+\ket{\leftarrow,e}):e\in E\right\}$$
    is an orthonormal basis for ${\cal B}_G$ whenever $\Psi_{\cal B}^-$ is an orthonormal basis for 
    $${\cal B}^-=\mathrm{span}\{\ket{\psi^-_\star(u)}:u\in V\setminus\{s,t\}\}\cup\{\ket{\leftarrow,s}+\ket{\psi^-_\star(s)},\ket{\rightarrow,t}+\ket{\psi^-_\star(t)}\},$$
    where
    $$\ket{\psi_\star^-(u)}=\sum_{e\in E^\rightarrow(u)}\frac{\sqrt{\w_e}}{2}(\ket{\rightarrow,e}-\ket{\leftarrow,e})+\sum_{e\in E^\leftarrow(u)}\frac{\sqrt{\w_e}}{2}(\ket{\leftarrow,e}-\ket{\rightarrow,e}).$$
\end{lemma}
For intuition, the space $\mathrm{\cal B}^-$ is the \emph{cut space} of $G$, which is the span of all states that are (correctly weighted) superpositions of ``edges'' -- where an edge is represented by $\ket{\rightarrow,e}-\ket{\leftarrow,e}$ -- that for a cut set (their removal leaves $G$ disconnected).

We give some intuition on the positive and negative witnesses of a switching network, which can also, to some extent, be used for intuition about the positive and negative witnesses for the more general subspace graphs considered in this paper. One can show that the orthogonal complement of ${\cal B}_G$ is the span of all unit $st$-flows (\defin{resistance}) of $G$, by which we mean, precisely, all states of the form
\begin{equation}\label{eq:flow-witness}
\ket{w}=\ket{s}-\ket{\leftarrow,s}+\underbrace{\sum_{e\in E}\frac{\theta(e)}{\sqrt{\w_e}}(\ket{\rightarrow,e}-\ket{\leftarrow,e})}_{\ket{\hat{w}}}+\ket{\rightarrow,t}-\ket{t},
\end{equation}
where $\theta$ is a unit $st$-flow. For every $e\in E$ such that $\varphi(e)=0$, we add $\ket{\rightarrow,e}-\ket{\leftarrow,e}$ to ${\cal A}_G+{\cal B}_G$, which adds the constraint on $\ket{w}\in {\cal A}_G^\bot\cap {\cal B}_G^\bot$ that $\theta(e)=0$. Thus, $\theta$ must be a unit flow on the graph $G(x)$. 

A negative witness exists if and only if there is no unit $st$-flow on $G(x)$, meaning $s$ and $t$ are not connected. In that case, there is always an $st$-cut-set of edges missing from $G(x)$, $F$. It is a simple exercise to show that 
$$\ket{w_{\cal A}}=\ket{s}+\ket{\leftarrow,s}+\sum_{e\in F}\sqrt{\w_e}(\ket{\rightarrow,e}-\ket{\leftarrow,e})-\ket{\rightarrow,t}-\ket{t}$$
is a cropped negative witness. However, not all witnesses have this form -- superpositions of cuts are also negative witnesses. Ref.~\cite{jarret2018connectivity} gives a tight analysis of negative witness for switching networks. 

We briefly remark here on the difference between switching networks and quantum walks. In a quantum walk, the positive witnesses are also of the form in \eq{flow-witness}, except that they are not restricted to a subgraph $G(x)$, but rather, can be on all of $G$. Negative witnesses for quantum walks are somewhat different, although they are also derived from $s$-$M$-cuts; they instead involve summing over all edges in the component containing $s$, which should not be connected to $M$ in the negative case. This is because in switching networks, ${\cal B}_G$ contains more, since $\Xi_E^{\cal B}$ is non-trivial, which makes finding a basis for ${\cal B}_G$ more difficult, but makes negative witnesses smaller.  

Finally, we give two examples of switching networks -- for computing the OR of $d$ bits (\sec{OR-gadget}) and the AND of $d$ bits (\sec{AND-gadget}). These examples will also be important building blocks for our later results.
Switching networks for Boolean formulas date back to work of Shannon~\cite{shannon1938switchingNetworks,shannon1949switchingNetworks}, and their analysis as subspace graphs was already implicit in the span program constructions of \cite{jeffery2017stConnFormula}. While \cite{jeffery2017stConnFormula} also analyzes the time complexity of evaluating switching networks, the reflection around ${\cal B}_G$ is done via a quantum walk on $G$, which results in a dependence on the relaxation time of $G$. In our examples, we improve on this by giving an orthonormal basis for ${\cal B}_G$ that can be used to reflect around it directly. Such a basis that is also efficient to generate might not be available for arbitrary switching networks, but in the case of series-parallel graphs, which correspond exactly to switching networks for Boolean formulas, our work implies that a good basis exists. 

\subsubsection{Switching Network for OR}\label{sec:OR-gadget}

In this section, we describe a switching network (\defin{switching-network}) that computes the OR of $d$ Boolean variables. This serves as a simple example, and will also be a building block in our divide-and-conquer application in \sec{D-and-C}. Formally, we prove the following.
\begin{lemma}\label{lem:OR-gadget}
    For any $d\geq 1$, and positive weights $\{\w_i\}_{i\in [d]}$, there is a switching network $G_{\textsc{or},d}$ that computes $\bigvee_{i=1}^d\varphi(e_i)$ with $\dim H_{G_{\textsc{or},d}}=2d+4$ such that:
    \begin{enumerate}
        \item $G_{\textsc{or},d}$ has $st$-composable working bases, with scaling factor $\r=2\sum_{i=1}^d\w_i$, that can be generated in $O(\log d)$ time, assuming the state proportional to $\sum_{i=1}^d\sqrt{\w_i}\ket{i}$ can be generated in time $O(\log d)$;
        \item if $\varphi(e_i)=1$ for some $i\in [d]$, $G_{\textsc{or},d}$ has cropped positive witness $\ket{\hat{w}}=\frac{1}{\sqrt{\w_i}}(\ket{\rightarrow,i}-\ket{\leftarrow,i})$; and
        \item if $\varphi(e_i)=0$ for all $i\in [d]$, $G_{\textsc{or},d}$ has cropped negative witness $\ket{\hat{w}_{\cal A}}=\sum_{i=1}^d\sqrt{\w_i}(\ket{\rightarrow,i}-\ket{\leftarrow,i})$.
    \end{enumerate}
\end{lemma}

\begin{figure}
\centering
    \begin{tikzpicture}
    \draw[dashed] (-1,0)--(0,0);
    \draw plot [smooth] coordinates{(0,0) (.75,.7) (1.5,1) (2.25,.7) (3,0)}; 
    \draw plot [smooth] coordinates{(0,0) (.75,.5) (1.5,.7) (2.25,.5) (3,0)}; 
    \node at (1.5,-.3) {$\vdots$};
    \draw plot [smooth] coordinates{(0,0) (.75,-.7) (1.5,-1) (2.25,-.7) (3,0)}; 
    \draw[dashed] (3,0)--(4,0);

        \filldraw (0,0) circle (.075);  
        \filldraw (3,0) circle (.075); 
        \node at (0,.3) {$s$};
            \node at (1.5,1.25) {$e_1$};
            \node at (1.5,.4) {$e_2$};
            \node at (1.5,-1.25) {$e_d$};
            \node at (3,.3) {$t$};
    \end{tikzpicture}
    \caption{The graph $G_{\textsc{or}}$. The dashed lines represent dangling boundary ``edges''.}\label{fig:or-graph}
\end{figure}

We remark that if $\w_i=1$ for all $i$, then \lem{OR-gadget} implies $\hat{W}_+(G_{\textsc{or},d})=2$ and $\hat{W}_-(G_{\textsc{or},d})=2d$, so by \thm{subspace-graph-to-alg}, there is a quantum algorithm for evaluating $d$-bit OR with time complexity $\widetilde{O}(\sqrt{d})$, which is optimal. This is a good sanity check, but we will mostly be interested in this construction as a building block, rather than in its own right.

Let $G=G_{\textsc{or},d}$ be defined, as shown in \fig{or-graph} by:
$$V=\{s,t\}
\mbox{ and }
E=\{e_i:i\in [d]\}$$
where each $e_i$ has endpoints $s$ and $t$, so $E(s)=E^\rightarrow(s)=E$ and $E(t)=E^\leftarrow(t)=E$. Since $G$ is a switching network, it has boundary
$B=V=\{s,t\}.$
We will let the graph be weighted, with weights $\w_{e_i}=\w_i$. The only reason to let these vary in $i$ is for later when we replace an edge with a gadget by composition, but for the sake of intuition, the reader may wish to imagine $\w_i=1$ for all $i$.
Since $G$ is a switching network, every edge is a switch, which fixes the following spaces (we simplify notation by using $i$ to label the edge $e_i$):
\begin{align*}
\forall i\in[d],\;&\Xi_{e_i}=\mathrm{span}\{\ket{\rightarrow,i},\ket{\leftarrow,i}\},\\
&\Xi_{e_i}^{\cal A}=\mathrm{span}\{\ket{\rightarrow,i}-(-1)^{\varphi(e_i)}\ket{\leftarrow,i}\},\mbox{ and }\Xi_{e_i}^{\cal B}=\mathrm{span}\{\ket{\rightarrow,i}+\ket{\leftarrow,i}\}.
\end{align*}
Furthermore, as a switching network has canonical $st$-boundary, the following spaces are fixed:
\begin{align*}
&\Xi_{s}=\mathrm{span}\{\ket{s},\ket{\leftarrow,s}\},\quad
\Xi_{s}^{\cal A}=\mathrm{span}\{\ket{s}+\ket{\leftarrow,s}\},\mbox{ and }\Xi_{s}^{\cal B}=\{0\}\\
&\Xi_{t}=\mathrm{span}\{\ket{\rightarrow,t},\ket{t}\},\quad
\Xi_{t}^{\cal A}=\mathrm{span}\{\ket{\rightarrow,t}+\ket{t}\},\mbox{ and }\Xi_{t}^{\cal B}=\{0\}.
\end{align*}
Finally, since all vertices are simple, the following spaces are fixed:
\begin{align*}&{\cal V}_{s} = \mathrm{span}\Bigg\{\ket{\leftarrow,s}+\underbrace{\sum_{i=1}^d\sqrt{\w_i}\ket{\rightarrow,i}}_{\ket{\psi_\star(s)}}\Bigg\} \mbox{ and }
{\cal V}_{t} = \mathrm{span}\Bigg\{\underbrace{\sum_{i=1}^d\sqrt{\w_i}\ket{\leftarrow,i}}_{\ket{\psi_\star(t)}}+\ket{\rightarrow,t}\Bigg\}.
\end{align*}

Then we have:
\begin{equation}\label{eq:or-cal-A-B}
    \begin{split}
        {\cal A}_G &= \bigoplus_{e\in E\cup B}\Xi_e^{\cal A} = \mathrm{span}\{\ket{\rightarrow,i}-(-1)^{\varphi(e_i)}\ket{\leftarrow,i}:i\in [d]\}\cup\{\ket{s}+\ket{\leftarrow,s},\ket{t}+\ket{\rightarrow,t}\}\\
        \mbox{and }{\cal B}_G &= {\cal V}_s\oplus{\cal V}_t+\bigoplus_{e\in E}\Xi_e^{\cal B}
={\cal V}_s\oplus{\cal V}_t+\mathrm{span}\{\ket{\rightarrow,i}+\ket{\leftarrow,i}:i\in [d]\}.
    \end{split}
\end{equation}
By \lem{cal-A-switching}, since we assume we can query the values $\varphi(e)$ in unit cost, there is a working basis for $\Psi_{\cal A}$ satisfying the conditions of \defin{composable-basis} that can be generated in unit time. To complete the proof of the the first item of \lem{OR-gadget}, we prove the following.
\begin{lemma}
Let
\begin{multline*}
\Psi_{\cal B}= \left\{\ket{b_0}=\frac{1}{\sqrt{2}}(\ket{\leftarrow,s}+\ket{\rightarrow,t}),\ket{b_1}=\frac{\ket{\leftarrow,s}-\ket{\rightarrow,t}+\sqrt{\r}\ket{\bar{b}_1}}{\sqrt{2+\r}}\right\}\\
\cup\left\{\ket{b_{e_i}}=\frac{1}{\sqrt{2}}(\ket{\rightarrow,i}+\ket{\leftarrow,i}):i\in[d]\right\},
\end{multline*}
where
$$\r=2\sum_{i=1}^d\w_i\quad\mbox{ and }\quad\ket{\bar{b}_1}=\frac{1}{\sqrt{\r}}\sum_{i=1}^d\sqrt{\w_i}(\ket{\rightarrow,i}-\ket{\leftarrow,i}).$$
Then $\Psi_{\cal B}$ is an orthonormal basis for ${\cal B}_G$, and it satisfies all the conditions of an $st$-composable basis in \defin{composable-basis}. Furthermore, as long as we can generate the state proportional to $\sum_{i=1}^d\sqrt{\w_i}\ket{i}$ in time $O(\log d)$, $\Psi_{\cal B}$ can be generated in time $O(\log d)$.
\end{lemma}
\begin{proof}
Note that:
$$\ket{\leftarrow,s}+\ket{\rightarrow,t}=\underbrace{\ket{\leftarrow,s}+\sum_{i=1}^d\sqrt{\w_i}\ket{\rightarrow,i}}_{\in{\cal V}_s}+\underbrace{\sum_{i=1}^d\sqrt{\w_i}\ket{\leftarrow,i}+\ket{\rightarrow,t}}_{\in {\cal V}_t} -\sum_{i=1}^d\sqrt{\w_i}\underbrace{(\ket{\rightarrow,i}+\ket{\leftarrow,i})}_{\in\Xi_e^{\cal B}}$$
and
$$\ket{\leftarrow,s}-\ket{\rightarrow,t}+\sqrt{\r}\ket{\bar{b}_1} = 
\underbrace{\ket{\leftarrow,s}+\sum_{i=1}^d\sqrt{\w_i}\ket{\rightarrow,i}}_{\in{\cal V}_s}-\underbrace{\sum_{i=1}^d\sqrt{\w_i}\ket{\leftarrow,i}+\ket{\rightarrow,t}}_{\in {\cal V}_t}.$$
From there, it is simple to check that $\Psi_{\cal B}\subset {\cal B}_G$. It is also simple to see that $\Xi_e^{\cal B}\subset\mathrm{span}\Psi_{\cal B}$ for all $e$. Thus, we check ${\cal V}_s$ and ${\cal V}_t$. We have:
\begin{align*}
\ket{\leftarrow,s}+\sum_{i=1}^d\sqrt{\w_i}\ket{\rightarrow,i}
={}& \frac{1}{2}\left(\ket{\leftarrow,s}+\sum_{i=1}^d\sqrt{\w_i}\left(\ket{\rightarrow,i}-\ket{\leftarrow,i}\right)-\ket{\rightarrow,t}\right)\\
&\qquad+\frac{1}{2}\left(\ket{\leftarrow,s}+\ket{\rightarrow,t}\right)+\frac{1}{2}\sum_{i=1}^d\sqrt{\w_i}(\ket{\rightarrow,i}+\ket{\leftarrow,i}),
\end{align*}
showing that ${\cal V}_s\subset\mathrm{span}\Psi_{\cal B}$. A similar proof shows that ${\cal V}_t\subset\mathrm{span}\Psi_{\cal B}$.

It is clear by inspection that $\Psi_{\cal B}$ satisfies the properties of \defin{composable-basis}, so to complete the proof, it is a simple exercise to observe that the basis can be generated in $O(\log d)$ time. 
\end{proof}

\noindent To prove the second item of \lem{OR-gadget}, we show the following.
\begin{lemma}\label{lem:or-gadget-pos}
If $\varphi(e_i)=1$ for some $i$, then the following is a positive witness for $G$ (\defin{pos-graph-wit}):
$$\ket{w}=\ket{s}-\ket{\leftarrow,s}+\underbrace{\frac{1}{\sqrt{\w_i}}\left(\ket{\rightarrow,i}-\ket{\leftarrow,i}\right)}_{\ket{\hat{w}}}+\ket{\rightarrow,t}-\ket{t}.$$
\end{lemma}
\begin{proof}
One can check that $\ket{w}\in {\cal A}_G^\bot\cap {\cal B}_G^\bot$, by verifying that it is orthogonal to each of the spaces in~\eq{or-cal-A-B}.
For orthogonality with 
$$\Xi_{e_i}^{\cal A}=\mathrm{span}\{\ket{\rightarrow,i} - (-1)^{\varphi(i)}\ket{\leftarrow,i}\},$$
we relied on the fact that $\varphi(i)=1$. Otherwise, we would have $\Xi_{e_i}^{\cal A}+\Xi_{e_i}^{\cal B}=\Xi_{e_i}$, and a positive witness must be orthogonal to $\Xi_e$, so the edge $e$ is effectively blocked. 
\end{proof}

\noindent Finally, to prove the third item of \lem{OR-gadget}, we show the following.

\begin{lemma}\label{lem:or-gadget-neg}
If $\varphi(e_i)=0$ for all $i$, then the following is a negative witness for $G$ (\defin{neg-graph-wit}):
\begin{align*}
\ket{w_{\cal A}} &=\ket{s}+\ket{\leftarrow,s}+\underbrace{\sum_{i=1}^d\sqrt{\w_i}\left(\ket{\rightarrow,i}-(-1)^{\varphi(i)}\ket{\leftarrow,i}\right)}_{\ket{\hat{w}_{\cal A}}}-\ket{\rightarrow,t}-\ket{t}.
\end{align*}
\end{lemma}
\begin{proof}
Since $\varphi(i)=0$ for all $i\in [d]$, we have:
\begin{multline*}
\ket{s}-\ket{t} = \underbrace{\ket{s}+\ket{\leftarrow,s}}_{\in \Xi_{s}^{\cal A}}-\Bigg(\underbrace{\ket{\leftarrow,s}+\sum_{i=1}^d\sqrt{\w_i}\ket{\rightarrow,i}}_{\in {\cal V}_s\subset {\cal B}}\Bigg)+ \sum_{i=1}^d\sqrt{\w_i}\underbrace{(\ket{\rightarrow,i}-(-1)^{\varphi(i)}\ket{\leftarrow,i})}_{\in \Xi_{e_i}^{\cal A}}\\
+ \Bigg(\underbrace{\sum_{i=1}^d\sqrt{\w_i}\ket{\leftarrow,i}+\ket{\rightarrow,t}}_{\in {\cal V}_t\subset{\cal B}}\Bigg)-(\underbrace{\ket{\rightarrow,t}+\ket{t}}_{\in\Xi_t^{\cal A}}),
\end{multline*}
implying that $\ket{w_{\cal A}}$ is a negative witness. 
Note that this crucially relies on $\varphi(e_i)=0$ for all $i$, otherwise the above expression fails to hold. 
\end{proof}

\subsubsection{Switching Network for AND}\label{sec:AND-gadget}

In this section, we describe a switching network (\defin{switching-network}) that computes the AND of $d$ Boolean variables.
\begin{lemma}\label{lem:AND-gadget}
    For any $d\geq 1$, and positive weights $\{\w_i\}_{i\in d}$, there is a switching network $G_{\textsc{and},d}$ that computes $\bigwedge_{i=1}^d\varphi(e_i)$ with $\dim H_{G_{\textsc{and},d}}=2d+4$ such that:
    \begin{enumerate}
        \item $G_{\textsc{and},d}$ has $st$-composable working bases, with scaling factor $\r=2/\sum_{i=1}^d(1/\w_i)$, that can be generated in $O(\log d)$ time, assuming the state proportional to $\sum_{i=1}^d\frac{1}{\sqrt{\w_i}}\ket{i}$ can be generated in time $O(\log d)$;
        \item if $\varphi(e_i)=1$ for all $i\in [d]$, $G_{\textsc{and},d}$ has cropped positive witness $\ket{\hat{w}}=\sum_{i=1}^d\frac{1}{\sqrt{\w_i}}(\ket{\rightarrow,i}-\ket{\leftarrow,i})$; and
        \item if $\varphi(e_i)=0$ for some $i\in [d]$, $G_{\textsc{and},d}$ has cropped negative witness $\ket{\hat{w}_{\cal A}}=\sqrt{\w_i}(\ket{\rightarrow,i}-\ket{\leftarrow,i})$.
    \end{enumerate}
\end{lemma}
As with the OR switching network in \sec{OR-gadget}, a corollary of this lemma is that there is a quantum algorithm for evaluating AND in optimal time $\widetilde{O}(\sqrt{d})$. 

\begin{figure}
\centering
    \begin{tikzpicture}
    \draw[dashed] (-1,0)--(0,0);
    \draw (0,0)--(3.5,0); \draw (4.5,0)--(5,0);
    \draw[dashed] (5,0)--(6,0);

        \filldraw (0,0) circle (.075);  \filldraw (1.5,0) circle (.075);  \filldraw (3,0) circle (.075); \node at (4,0) {$\dots$};  \filldraw (5,0) circle (.075);

        \node at (0,.3) {$u_0$};
            \node at (.75,.3) {$e_1$};
            \node at (1.5,.3) {$u_1$};
            \node at (2.25,.3) {$e_2$};
            \node at (3,.3) {$u_2$};
            \node at (5,.3) {$u_d$};
    \end{tikzpicture}
    \caption{The graph $G_{\textsc{and}}$. The dashed lines represent dangling boundary ``edges''.}\label{fig:and-graph}
\end{figure}

Let $G=G_{\textsc{and},d}$ be the graph in \fig{and-graph}, defined by
$$V=\{s=u_0,u_1,\dots,u_d=t\}
\mbox{ and }
E=\{e_i=(u_{i-1},u_i):i\in [d]\},$$
so $E(s)=E^{\rightarrow}(s)=\{e_1\}$, $E(t)=E^{\leftarrow}(t)=\{e_d\}$, and for all $i\in [d-1]$, $E^\leftarrow(u_i)=\{e_i\}$ and $E^\rightarrow(u_i)=\{e_{i+1}\}$. Since $G$ is a switching network, it has boundary $B=\{s,t\}$. We will let the graph be weighted, with weights $\w_{e_i}=\w_i$. Since $G$ is a switching network: every edge is a switch, which fixes $\Xi_{e_i}$, $\Xi_{e_i}^{\cal A}$ and $\Xi_{e_i}^{\cal B}$ for all $i\in [d]$; 
$G$ has canonical $st$-boundary, which fixes $\Xi_s$, $\Xi_s^{\cal A}$, $\Xi_s^{\cal B}$, $\Xi_t$, $\Xi_t^{\cal A}$, and $\Xi_t^{\cal B}$; and every vertex is simple, which fixes the spaces ${\cal V}_{u_i}$ for $i\in\{0,\dots,d\}$. In particular, letting $\w_0=\w_{d+1}=1$, $s=0$, $t=d+1$, and $i\in [d]$ label $e_i$, we must have:
\begin{align*}
\forall i\in \{0,\dots,d\},\; &{\cal V}_{u_i} = \mathrm{span}\left\{\sqrt{\w_i}\ket{\leftarrow,i}+\sqrt{\w_{i+1}}\ket{\rightarrow,i+1}\right\}.
\end{align*}
This fully defines
$${\cal A}_G=\bigoplus_{e\in E\cup B}\Xi_e^{\cal A}
\quad\mbox{ and }\quad
{\cal B}_G=\bigoplus_{i=0}^d{\cal V}_{u_i}+\bigoplus_{e\in E}\Xi_e^{\cal B}.$$
By \lem{cal-A-switching}, since we assume we can query the values $\varphi(e)$ in unit cost, there is a working basis for $\Psi_{\cal A}$ satisfying the conditions of \defin{composable-basis} that can be generated in unit time. To complete the proof of the the first item of \lem{AND-gadget}, we prove the following.
\begin{lemma}
    Let $\ket{b_2},\dots,\ket{b_d}$ be an orthonormal basis for 
    $$\mathrm{span}\{\ket{\rightarrow,i}-\ket{\leftarrow,i}\}\cap \mathrm{span}\left\{\ket{\bar{b}_1}=\frac{\sum_{i=1}^d\frac{1}{\sqrt{\w_i}}(\ket{\rightarrow,i}-\ket{\leftarrow,i})}{\sqrt{\sum_{i=1}^d\frac{2}{\w_i}}}\right\}^\bot.$$
Let 
    \begin{equation*}
        \Psi_{\cal B}^-=\left\{\ket{b_0}=\frac{1}{\sqrt{2}}(\ket{\leftarrow,s}+\ket{\rightarrow,t}),\ket{b_1}=\frac{\ket{\leftarrow,s}-\ket{\rightarrow,t}+\sqrt{\r}\ket{\bar{b}_1}}{\sqrt{2+\r}}\right\}\\
    \cup\left\{\ket{b_2},\dots,\ket{b_d}\right\}.
    \end{equation*}
Then 
    \begin{equation*}
        \Psi_{\cal B}=\Psi_{\cal B}^-\cup\left\{\ket{b_{e_i}}=\frac{1}{\sqrt{2}}(\ket{\rightarrow,i}+\ket{\leftarrow,i}):i\in [d]\right\}
    \end{equation*}
    is an orthonormal basis for ${\cal B}_G$ when $\r=2/{\sum_{i=1}^d\frac{1}{\w_i}}$, and it satisfies all the conditions of an $st$-composable basis in \defin{composable-basis}. Furthermore, as long as we can generate the state proportional to $\sum_{i=1}^d\frac{1}{\sqrt{\w_i}}\ket{i}$ in time $O(\log d)$, $\Psi_{\cal B}$ can be generated in time $O(\log d)$.
\end{lemma}
\begin{proof}
First, it is easy to verify that $\Psi_{\cal B}$ is indeed orthonormal. For each $i\in[d]$, $\ket{\rightarrow,i}+\ket{\leftarrow,i}$ is orthogonal to $\Psi_{\cal B}^-\subset\mathrm{span}\{\ket{\leftarrow,s},\ket{\rightarrow,t}\}\cup\{\ket{\rightarrow,i}-\ket{\leftarrow,i}:i\in [d]\}$. Finally, $\Psi_{\cal B}^-$ is an orthonormal set because for $i\in \{2,\dots,d\}$, $\braket{\leftarrow,s}{b_i}=\braket{\rightarrow,t}{b_i}=\braket{\bar{b}_1}{b_i}=0$, and 
$\ket{b_0}$ and $\ket{b_1}$ are also orthogonal.

Next, observe
\begin{align*}
    {\cal B} &= \mathrm{span}
    \left\{
        \sqrt{\w_i}\ket{\leftarrow,i}+\sqrt{\w_{i+1}}\ket{\rightarrow,i+1}: i\in \{0,\dots,d\}
    \right\}
\cup \left\{
        \sqrt{\w_i}(\ket{\rightarrow,i}+\ket{\leftarrow,i}): i\in [d]
    \right\},
\end{align*}
from which it follows that $\dim{\cal B}\leq d+1+d = 2d+1$, and in fact, it is not difficult to see that $\dim{\cal B}=2d+1$. 

We next argue that $\overline{\cal B}_G={\cal B}_G^\bot$, where 
\begin{equation}\label{eq:overline-cal-B}
    \overline{\cal B}_G = \mathrm{span}\{\ket{s},\ket{t}\}\oplus  
    \mathrm{span}\left\{ 
         \ket{\leftarrow,s}-\ket{\rightarrow,t}-\frac{2}{\sqrt{\r}}\ket{\bar{b}_1}
    \right\}.
\end{equation}
We clearly have $\ket{s},\ket{t}\in {\cal B}_G^\bot$, and rewriting:
\begin{align*}
    \ket{\leftarrow,s}-\ket{\rightarrow,t}-\frac{2}{\sqrt{\r}}\ket{\bar{b}_1}&=\ket{\leftarrow,s}-\ket{\rightarrow,t}-\sum_{i=1}^d\frac{1}{\sqrt{\w_i}}(\ket{\rightarrow,i}-\ket{\leftarrow,i})\\
    &= \sum_{i=0}^{d}\left(\frac{1}{\sqrt{\w_i}}\ket{\leftarrow,i}-\frac{1}{\sqrt{\w_{i+1}}}\ket{\rightarrow,i+1}\right),
\end{align*}
where we have used $s=0$, $t=d+1$, and $\w_0=\w_{d+1}=1$,
we see that it is also in ${\cal B}_G^\bot$, so $\overline{\cal B}_G\subseteq{\cal B}_G^\bot$. Since $\dim\overline{\cal B}_G=3$, and
$$\dim{\cal B}_G^\bot = \underbrace{\dim\Xi_s+\sum_{i=1}^d\dim\Xi_{e_i}+\dim\Xi_t}_{H_G} - \dim{\cal B}
= 2+2d+2 - (2d+1) = 3,$$
(see \defin{switch-edge} and \defin{canonical-boundary}) we see that $\overline{\cal B}_G={\cal B}_G^\bot$.

Next, to check that $\Psi_{\cal B}\subseteq {\cal B}_G$, we just need to check that each element of $\Psi_{\cal B}$ is orthogonal to each of the three vectors in the definition of $\overline{\cal B}_G$. This will be sufficient, since $\Psi_{\cal B}$ is an orthonormal set with $|\Psi_{\cal B}|=2+d-1+d = 2d+1=\dim{\cal B}_G$. It is simple to see that everything in $\Psi_{\cal B}$ is orthogonal to $\ket{s}$ and $\ket{t}$. It is also simple to see that everything in $\Psi_{\cal B}\setminus\Psi_{\cal B}^-$ is orthogonal to all of $\overline{\cal B}_G$. 

We turn our attention to $\Psi_{\cal B}^-$. $\ket{b_0}$ is obviously orthogonal to $\overline{\cal B}_G$. For $\ket{b_1}$, we have:
\begin{align*}
    \bra{b_1}\left(\ket{\leftarrow,s}-\ket{\rightarrow,t}-\frac{2}{\sqrt{\r}}\ket{\bar{b}_1}\right)
    ={}& \frac{1}{\sqrt{2+\r}}\left(\bra{\leftarrow,s}-\bra{\rightarrow,t}+\sqrt{\r}\bra{\bar{b}_1}\right)\left(\ket{\leftarrow,s}-\ket{\rightarrow,t}-\frac{2}{\sqrt{\r}}\ket{\bar{b}_1}\right)\\
    ={}& \frac{1}{\sqrt{2+\r}}\left(2-\sqrt\r\frac{2}{\sqrt{\r}}\braket{\bar{b}_1}{\bar{b}_1}\right)
    =0.
\end{align*}
Finally, for any $i\in\{2,\dots,d\}$, we have:
\begin{align*}
    \bra{b_i}\left(\ket{\leftarrow,s}-\ket{\rightarrow,t}-\frac{2}{\sqrt{\r}}\ket{\bar{b}_1}\right)
    &= -\frac{2}{\sqrt{\r}}\braket{b_i}{\bar{b}_1}=0,
    \end{align*}
    since the $\ket{b_i}$ are all orthogonal to $\ket{\bar{b}_1}$.
    Thus $\Psi_{\cal B}\subseteq {\cal B}_G$, and so $\Psi_{\cal B}$ is an orthonormal basis for ${\cal B}_G$.

It is clear by inspection that $\Psi_{\cal B}$ satisfies the properties of \defin{composable-basis}, so to complete the proof, we need to argue that $\Psi_{\cal B}$ can be generated in $O(\log d)$ time. It is a simple exercise to observe that the basis in \eq{overline-cal-B} can be generated in $O(\log d)$ time, and then the result follows from \cor{dual-basis-gen}. 
\end{proof}

\noindent To prove the second item of \lem{AND-gadget}, we show the following.

\begin{lemma}
    If $\varphi(e_i)=1$ for all $i\in [d]$, then the following is a positive witness for $G$ (\defin{pos-graph-wit}): 
    $$\ket{w}=\ket{s}-\ket{\leftarrow,s}+\underbrace{\sum_{i=1}^{d}\frac{1}{\sqrt{\w_i}}(\ket{\rightarrow,i}-\ket{\leftarrow,i})}_{\ket{\hat{w}}}+\ket{\rightarrow,t}-\ket{t}.$$
\end{lemma}
\begin{proof}
Since $\varphi(e_i)=1$ for all $i\in [d]$, we have $\Xi_{e_i}^{\cal A}=\Xi_{e_i}^{\cal B}=\mathrm{span}\{\ket{\rightarrow,i}+\ket{\leftarrow,i}\}$, so clearly $\ket{w}$ is orthogonal to $\Xi_{e_i}^{\cal A}+\Xi_{e_i}^{\cal B}$ for each $i$. One can similarly verify by inspection that it is orthogonal to $\Xi_s^{\cal A}$, as well as $\Xi_{t}^{\cal A}$. This leaves only the spaces ${\cal V}_{u_i}$ for $i\in\{0,\dots,d\}$, which we can verify are orthogonal to $\ket{w}$ by rewriting it as:
$$\ket{w}=\ket{s}-\sum_{i=0}^{d}\left(\frac{1}{\sqrt{\w_i}}\ket{\leftarrow,i}-\frac{1}{\sqrt{\w_{i+1}}}\ket{\rightarrow,i+1}\right)-\ket{t}.$$
We used the fact that $\w_0=\w_{d+1}=1$, $0=s$ and $d+1=t$. 
Thus, $\ket{w}\in {\cal A}_G^\bot\cap{\cal B}_G^\bot$.
\end{proof}

\noindent Finally, to prove the third item of \lem{AND-gadget}, we show the following.
\begin{lemma}
    If $\varphi(e_i)=0$ for some $i\in [d]$, then the following is a negative witness for $G$ (\defin{neg-graph-wit}):
$$\ket{w_{\cal A}}=\ket{s}+\ket{\leftarrow,s}+\underbrace{{\sqrt{\w_j}}(\ket{\rightarrow,j}-\ket{\leftarrow,j})}_{\ket{\hat{w}}}-\ket{\rightarrow,t}-\ket{t}.$$
\end{lemma}
\begin{proof}
First note that for any $j\in[d]$, using $\w_0=1$ and $\ket{\leftarrow,0}=\ket{\leftarrow,s}$:
\begin{align*}
    &\sum_{i=1}^{j}\underbrace{(\sqrt{\w_{i-1}}\ket{\leftarrow,i-1}+\sqrt{\w_i}\ket{\rightarrow,i})}_{\in {\cal V}_{u_{i-1}}\subset {\cal B}}-\sum_{i=1}^{j-1}\sqrt{\w_i}(\underbrace{\ket{\rightarrow,i}+\ket{\leftarrow,i}}_{\in\Xi_{e_i}^{\cal B}})\\
    ={}& \sqrt{\w_0}\ket{\leftarrow,0}+\sqrt{\w_j}\ket{\rightarrow,j}
    =\ket{\leftarrow,s}+\sqrt{\w_j}\ket{\rightarrow,j}
\end{align*}
is in ${\cal B}_G$. By a similar argument, $\sqrt{\w_j}\ket{\leftarrow,j}+\ket{\rightarrow,t}\in {\cal B}_G$.

Since $\varphi(e_i)=0$, $\ket{\rightarrow,i}-(-1)^{\varphi(e_i)}\ket{\leftarrow,i}=\ket{\rightarrow,i}-\ket{\leftarrow,i}$ is in $\Xi_{e_i}^{\cal A}$, we have:
\begin{multline*}
\ket{s}-\ket{t}=\underbrace{\ket{s}+\ket{\leftarrow,s}}_{\in\Xi_s^{\cal A}}-\underbrace{(\ket{\leftarrow,s}+\sqrt{\w_j}\ket{\rightarrow,j})}_{\in{\cal B}}
+\sqrt{\w_j}\underbrace{(\ket{\rightarrow,j}-\ket{\leftarrow,j})}_{\in \Xi_{e_j}^{\cal A}}\\
+\underbrace{\sqrt{\w_j}\ket{\leftarrow,j}+\ket{\rightarrow,t}}_{\in{\cal B}}-\underbrace{(\ket{\rightarrow,t}+\ket{t})}_{\in \Xi_t^{\cal A}},
\end{multline*}
implying that $\ket{{w}_{\cal A}}$ is a negative witness. 
\end{proof}

\subsection{Example: Classical Deterministic Algorithms}

In the next section, we will see how to turn a quantum algorithm into a subspace graph. As a warmup example, we first see how classical deterministic algorithms fit into this framework. The simplest example is a reversible classical algorithm, which we can model as an undirected random walk on a line. It is not very sensible to model a classical algorithm as a random walk on a line, because we incur a quadratic slow-down by letting the algorithm travel backwards sometimes, instead of always forwards, but when we move from random walks to quantum walks, it is a perfectly valid thing to do, since a quantum walker traverses a line in linear time, rather than quadratic.

A reversible classical deterministic algorithm can be described by the following objects.
\begin{itemize}
    \item A finite set ${\cal Z}$ of the algorithm's states;
    \item A sequence of permutations $\pi_1,\ldots,\pi_T:{\cal Z}\to {\cal Z}$ that depend (implicitly) on the input;
    \item An initial state $z_1\in {\cal Z}$;
    \item A partitioning of ${\cal Z}$ into accepting and rejecting states ${\cal Z}={\cal Z}_{\text{acc}}\sqcup {\cal Z}_{\text{rej}}$. 
\end{itemize}

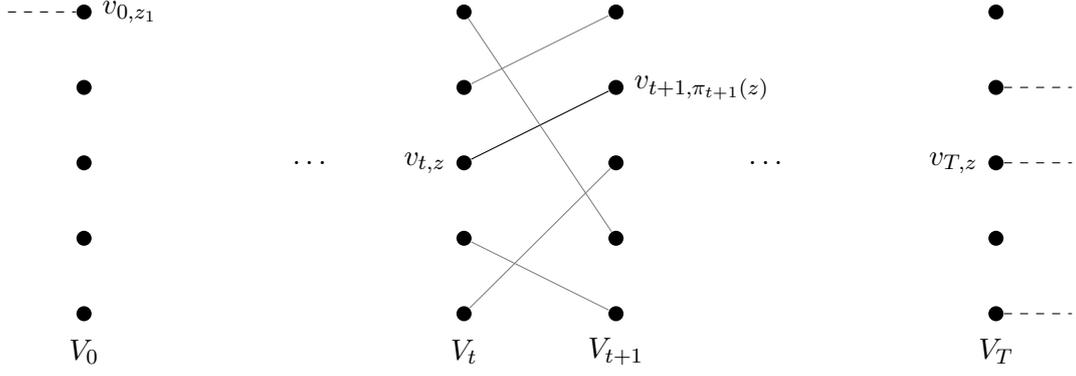
\begin{figure}
\centering
\begin{tikzpicture}

\foreach \i in {1,2,3,4,5} {
    \node[fill=black, circle, inner sep=2pt] (LL\i) at (-2,-\i) {};
}

\node at (1,-3) {$\dots$};

\foreach \i in {1,2,3,4,5} {
    \node[fill=black, circle, inner sep=2pt] (L\i) at (3,-\i) {};
}

\foreach \i in {1,2,3,4,5} {
    \node[fill=black, circle, inner sep=2pt] (R\i) at (5,-\i) {};
}

\node at (7,-3) {$\dots$};

\foreach \i in {1,2,3,4,5} {
    \node[fill=black, circle, inner sep=2pt] (RR\i) at (10,-\i) {};
}

\node[anchor=west] at (LL1.east) {$v_{0,z_1}$};
\node[anchor=east] at (L3.west) {$v_{t,z}$};
\node[anchor=west] at (R2.east) {$v_{t+1,\pi_{t+1}(z)}$};
\node[anchor=east] at (RR3.west) {$v_{T,z}$};

\draw[-] (L3) -- (R2);
\draw[-, gray] (L1) -- (R4);
\draw[-, gray] (L2) -- (R1);
\draw[-, gray] (L4) -- (R5);
\draw[-, gray] (L5) -- (R3);

\draw[dashed] (-3,-1)--(LL1);
\draw[dashed] (RR2)--(11,-2);
\draw[dashed] (RR3)--(11,-3);
\draw[dashed] (RR5)--(11,-5);

\node at (-2,-5.5) {$V_0$};
\node at (3,-5.5) {$V_t$};
\node at (5,-5.5) {$V_{t+1}$};
\node at (10,-5.5) {$V_T$};

\end{tikzpicture}
\caption{ The graph $G$ for a reversible classical deterministic algorithm with five algorithm states. Between each set $V_t$ and $V_{t+1}$, there is exactly one edge connected to $s=v_{0,z_1}$. The dashed lines represent dangling boundary ``edges''. }\label{fig:classical_algo_graph}
\end{figure}

From these, we define $G=(V,E)$ by specifying its vertices and edges:
\begin{align*}
V&=\bigsqcup_{\ell=0}^T V_\ell, \mbox{ where }V_\ell=\{v_{\ell,z}: z\in {\cal Z}\},\\
\mbox{and }E&=\{e_{\ell,z}=(v_{\ell,z},v_{\ell+1,\pi_{\ell+1}(z)}):0\le \ell \le T-1, z\in {\cal Z}\}.
\end{align*}
Note that the graph $G$ depends on the input because the edges are specified by the permutations. 
The boundary is defined $B=\{s\}\cup M$ where $s=v_{0,z_1}$ and $M=\{v_{T,z}:z\in{\cal Z}_{\text{acc}}\}$.

Next, we introduce a subspace graph structure on $G$ in order to run a quantum walk following \examp{multiQW-QW}. It is easy to see that the connected component containing $s$ is a line with one vertex in each of $V_{\ell}$, and the algorithm accepts if and only if the vertex in $V_T$ connected to $s$ is in $M$ -- that is, if and only if $s$ is connected to some vertex in $M$.
Thus the quantum walk will decide whether $\pi_T\circ\ldots\circ\pi_1(z_1)\in {\cal Z}_{\text{acc}}$.  See \fig{classical_algo_graph} for a visualization of $G$. We set all edge weights to be $1$, and then, following \examp{multiQW-QW}, we define boundary and edge spaces:
\begin{align*}
    \Xi_{s}&=\mathrm{span} \{ \ket{s},\ket{\leftarrow,s}\},\quad\Xi_{s}^{\cal A}=\mathrm{span} \{ \ket{s}+\ket{\leftarrow,s}\},\quad\mbox{and}\quad \Xi_{s}^{\cal B}=\{0\}\\
    \forall z\in {\cal Z}_{\text{acc}},\; \Xi_{v_{T,z}}&=\mathrm{span} \{ \ket{\rightarrow,v_{T,z}},\ket{v_{T,z}}\},\quad\mbox{and}\quad\Xi_{v_{T,z}}^{\cal A}=\Xi_{v_{T,z}}^{\cal B}=\{0\}\\
    \forall \ell\in\{0,\dots,T-1\},z\in{\cal Z},\; \Xi_{e_{\ell,z}}&=\mathrm{span}\{\ket{\rightarrow,\ell,z},\ket{\leftarrow,\ell,z}\},\\
    \Xi_{e_{\ell,z}}^{\cal A}&=\mathrm{span}\{\ket{\rightarrow,\ell,z}+\ket{\leftarrow,\ell,z}\},\;\mbox{ and }\Xi_{e_{\ell,z}}^{\cal B}=\{0\},
\end{align*}
and all vertices are simple (\defin{simple-vertex}), so:
\begin{align*}
    \forall z\in {\cal Z},\; {\cal V}_{v_{0,z}}&=\mathrm{span}\{\ket{\psi_\star(v_{0,z})}=\ket{\rightarrow,{0,z}}\},\\
    \forall t\in\{1,\dots,T-1\}, z\in{\cal Z},\; {\cal V}_{v_{\ell,z}}&=\mathrm{span}\{\ket{\psi_\star(v_{\ell,z})}=\ket{\leftarrow,{\ell-1,\pi_{\ell}^{-1}(z)}}+\ket{\rightarrow,{\ell,z}}\}\\
    \forall z\in {\cal Z},\; {\cal V}_{v_{T,z}}&=\mathrm{span}\{\ket{\psi_\star(v_{0,z})}=\ket{\leftarrow,{T-1,\pi_T^{-1}(z)}}\}.
\end{align*}
While the algorithm is represented by a path, as is the graph for AND in \sec{AND-gadget}, in contrast to \sec{AND-gadget}, we represent this path in terms of the local constraints $\ket{\leftarrow,\ell-1,\pi_\ell^{-1}(z)}+\ket{\rightarrow,\ell,z}$. That is because, unlike the path in \sec{AND-gadget}, this path is defined sequentially. These local constraints can be generated by calling the relevant permutation, which is one step of the algorithm, whereas one can verify that generating a basis like that in \sec{AND-gadget}, which includes a superposition over the whole path, or whole sequence of algorithm steps, would require running the whole algorithm.

Unlike the other examples in this paper, to get a quantum walk algorithm, following \examp{multiQW-QW}, we need $\ket{\psi_0}=\ket{s}$. We do not have $B=\{s,t\}$ (assuming $|{\cal Z}_{\text{acc}}|\neq 1$), so we cannot use canonical boundary here, but it would be possible to modify this subspace graph to have canonical boundary, as follows. We can essentially take two copies of $G$, $G^{\rightarrow}$ and $G^\leftarrow$, and identify the points $M$ in each of the two copies, so that the connected component of $s$ is a line that first goes through $G^\rightarrow$ to some vertex in $V_T$, and then, if it has reached a point in $M$, goes through $G^\leftarrow$ back to the second copy of $s$, which we could call $t$. This is like a classical algorithm that computes a final state, copies out the answer bit, and then uncomputes, which is also what we do for quantum algorithms in the next section, except that the answer bit goes into the phase. 

Another modification we could make to this subspace graph is as follows. We mentioned that a classical reversible algorithm is like a walk on a line, but $G$ is not a line, though the connected component containing $s$ always is. We can actually modify $G$ to get a subspace graph $G'$ that is a line of length $T$, where $v_{\ell}$ corresponds to the set $V_{\ell}$ in $G$, which better captures our intuition that a deterministic algorithm is a linear process. 
The edge spaces of $G'$ are obtained by combining those of $G$: $\Xi_{e_{\ell}}'=\bigoplus_{z\in{\cal Z}}\Xi_{e_{\ell,z}}$, and similarly for $\Xi_{e_{\ell}}^{'\cal A}$. The vertex spaces are similarly obtained: ${\cal V}_{v_\ell}'=\bigoplus_{z\in{\cal Z}}{\cal V}_{v_{\ell,z}}$.
This is essentially what we do for quantum algorithms in the next section.

\subsection{Example: Quantum Algorithms}\label{sec:alg-subspace-graph}

In this section, we describe a subspace graph from any quantum algorithm. This construction is essentially the same as the one in \cite{jeffery2022subroutines}, but we state it here explicitly as a subspace graph. While the construction in \cite{jeffery2022subroutines} applies to \emph{variable-time} quantum algorithms, for simplicity, we will consider algorithms that run for a fixed time $\sf T$. However, \lem{alg-subspace-graph} below, and all results that use it, also apply for variable-time algorithms, replacing $\sf T$ with the expected running time.

A quantum algorithm is a sequence of unitaries $U_1,\dots,U_{T}$ acting on the space
$${H}_{\cal Y}\otimes H_{\cal Z}=\mathrm{span}\{\ket{a}\ket{z}:a\in \{0,1\},z\in {\cal Z}\}.$$
We assume each unitary can be implemented in unit cost, and moreover, the unitary $\sum_{r=1}^T\ket{r}\bra{r}\otimes U_r$ can be implemented in unit cost.  

The algorithm may implicitly depend on some input $x\in\{0,1\}^n$, for example, by letting some of the unitaries be queries to $x$ (or in any other way, we don't care). 
We say the algorithm \emph{computes} some $f:\{0,1\}^n\rightarrow\{0,1\}$ if 
$$U_T\dots U_1\ket{0,0}\in\mathrm{span}\{\ket{f(x),z}:z\in{\cal Z}\}$$
on input $x$. We are assuming the algorithm has no error for simplicity. In some of our conclusions, we can always substitute an algorithm with error if the error is small enough that the algorithm is indistinguishable from an error-free one. In this section, we will prove the following.
\begin{lemma}\label{lem:alg-subspace-graph}
    From any quantum algorithm computing some $f:\{0,1\}^n\rightarrow\{0,1\}$ with no error in time $T$, and any positive weights $\{\alpha_r\}$ such that $\alpha_0=1$, we can derive an $st$-composable subspace graph (\defin{st-composable-subspace-graph}) $G$ that computes $f$ with $\dim H_G = O(|{\cal Z}|T)$,  $\hat{W}_+\leq 2\sum_{r=1}^{T}\frac{1}{\alpha_r}$ and $\hat{W}_-\leq 2\sum_{r=1}^T\alpha_r$; with $st$-composable working bases that have scaling factor $\r=2$ and can be generated in time $O(\log(T))$.
\end{lemma}
\noindent For example, we can set the weights all to 1, and then we get $\hat{W}_+$ and $\hat{W}_-$ both at most $2T$.

In order to remain consistent with \cite{jeffery2022subroutines} so that we can easily use results proven there, while still keeping the boundary of $G$ simple, we are going to assume that $U_1=-I$, at the expense of potentially making the algorithm one step longer. Thus, we will let ${\sf T}=T+1$, and the algorithm referred to in \lem{alg-subspace-graph} is actually $U_2,\dots,U_{\sf T}$. 

We define a subspace graph $G=(V,E)$ as follows. We will assume $\sf T$ is even, and let:
$$V=\{v_{2\ell-1}^{\rightarrow},v_{2\ell-1}^{\leftarrow}\}_{\ell=1}^{{\sf T}/2},
\quad
E=\{e_{2\ell}^{\rightarrow},e_{2\ell}^{\leftarrow}\}_{\ell=1}^{{\sf T}/2-1}\cup\{e_{\sf T}^{\leftrightarrow}\},\quad\mbox{and}\quad
B=\{s\colonequals v_1^{\rightarrow},t\colonequals v_1^{\leftarrow}\}$$
with $e_i^{\rightarrow}=(v_{i-1}^{\rightarrow},v_{i+1}^{\rightarrow})$ (and similarly for $e_i^{\leftarrow}$), and $e_{\sf T}^{\leftrightarrow}=(v_{{\sf T}-1}^{\rightarrow},v_{{\sf T}-1}^{\leftarrow})$.
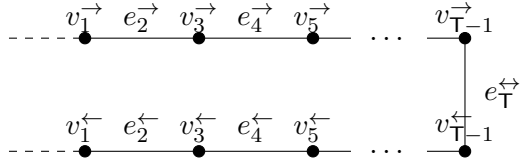
\begin{figure}
\centering
    \begin{tikzpicture}
    \draw[dashed] (-1,1.5)--(0,1.5);
    \draw[dashed] (-1,0)--(0,0);
    \draw (0,1.5)--(3.5,1.5);
    \draw (4.5,1.5)--(5,1.5)--(5,0)--(4.5,0);
    \draw (0,0)--(3.5,0);

        \filldraw (0,1.5) circle (.075);  \filldraw (1.5,1.5) circle (.075);  \filldraw (3,1.5) circle (.075); \node at (4,1.5) {$\dots$};  \filldraw (5,1.5) circle (.075);

        \filldraw (0,0) circle (.075);  \filldraw (1.5,0) circle (.075);  \filldraw (3,0) circle (.075); \node at (4,0) {$\dots$};  \filldraw (5,0) circle (.075);

        \node at (0,1.75) {$v_1^{\rightarrow}$};
            \node at (.75,1.75) {$e_2^{\rightarrow}$};
            \node at (1.5,1.75) {$v_3^{\rightarrow}$};
            \node at (2.25,1.75) {$e_4^{\rightarrow}$};
            \node at (3,1.75) {$v_5^{\rightarrow}$};
            \node at (5,1.75) {$v_{{\sf T}-1}^{\rightarrow}$};

        \node at (5.5,.75) {$e_{\sf T}^{\leftrightarrow}$};

        \node at (0,.3) {$v_1^{\leftarrow}$};
            \node at (.75,.3) {$e_2^{\leftarrow}$};
            \node at (1.5,.3) {$v_3^{\leftarrow}$};
            \node at (2.25,.3) {$e_4^{\leftarrow}$};
            \node at (3,.3) {$v_5^{\leftarrow}$};
            \node at (5,.3) {$v_{{\sf T}-1}^{\leftarrow}$};
    \end{tikzpicture}
    \caption{The graph $G$. The dashed lines represent dangling boundary ``edges''.}
\end{figure}
We define:
\begin{equation}
    \begin{split}
        \forall \ell\in\{1,\dots,{\sf T}/2-1\},\; \Xi_{e_{2\ell}^{\rightarrow}} &\colonequals  \mathrm{span}\{\ket{\rightarrow}\ket{a,z}\ket{2\ell},\ket{\rightarrow}\ket{a,z}\ket{2\ell+1}:a\in \{0,1\},z\in{\cal Z}\}\\
        \forall \ell\in\{1,\dots,{\sf T}/2-1\},\; \Xi_{e_{2\ell}^{\leftarrow}} &\colonequals  \mathrm{span}\{\ket{\leftarrow}\ket{a,z}\ket{2\ell},\ket{\leftarrow}\ket{a,z}\ket{2\ell+1}:a\in \{0,1\},z\in{\cal Z}\}\\
        \Xi_{e_{\sf T}^{\leftrightarrow}} &\colonequals \mathrm{span}\{\ket{\rightarrow}\ket{a,z}\ket{{\sf T}},\ket{\leftarrow}\ket{a,z}\ket{{\sf T}}:a\in \{0,1\},z\in{\cal Z}\}\\
        \Xi_s=\Xi_{v_1^{\rightarrow}}&\colonequals  \mathrm{span}\{\ket{s}\colonequals \ket{\rightarrow}\ket{0,0}\ket{0},\ket{\leftarrow,s}\colonequals \ket{\rightarrow}\ket{0,0}\ket{1}\}\\
        \Xi_t=\Xi_{v_1^{\leftarrow}}&\colonequals  \mathrm{span}\{\ket{t}\colonequals \ket{\leftarrow}\ket{0,0}\ket{0},\ket{\rightarrow,t}\colonequals \ket{\leftarrow}\ket{0,0}\ket{1}\}.
    \end{split}
\end{equation}

Following \cite[Definition 3.1]{jeffery2022subroutines}, we define the following sets of vectors in $H_G=\Xi_E\oplus\Xi_B$ to use in the construction of our subspaces, for positive weights $\{\alpha_r\}_r$ such that $\alpha_0=\alpha_1=1$.
\begin{definition}\label{def:transition-states}
The \emph{forward transition states} are defined:
\begin{equation*}
\begin{split}
\forall r\in\{0,\dots,{\sf T}-1\},\;
\Psi_r^{\rightarrow}&\colonequals \left\{\ket{\psi_{a,z,r}^{\rightarrow}} \colonequals  \ket{\rightarrow}\left( \sqrt{\alpha_r}\ket{a,z}\ket{r} - \sqrt{\alpha_{r+1}}U_{r+1}\ket{a,z}\ket{r+1} \right): a\in\{0,1\},z\in {\cal Z}\right\}.
\end{split}
\end{equation*}
The \emph{backward transition states} are defined:
\begin{equation*}
\begin{split}
\forall r\in\{0,\dots,{\sf T}-1\},\;
\Psi_r^{\leftarrow}&\colonequals \left\{\ket{\psi_{a,z,r}^{\leftarrow}} \colonequals  \ket{\leftarrow}\left( \sqrt{\alpha_r}\ket{a,z}\ket{r} - \sqrt{\alpha_{r+1}}U_{r+1}\ket{a,z}\ket{r+1} \right): a\in\{0,1\},z\in {\cal Z}\right\}.
\end{split}
\end{equation*}
The \emph{reversal states} are defined:
\begin{equation*}
\Psi_{\sf T}^{\leftrightarrow}\colonequals \left\{ \ket{\psi_{a,z,{\sf T}}^{\leftrightarrow}} \colonequals  \sqrt{\alpha_{\sf T}}\left(\ket{\rightarrow}-(-1)^{a}\ket{\leftarrow}\right) \ket{a,z}\ket{{\sf T}}:a\in\{0,1\}, z\in {\cal Z}\right\}.
\end{equation*}
Finally, let:
$$\Psi_0=\bigcup_{\ell=0}^{{\sf T}/2-1}(\Psi_{2\ell}^{\rightarrow}\cup\Psi_{2\ell}^{\leftarrow})\cup\Psi_{\sf T}^{\leftrightarrow}
\quad\mbox{ and }\quad
\Psi_1=\bigcup_{\ell=0}^{{\sf T}/2}(\Psi_{2\ell-1}^{\rightarrow}\cup\Psi_{2\ell-1}^{\leftarrow}).
$$
\end{definition}
Note that since we assume $U_1=-I$, we have:
$$\ket{\psi_{0,0,0}^{\rightarrow}}=\ket{\rightarrow}\ket{0,0}\ket{0} - \ket{\rightarrow}U_1\ket{0,0}\ket{1}
= \ket{s}+\ket{\rightarrow}\ket{0,0}\ket{1}=\ket{s}+\ket{\leftarrow,s},$$
and similarly,
$$\ket{\psi_{0,0,0}^{\leftarrow}}=\ket{\leftarrow}\ket{0,0}\ket{0} +\ket{\leftarrow}\ket{0,0}\ket{1}
=\ket{t}+\ket{\rightarrow,t}.$$

Then we define:
\begin{equation}
    \begin{split}
        &{\cal V}_{v_1^{\rightarrow}} \colonequals  \mathrm{span}\{\ket{\psi_{0,0,1}^{\rightarrow}}\} \quad\mbox{ and }\quad 
        {\cal V}_{v_1^{\leftarrow}} \colonequals  \mathrm{span}\{\ket{\psi_{0,0,1}^{\leftarrow}}\}\\
        \forall \ell\in\{2,\dots,{\sf T}/2\}, \;& {\cal V}_{v_{2\ell-1}^{\rightarrow}}\colonequals \mathrm{span}\Psi_{2\ell-1}^{\rightarrow}
        \quad\mbox{ and }\quad {\cal V}_{v_{2\ell-1}^{\leftarrow}}\colonequals \mathrm{span}\Psi_{2\ell-1}^{\leftarrow}\\
        \forall \ell\in\{1,\dots,{\sf T}/2-1\},  \; &\Xi_{e_{2\ell}^{\rightarrow}}^{\cal A}\colonequals \mathrm{span}\Psi_{2\ell}^{\rightarrow}
        \quad\mbox{ and }\quad \Xi_{e_{2\ell}^{\leftarrow}}^{\cal A}\colonequals \mathrm{span}\Psi_{2\ell}^{\leftarrow}\\
        &\Xi_{e_{\sf T}^{\leftrightarrow}}^{\cal A}\colonequals \mathrm{span}\Psi_{\sf T}^{\leftrightarrow}\\
        &\Xi_s^{\cal A} \colonequals  \mathrm{span}\{\ket{\psi_{0,0,0}^{\rightarrow}}\}
        =\mathrm{span}\{\ket{s}+\ket{\leftarrow,s}\}\\
         &\Xi_t^{\cal A} \colonequals  \mathrm{span}\{\ket{\psi_{0,0,0}^{\leftarrow}}\}
         =\mathrm{span}\{\ket{t}+\ket{\rightarrow,t}\}\\
         &{\cal V}_B\colonequals\mathrm{span}\{\ket{\leftarrow,s}+\ket{\rightarrow,t}\}. 
    \end{split}
\end{equation}
For all $e\in E\cup B$, we will have $\Xi_e^{\cal B}=\{0\}$.
Note that the subspace graph we have just defined has canonical $st$-boundary (\defin{canonical-boundary}). 
As per \defin{multid-QW}, we have
\begin{equation}\label{eq:alg-cal-A}
    \begin{split}
{\cal A}_G&=\Xi_{s}^{\cal A}\oplus \Xi_{t}^{\cal A}\oplus\bigoplus_{\ell=1}^{{\sf T}/2-1}(\Xi_{e_{2\ell}^{\rightarrow}}^{\cal A}\oplus \Xi_{e_{2\ell}^{\leftarrow}}^{\cal A})\oplus \Xi_{e_{\sf T}^{\leftrightarrow}}^{\cal A}\\
{\cal B}_G&=\bigoplus_{\ell=1}^{{\sf T}/2}\left({\cal V}_{v_{2\ell-1}^{\rightarrow}}\oplus{\cal V}_{v_{2\ell-1}^{\leftarrow}}\right)+{\cal V}_B.
    \end{split}
\end{equation}

In \cite{jeffery2022subroutines}, the vectors in \defin{transition-states} form the working basis, which works perfectly fine, because since all $\Xi_e^{\cal B}$ are trivial, we have a decomposition of ${\cal B}_G$ into orthogonal spaces ${\cal V}_v$, so we can simply combine their bases. However, here, we will use a slightly different working basis for ${\cal B}_G$ to conform to \defin{working-basis}, as follows.

Let $\bar{\Psi}_r^{\rightarrow}$ contain normalized vectors.

$$\Psi_{{\cal A}} = \left\{\ket{\bar\psi_{0,0,0}^{\rightarrow}}=\frac{1}{\sqrt{2}}\left(\ket{s}+\ket{\leftarrow,s}\right),\ket{\bar\psi_{0,0,0}^\rightarrow}=\frac{1}{\sqrt{2}}\left(\ket{t}+\ket{\rightarrow,t}\right)\right\}\cup\bigcup_{\ell=1}^{{\sf T}/2-1}(\bar\Psi_{2\ell}^{\rightarrow}\cup\bar\Psi_{2\ell}^{\leftarrow})\cup\bar\Psi_{{\sf T}}^{\leftrightarrow}.$$

Note that
$$\ket{\psi_{0,0,1}^{\rightarrow}}=\ket{\rightarrow}\ket{0,0}\ket{1}-\ket{\rightarrow}U_2\ket{0,0}\ket{2} = \ket{\leftarrow,s}-\ket{\rightarrow}U_2\ket{0,0}\ket{2}$$
and similarly, 
$$\ket{\psi_{0,0,1}^{\leftarrow}}=\ket{\rightarrow,t}-\ket{\leftarrow}U_2\ket{0,0}\ket{2}.$$
Define:
\begin{equation*}
\begin{split}
    \ket{b_0} &\colonequals  \frac{1}{\sqrt{2}}\left(\ket{\psi_{0,0,1}^{\rightarrow}}+\ket{\psi_{0,0,1}^{\leftarrow}}+(\ket{\rightarrow}+\ket{\leftarrow})U_2\ket{0,0}\ket{2})\right)=\frac{1}{\sqrt{2}}(\ket{\leftarrow,s}+\ket{\rightarrow,t})\\
    \ket{b_1} &\colonequals  \frac{1}{2}\left(\ket{\psi_{0,0,1}^{\rightarrow}}-\ket{\psi_{0,0,1}^{\leftarrow}}\right)=\frac{1}{2}\left(\ket{\leftarrow,s}-\ket{\rightarrow,t} - (\ket{\rightarrow}-\ket{\leftarrow})U_2\ket{0,0}\ket{2}\right)\\
    \ket{b_2} &\colonequals  \frac{1}{\sqrt{2}}(\ket{\rightarrow}+\ket{\leftarrow})U_2\ket{0,0}\ket{2}
\end{split}
\end{equation*}
Then we have
\begin{equation*}
    \mathrm{span}\{\ket{b_0}-\ket{b_2},\ket{b_1}\} = \mathrm{span}\{\ket{\psi_{0,0,1}^{\rightarrow}},\ket{\psi_{0,0,1}^{\leftarrow}}\}={\cal V}_{v_1^{\rightarrow}}\oplus{\cal V}_{v_1^{\leftarrow}}.
\end{equation*}
Thus
$$\Psi_{{\cal B}} = \{\ket{b_0},\ket{b_1},\ket{b_2}\} \cup \bigcup_{\ell=2}^{{\sf T}/2}(\Psi_{2\ell-1}^{\rightarrow}\cup\Psi_{2\ell-1}^{\leftarrow})\cup\Psi_{{\sf T}}^{\leftrightarrow}$$
is an orthonormal basis for ${\cal B}_G$.

\noindent The following is given by a trivial modification of \cite[Lemma 3.3]{jeffery2022subroutines}.
\begin{lemma}\label{lem:alg-subspace-graph-basis}
    Assuming unit-time access to $\sum_{\ell=1}^{{\sf T}}\ket{\ell}\bra{\ell}\otimes U_{\ell}$, the working bases $\Psi_{{\cal A}}$ and $\Psi_{{\cal B}}$ can be generated in $O(\log({\sf T}))$ complexity.
\end{lemma}

\noindent We will also use the following states from \cite[Definition 3.4]{jeffery2022subroutines} in the construction of our witnesses:
\begin{definition}[Algorithm States]\label{def:algorithm-states}
Define the following \emph{algorithm states} in $H_{G}$: 
\begin{align*}
\ket{w^{0}} &= \ket{0},\quad
\forall t\in [{\sf T}],\; \ket{w^t} = U_{t}\ket{w^{t-1}}.
\end{align*}
The \emph{positive history state} of the algorithm is defined:
\begin{equation*}
\ket{w_+} = (\ket{\rightarrow}+(-1)^{f(x)}\ket{\leftarrow})\sum_{t=0}^{{\sf T}}\frac{1}{\sqrt{\alpha_t}}\ket{w^t}\ket{t}.
\end{equation*}
The \emph{negative history state} of the algorithm is defined:
\begin{equation*}
\ket{w_-} = (\ket{\rightarrow}-(-1)^{f(x)}\ket{\leftarrow})\sum_{t=0}^{{\sf T}}{\sqrt{\alpha_t}}(-1)^t\ket{w^t}\ket{t}.
\end{equation*}
\end{definition}

\noindent We will use the following properties from \cite{jeffery2022subroutines}, using the fact that the algorithm has no error: 
\begin{claim}\label{clm:history-state-props}
\begin{enumerate}
\item $\norm{\ket{w_-}}^2 = 2\sum_{t=0}^{{\sf T}}\alpha_t$ (Corollary 3.7)
\item $\norm{\ket{w_+}}^2 = 2\sum_{t=0}^{{\sf T}}\frac{1}{\alpha_t}$ (Corollary 3.7)
\item $\ket{w_+}\in \mathrm{span}\Psi_0^\bot\cap\mathrm{span}\Psi_1^\bot$ 
(Claim 3.8).
\item $\ket{w_-}\in \mathrm{span}\Psi_0$ and $\ket{w_-}-\underbrace{(\ket{\rightarrow}-(-1)^{f(x)}\ket{\leftarrow})\ket{0}\ket{0}}_{\ket{s}-(-1)^{f(x)}\ket{t}}\in \mathrm{span}\Psi_1$ (Claim 3.9).
\end{enumerate}
\end{claim}
Items (1) and (2) are rather obvious in the case of fixed time quantum algorithms that we are specializing to here in order to avoiding defining variable-time quantum algorithms. However, we remark that in \cite{jeffery2022subroutines}, the more general case of variable-time algorithms is considered, so ${\sf T}$ is allowed to be a random variable, and the expressions in (1) and (2) are replaced by expectations. The construction in this section, which again, is just putting the construction of \cite{jeffery2022subroutines} into the language of subspace graphs, would be the same in this more general case, and thus our composition results also hold for this more general kind of algorithm. 

We have the following corollaries, which, combined with \lem{alg-subspace-graph-basis}, prove \lem{alg-subspace-graph}.
\begin{corollary}[Positive analysis]
    When $f(x)=1$, $\ket{w_+}$ is a positive witness for $G$ (see \defin{pos-graph-wit}) with
    $$\norm{\ket{\hat{w}_+}}^2=2\sum_{\ell=2}^{{\sf T}}\frac{1}{\alpha_{\ell}}.$$
\end{corollary}
\begin{proof}
    When $f(x)=1$, since $\alpha_0=\alpha_1=1$ and $\ket{w^0}=\ket{0,0}$ and $\ket{w^1}=U_1\ket{0,0}$,
        \begin{multline*}
    \ket{w_+} = \underbrace{\ket{\rightarrow}\ket{0,0}\ket{0}}_{\ket{e}}-\underbrace{\ket{\leftarrow}\ket{0,0}\ket{0}}_{\ket{t}} + \underbrace{\ket{\rightarrow}U_1\ket{0,0}\ket{1}}_{-\ket{\leftarrow,s}}-\underbrace{\ket{\leftarrow}U_1\ket{0,0}\ket{1}}_{-\ket{\rightarrow,t}}\\
     + \underbrace{(\ket{\rightarrow}+(-1)^{f(x)}\ket{\leftarrow})\sum_{\ell=2}^{{\sf T}}\frac{1}{\sqrt{\alpha_\ell}}\ket{w^\ell}\ket{\ell}}_{\ket{\hat{w}_+}}.
    \end{multline*}
    By inspection (see \eq{alg-cal-A} and \defin{transition-states}), we have ${\cal A}_G\subseteq \mathrm{span}\Psi_0$ and ${\cal B}_G\subseteq \mathrm{span}\Psi_1\cup\{\ket{b_0}\}$.  
    By \clm{history-state-props}, $\ket{w_+}\in\mathrm{span}\Psi_0^\bot\cap\mathrm{span}\Psi_1^\bot \subset {\cal A}_G^\bot\cap{\cal B}_G^\bot$, 
    and we can see that $\ket{w_+}$ is orthogonal to $\ket{b_0}=\frac{1}{\sqrt{2}}(\ket{\leftarrow,s}+\ket{\rightarrow,t})$.
    Thus $\ket{w_+}$ is a positive witness, and 
    $$\norm{\ket{\hat{w}_+}}^2 = \norm{(I-\Pi_B)\ket{w_+}}^2 = 2\sum_{\ell=2}^{{\sf T}}\frac{1}{\alpha_\ell}.$$
\end{proof}

\begin{corollary}[Negative analysis]
    When $f(x)=0$, $\ket{w_-}$ is a negative witness for $G$ (see \defin{neg-graph-wit}) with 
    $$\norm{\ket{\hat{w}_-}}^2 = 2\sum_{\ell=2}^{{\sf T}}\alpha_\ell.$$
    Thus $\hat{W}_-(G)\leq 2\sum_{\ell=2}^{{\sf T}}\alpha_\ell$. 
\end{corollary}
\begin{proof}
    By \clm{history-state-props}, we have $\ket{w_+}\in \mathrm{span}\Psi_0$. Though $\Psi_0$ is slightly bigger than ${\cal A}_G$, since we also have $\ket{w_+^0}=\ket{0,0}$ and $\ket{w_+^1}=-\ket{0,0}$, we can see that $\ket{w_+}\in {\cal A}_G$. Similarly, it is not difficult to conclude that when $f(x)=0$, $\ket{w_-}-(\ket{s}-\ket{t})\in {\cal B}_G$, since it's in $\Psi_1$ (\clm{history-state-props}), and its projection onto $\ket{0}$ or $\ket{1}$ in the last register is in $\mathrm{span}\Psi_0^{\rightarrow}\cup\Psi_0^\leftarrow\cup\Psi_1^\rightarrow\cup\Psi_1^\leftarrow$.
\end{proof}

\section{Recursion with Subspace Graphs}\label{sec:composition}

Subspace graphs lend themselves well to very general kinds of recursion. We can compose subspace graphs, say $G_1$, $G_2$ and $G_3$, as in \fig{arbitrary-comp}, by identifying some of the vertices on their boundaries. In full generality, this may result in a subspace graph that is difficult to analyze. For example, if we replace two parallel edges with subspace graphs derived from quantum algorithms, ``flow'' along these edges may have complex phases that interfere on the other side. Our nice classical intuition of flows breaks down in the fully general case, which is not surprising, since quantum algorithms are not classical. However, an important question is in which special cases, and to what extent, we can compose subspace graphs and keep enough classical intuition to analyze them. 

One special case is implicit in \cite{jeffery2022subroutines}, and here we present another special case, which we call \emph{switch composition}. Switch composition generalizes a very simple kind of recursion that can be done in switching networks. Since an $st$-path (or more generally, flow) is allowed to use an edge if and only if $\varphi(e)=1$, we can replace it with a switching network $G^e$ for some function $f_e$, which will have an $st$-path from one endpoint to the other if and only if $f_e(x)=1$. We can view this as removing $e$ from the graph, and then adding its endpoints to the boundary, to then be identified with the boundary $\{s,t\}$ of $G^e$. 
By applying this type of composition, the OR and AND switching networks in \sec{OR-gadget} and \sec{AND-gadget} can be combined to make switching networks for any Boolean formula~\cite{jeffery2017stConnFormula}.

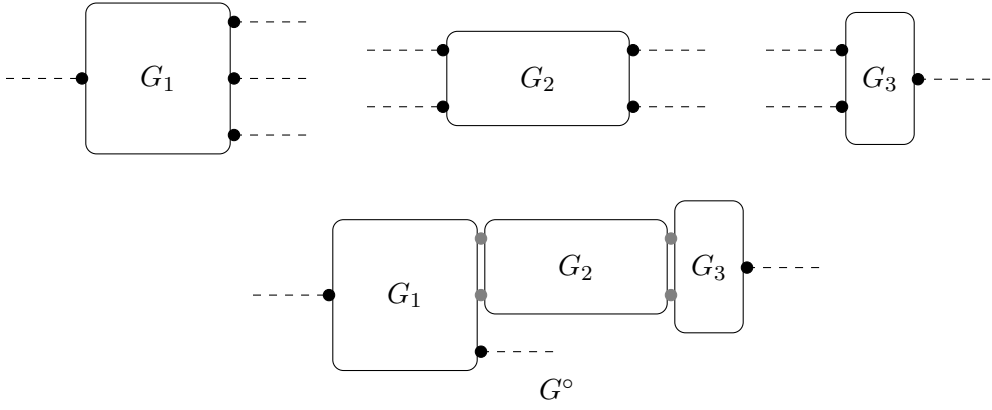
\begin{figure}
\centering
    \begin{tikzpicture}
\node at (0,0) {\begin{tikzpicture}
    \draw[dashed] (-1,0)--(0,0);
    \filldraw (0,0) circle (.075);
    \draw[rounded corners] (0.05,1) rectangle (1.95,-1);
    \node at (1,0) {$G_1$};
    \filldraw (2,.75) circle (.075); \draw[dashed] (2,.75)--(3,.75);
    \filldraw (2,0) circle (.075); \draw[dashed] (2,0)--(3,0);
    \filldraw (2,-.75) circle (.075); \draw[dashed] (2,-.75)--(3,-.75);
\end{tikzpicture}};

\node at (5,0) {\begin{tikzpicture}
    \draw[dashed] (1,.75)--(2,.75); \filldraw (2,.75) circle (.075); 
    \draw[dashed] (1,0)--(2,0); \filldraw (2,0) circle (.075); 
    \draw[rounded corners] (2.05,1) rectangle (4.45,-.25);
    \node at (3.25,.365) {$G_2$};
    \filldraw (4.5,.75) circle (.075); \draw[dashed] (4.5,.75)--(5.5,.75);
    \filldraw (4.5,0) circle (.075); \draw[dashed] (4.5,0)--(5.5,0);
\end{tikzpicture}};

\node at (9.5,0) {\begin{tikzpicture}
    \filldraw (4.5,.75) circle (.075); \draw[dashed] (3.5,.75)--(4.5,.75);
    \filldraw (4.5,0) circle (.075); \draw[dashed] (3.5,0)--(4.5,0);
    \draw[rounded corners] (4.55,1.25) rectangle (5.45,-.5);
    \node at (5,.365) {$G_3$};
    \filldraw (5.5,.365) circle (.075); \draw[dashed] (5.5,.365)--(6.5,.365);
\end{tikzpicture}};

\node at (5,-3) {\begin{tikzpicture}
    \draw[dashed] (-1,0)--(0,0);
    \filldraw (0,0) circle (.075);
    \draw[rounded corners] (0.05,1) rectangle (1.95,-1);
    \node at (1,0) {$G_1$};
    \filldraw (2,-.75) circle (.075); \draw[dashed] (2,-.75)--(3,-.75);

    \draw[rounded corners] (2.05,1) rectangle (4.45,-.25);
    \node at (3.25,.365) {$G_2$};
    \filldraw[gray] (2,.75) circle (.075); 
    \filldraw[gray] (2,0) circle (.075); 

    \draw[rounded corners] (4.55,1.25) rectangle (5.45,-.5);
    \node at (5,.365) {$G_3$};
    \filldraw[gray] (4.5,.75) circle (.075); 
    \filldraw[gray] (4.5,0) circle (.075); 
    \filldraw (5.5,.365) circle (.075); \draw[dashed] (5.5,.365)--(6.5,.365);

    \node at (3,-1.25) {$G^\circ$};
\end{tikzpicture}};
    
    \end{tikzpicture}
    \caption{Some perhaps complicated graphs $G_1$, $G_2$ and $G_3$, but we need only consider their boundaries, and can abstract their internal structure. To compose them, we identify points on their boundaries, yielding a new graph, $G^\circ$.}\label{fig:arbitrary-comp}
\end{figure}

\subsection{Switch Composition}\label{sec:switch-comp}

Here we investigate how to compose subspace graphs with specific properties that generalize switching networks -- $st$-composable subspace graphs (\defin{st-composable-subspace-graph}) -- into the switches of a graph $G$. Specifically, if $\overline{E}$ is the set of edges of $G$ that are switches, we will replace $e\in\overline{E}$ with an $st$-composable graph $G^e$. 
We can assume without loss of generality that we replace every $e\in\overline{E}$ with some $G^e$, since letting $G^e$ be the graph consisting of a single edge from $s$ to $t$ is just like not replacing $e$. 
In particular, we will prove the following theorem.
\begin{theorem}\label{thm:comp}
    Let $G$ be an $st$-composable subspace graph (\defin{st-composable-subspace-graph}) with $st$-composable working bases that can be generated in time $T$, with scaling factor $\r$. For each $e\in\overline{E}$, the set of switches of $G$, let $G^e$ be an $st$-composable subspace graph with $st$-composable working bases that can be generated in time at most $T'$, with scaling factor $\r^e$.  
    Then there exists an $st$-composable subspace graph $G^\circ$ with $\dim H_{G^\circ}\leq\dim H_G-2|\overline{E}|+\sum_{e\in\overline{E}}\dim H_{G^e}$
    such that:
    \begin{itemize}
        \item $G^\circ$ has $st$-composable working bases that can be generated in time $T+T'+O(1)$, with scaling factor~$\r$. 
        \item If $\ket{\hat{w}}$ is a cropped positive witness for $G$ (\defin{pos-graph-wit}), and for each $e\in\overline{E}$ such that $\braket{\rightarrow,e}{\hat{w}}\neq 0$, $\ket{\hat{w}^e}$ is a cropped positive witness for $G^e$, then 
        $$\ket{\hat{w}^\circ}=\sum_{e\in\overline{E}}\sqrt{\frac{\r^e}{2}}\braket{\rightarrow,e}{\hat{w}}\ket{\hat{w}^e}+\Pi_{E\setminus\overline{E}}\ket{\hat{w}}$$
        is a cropped positive witness for $G^\circ$.
        \item If $\ket{\hat{w}_{\cal A}}$ is a cropped negative witness for $G$ (\defin{neg-graph-wit}), and for each $e\in\overline{E}$ such that $\braket{\rightarrow,e}{\hat{w}_{\cal A}}\neq 0$, $\ket{\hat{w}^e_{\cal A}}$ is a cropped negative witness for $G^e$, then 
        $$\ket{\hat{w}^\circ_{\cal A}}=\sum_{e\in\overline{E}}\sqrt{\frac{2}{\r^e}}\braket{\rightarrow,e}{\hat{w}_{\cal A}}\ket{\hat{w}^e_{\cal A}}+\Pi_{E\setminus\overline{E}}\ket{\hat{w}_{\cal A}}$$
        is a cropped negative witness for $G^\circ$.
    \end{itemize}
\end{theorem}
It follows, using the observation that $\ket{\hat{w}}$ must only overlap switches that are on, and $\ket{\hat{w}_{\cal A}}$ must only overlap switches tht are off,
that if $G$ computes $f$ and each $G^e$ computes $f^e$, the function computed by $G^\circ$ is the composed function $f\circ(f_e)_{e\in\overline{E}}$ (see \sec{formulas}).

Before we prove \thm{comp}, we note the following useful corollary of \thm{comp} and \lem{OR-gadget}.
\begin{corollary}\label{cor:scaling}
    Let $G$ be an $st$-composable subspace graph with $st$-composable working bases that can be generated in time $T$ that computes $f$. Then there is an $st$-composable subspace graph $G^\circ$ with $st$-composable bases that can be generated in time $T+O(1)$ that computes $f$, and such that $\dim H_{G^\circ}=\dim H_G+O(1)$, $\hat{W}_+(G^\circ)\leq 1$ and $\hat{W}_-(G^\circ) \leq \hat{W}_+(G)\hat{W}_-(G)$.
\end{corollary}
\begin{proof}
    Let $G_{\textsc{or},1}$ be the graph from \lem{OR-gadget} with $d=1$, and $\w_1=\frac{\r}{2}\hat{W}_+(G)$, where $\r$ is the scaling factor of the basis for $G$. It is $st$-composable since it is a switching network. This graph has a single edge, $e_1$, and we will compose $G$ into this edge using \thm{comp}, so $G_{\textsc{or}}$ plays the role of $G$ in \thm{comp}, and $G$ plays the role of $G^{e_1}$. If $f(x)=1$, then $G$ has a cropped positive witness $\ket{\hat{w}^1}$, and using the positive witness from \lem{OR-gadget}, by \thm{comp}, $G^\circ$ has cropped positive witness 
    $$\ket{\hat{w}^\circ} = \sqrt{\frac{\r}{2}}\frac{1}{\sqrt{\w_1}}\ket{\hat{w}^1}=\frac{1}{\sqrt{\hat{W}_+(G)}}\ket{\hat{w}^1},$$
    so since $\norm{\ket{\hat{w}^1}}^2\leq \hat{W}_+(G)$, $\hat{W}_+(G^\circ)\leq 1$. 

    On the other hand, if $f(x)=0$, then $G$ has a cropped negative witness $\ket{\hat{w}_{\cal A}^1}$, and using the negative witness for \lem{OR-gadget}, by \thm{comp}, $G^\circ$ has cropped negative witness
    $$\ket{\hat{w}_{\cal A}^\circ}=\sqrt{\frac{2}{\r}}\sqrt{\w_1}\ket{\hat{w}_{\cal A}^1}=\sqrt{\hat{W}_+(G)}\ket{\hat{w}_{\cal A}^1},$$
    so $\hat{W}_-(G^\circ) \leq \hat{W}_+(G)\hat{W}_-(G)$.

    To complete the proof, by \lem{OR-gadget}, the basis for $G_{\textsc{or}}$ can be generated in unit time, and thus by \thm{comp}, the basis for $G^\circ$ can be generated in time $T+O(1)$.
\end{proof}

\noindent In the remainder of this section, we prove \thm{comp}.

\vskip10pt

\begin{figure}
\centering
    \begin{tikzpicture}
\node at (0,-1) {$G$};
\node at (0,0) {\begin{tikzpicture}
    \draw[dashed] (-1,0)--(0,0);
    \node at (0,.25) {$s$};
    \filldraw (0,0) circle (.075);

    \filldraw (1,1) circle (.075);
    \filldraw (2,0) circle (.075);
    \filldraw (4,1) circle (.075);    
    \filldraw (4,-1) circle (.075);

    \draw[orange] (2,0)--(4,1);
    \draw[blue] (0,0)--(2,0);
    \draw (0,0)--(1,1)--(2,0);
    \draw (4,1)--(5,0)--(4,-1)--(2,0);

\node at (1,-.25) {$e_1$};
\node at (3,.75) {$e_2$};

    \draw[dashed] (5,0)--(6,0);
    \node at (5,.25) {$t$};
    \filldraw (5,0) circle (.075);
\end{tikzpicture}};

\node at (6.25,-1.5) {$G^{e_1}$};
\node at (6.25,0) {\begin{tikzpicture}
    \fill[fill=blue!20,rounded corners] (0,.15) rectangle (2,-2);
    
    \draw[dashed] (-1,0)--(0,0);
    \node at (0,.25) {$s$};
    \filldraw (0,0) circle (.075);

    \filldraw[blue] (1,-1) circle (.075);
    \filldraw[blue] (.5,-1.7) circle (.075);
    \filldraw[blue] (1.5,-1.7) circle (.075);

    \draw[blue] (0,0)--(1,-1);
    \draw[blue] (0,0)--(.5,-1.7);
    \draw[blue] (.5,-1.7)--(1,-1)--(1.5,-1.7)--(2,0);
    \draw[blue] (2,0)--(1,-1);
    \draw[blue] (.5,-1.7)--(1.5,-1.7);
    
    \draw[dashed] (2,0)--(3,0);
    \node at (2,.25) {$t$};
    \filldraw (2,0) circle (.075);
\end{tikzpicture}};

\node at (11,-1.5) {$G^{e_2}$};
\node at (11,0) {\begin{tikzpicture}
    \fill[fill=orange!20,rounded corners] (2,0) rectangle (4,2);
    
    \draw[dashed] (1,0)--(2,0);
    \node at (1.8,.25) {$s$};
    \filldraw (2,0) circle (.075);

    \filldraw[orange] (3,1.7) circle (.075);
    \filldraw[orange] (2.5,.85) circle (.075);
    \filldraw[orange] (3,.5) circle (.075);
    \filldraw[orange] (3.5,1.25) circle (.075);

    \draw[orange] (2,0)--(4,1);
    \draw[orange] (2,0)--(2.5,.85)--(3,1.7)--(3.5,1.25)--(4,1);
    \draw[orange] (2.5,.85)--(3.5,1.25)--(3,.5);

    \draw[dashed] (4,1)--(5,1);
    \node at (4.15,1.25) {$t$};
    \filldraw (4,1) circle (.075);
\end{tikzpicture}};

\node at (5,-5.3) {$G^\circ$};
\node at (5,-4.3) {\begin{tikzpicture}

    \fill[fill=blue!20,rounded corners] (0,.15) rectangle (2,-2);
    
    \filldraw[blue] (1,-1) circle (.075);
    \filldraw[blue] (.5,-1.7) circle (.075);
    \filldraw[blue] (1.5,-1.7) circle (.075);

    \draw[blue] (0,0)--(1,-1);
    \draw[blue] (0,0)--(.5,-1.7);
    \draw[blue] (.5,-1.7)--(1,-1)--(1.5,-1.7)--(2,0);
    \draw[blue] (2,0)--(1,-1);
    \draw[blue] (.5,-1.7)--(1.5,-1.7);
    
    \fill[fill=orange!20,rounded corners] (2,0) rectangle (4,2);
    
    \filldraw[orange] (3,1.7) circle (.075);
    \filldraw[orange] (2.5,.85) circle (.075);
    \filldraw[orange] (3,.5) circle (.075);
    \filldraw[orange] (3.5,1.25) circle (.075);

    \draw[orange] (2,0)--(4,1);
    \draw[orange] (2,0)--(2.5,.85)--(3,1.7)--(3.5,1.25)--(4,1);
    \draw[orange] (2.5,.85)--(3.5,1.25)--(3,.5);

    \draw[dashed] (-1,0)--(0,0);
    \node at (0,.25) {$s$};
    \filldraw (0,0) circle (.075);

    \filldraw (1,1) circle (.075);
    \filldraw (2,0) circle (.075);
    \filldraw (4,1) circle (.075);    
    \filldraw (4,-1) circle (.075);

    \draw (0,0)--(1,1)--(2,0);
    \draw (4,1)--(5,0)--(4,-1)--(2,0);

    \draw[dashed] (5,0)--(6,0);
    \node at (5,.25) {$t$};
    \filldraw (5,0) circle (.075);

\end{tikzpicture}};
    
    \end{tikzpicture}
    \caption{An example of replacing edges $e_1$ and $e_2$ in $G$ with graphs $G^{e_1}$ and $G^{e_2}$ to obtain $G^\circ$. Boundary ``edges'' are represented by dashed lines. Switching networks already lend themselves to this type of recursion: we can replace an edge with a 2-terminal graph, and then the ``edge'' is traversable if and only if the terminals are connected. The difference with the more general recursion in \thm{comp} is that the subspace graphs used to replace edges (as well as other parts of $G$) might have more complicated structure than just their graph structure; for example, they might encode quantum algorithms like in \sec{alg-subspace-graph}. }\label{fig:switch-comp}
\end{figure}
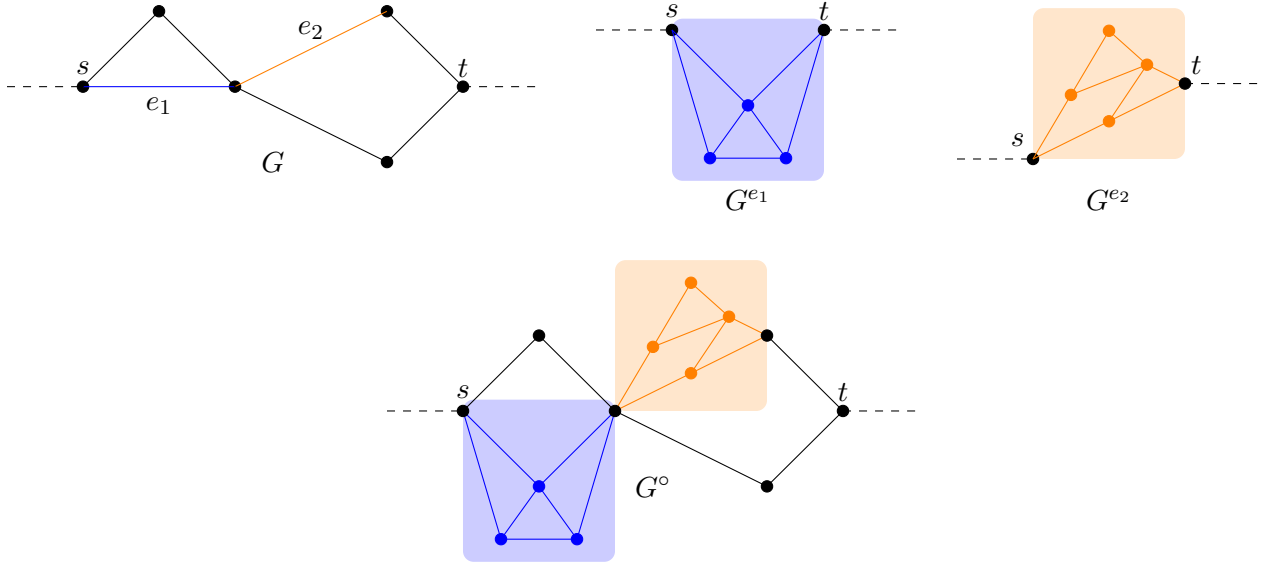

Informally, we obtain $G^\circ$ from $G$ by, for each $e\in\overline{E}$, removing the edge $e$, and identifying its endpoints with the vertices $s=s^e$ and $t=t^e$ of the graph $G^e$. This is illustrated in \fig{switch-comp}.
The subspaces associated with $G^\circ$ are mostly inherited from $G$ and $G^e$, except for those on the glued boundaries. There, intuitively, we replace $\ket{\rightarrow,e}$ with the edges incident to $s=s^e$ in $G^e$, and $\ket{\leftarrow,e}$ with the edges incident to $t=t^e$ in $G^e$. This is formalized by the map $\widetilde\Lambda$ defined in \eq{tilde-Lambda}.

\paragraph{Definition of the Subspace Graph $G^\circ$} We now formally define $G^\circ$ and its associated spaces. We will shortly define a working basis $\Psi_{{\cal B}_{G^\circ}}$ for ${\cal B}_{G^\circ}$, from working bases for ${\cal B}_{G}$ and ${\cal B}_{G^e}$. It is actually these bases that are most important, but by defining ${\cal B}_{G^\circ}$ and ${\cal A}_{G^\circ}$ as being built up from local spaces, we show that the graph structure and intuition are preserved by the composition. 

The composed graph $G^\circ$ has vertex and edge sets:
\begin{equation}
\begin{split}
    V^\circ &= V\sqcup \bigsqcup_{e\in\overline{E}}(V^e\setminus B^e) = V\sqcup \bigsqcup_{e\in\overline{E}}(V^e\setminus \{s^e,t^e\})\\
    E^\circ &= (E\setminus\overline{E})\sqcup\bigsqcup_{e\in\overline{E}}E^e
\end{split}    
\end{equation}
and boundary 
\begin{equation*}
B^\circ = B=\{s,t\}.
\end{equation*}
We identify $s^e$ and $t^e$ in $G^e$ with the endpoints of $e$ in $G$, so that the following edges are incident to $u\in V^\circ$:
\begin{equation}\label{eq:E-circ-u}
    \begin{split}
        E^\circ(u) = \left\{\begin{array}{cc}
            E^e(u) & \mbox{if }u\in V^e\setminus\{s^e,t^e\} \\
            (E(u)\setminus\overline{E})\sqcup\bigsqcup_{e\in\overline{E}\cap E^\rightarrow(u)}E^e(s^e)\sqcup\bigsqcup_{e\in\overline{E}\cap E^\leftarrow(u)}E^e(t^e) & \mbox{if }u\in V.
        \end{array}\right.
    \end{split}
\end{equation}
Recall that we assume for convenience that each edge has an orientation, and $E^\rightarrow(u)\subseteq E(u)$ are the edges oriented outwards from $u$, and $E^\leftarrow(u) = E(u)\setminus E^{\rightarrow}(u)$ those oriented inwards. This is simply to decide which of the endpoints of $e$ is identified with $s^e$ and which with $t^e$.

We define the edge spaces and edge subspaces of $G^\circ$ from those of $G$ and $G^e$ as follows:
\begin{equation}
    \Xi_{e'}^{\circ} = \left\{\begin{array}{ll}
        \Xi_{e'} & \mbox{if }e'\in (E\setminus\overline{E})\cup B \\
         \Xi_{e'}^e & \mbox{if }e'\in E^e,
    \end{array}\right.
\end{equation}
and similarly for $\Xi_e^{\circ\cal A}$ and $\Xi_e^{\circ\cal B}$. Note that since $G$ has canonical $st$-boundary, so does $G^\circ$. 
 Furthermore, since the switches and non-switches of $G^\circ$ are:
$$
\overline{E}^\circ = \bigsqcup_{e\in\overline{E}}\overline{E}^e
\mbox{ and }
E^\circ\setminus\overline{E}^\circ = (E\setminus \overline{E})\sqcup\bigsqcup_{e\in\overline{E}}(E^e\setminus\overline{E}^e),
$$
we inherit the property that for all $e\in E^\circ\setminus\overline{E}^\circ$, $\Xi_e^{\circ\cal B}=\{0\}$.

Next we define the vertex spaces. First, for each $e\in \overline{E}$, we are assuming that ${\cal V}_{s^e}^e$ and ${\cal V}_{t^e}^e$ can be written:
\begin{equation}
    \begin{split}
        {\cal V}_{s^e}^e &= \mathrm{span}\{\ket{\leftarrow,s^e}+\ket{\psi_\star(s^e)}\}\\
        {\cal V}_{t^e}^e &= \mathrm{span}\{\ket{\rightarrow,t^e}+\ket{\psi_\star(t^e)}\},
    \end{split}
\end{equation}
for some $\ket{\psi_\star(s^e)}\in \Xi_{E^e(s^e)}^e$ and $\ket{\psi_\star(t^e)}\in \Xi_{E^e(t^e)}^e$.
Then define a linear map $\widetilde\Lambda$ on $H_G$ by: 
\begin{equation}\label{eq:tilde-Lambda}
\widetilde\Lambda = \sum_{e\in\overline{E}}\frac{1}{\sqrt{\r^e}}(\ket{\psi_\star(s^e)}\bra{\rightarrow,e}+\ket{\psi_\star(t^e)}\bra{\leftarrow,e})+\Pi_{E\setminus\overline{{E}}}+\Pi_B,
\end{equation}
and let
\begin{equation}
    {\cal V}_u^{\circ} = \left\{\begin{array}{ll}
        \widetilde\Lambda({\cal V}_u) & \mbox{if }u\in V \\
         {\cal V}_u^e & \mbox{if }u\in V^e\setminus\{s^e,t^e\}.
    \end{array}\right.
\end{equation}
The following lemma shows that $G^\circ$ is a well-defined subspace graph.
\begin{lemma}
    $\widetilde\Lambda ({\cal V}_u) \subseteq \Xi_{E^\circ(u)}^{\circ}$.
\end{lemma}
\begin{proof}
We have ${\cal V}_u\subseteq \Xi_{E(u)}$, so
$$\widetilde\Lambda({\cal V}_u) = \widetilde\Lambda\left(\Xi_{E(u)\setminus\overline{E}}\cap {\cal V}_u\right)\oplus\bigoplus_{e\in\overline{E}\cap E(u)}\widetilde\Lambda\left(\Xi_{e}\cap {\cal V}_u\right),$$
We have $\widetilde\Lambda\left(\Xi_{E(u)\setminus\overline{E}}\cap {\cal V}_u\right)=\Xi_{E(u)\setminus\overline{E}}\cap {\cal V}_u$, since $\widetilde\Lambda$ fixes $\Xi_{E\setminus\overline{E}}$.
Recall from \defin{switch-edge} that for $e\in\overline{E}\cap E(u)$
$$\Xi_{e}\cap {\cal V}_u=\left\{\begin{array}{ll}
\mathrm{span}\{\ket{\rightarrow,e}\} & \mbox{if }e\in E^{\rightarrow}(u)\\
\mathrm{span}\{\ket{\leftarrow,e}\} & \mbox{if }e\in E^{\leftarrow}(u).
\end{array}\right.$$
Thus, 
$$\widetilde\Lambda\left(\Xi_{e}\cap {\cal V}_u\right) = \left\{\begin{array}{ll}
\mathrm{span}\{\widetilde\Lambda\ket{\rightarrow,e}\}=\mathrm{span}\{\ket{\psi_\star(s^e)}\}\subseteq\Xi_{E^e(s^e)}^e & \mbox{if }e\in E^{\rightarrow}(u)\\
\mathrm{span}\{\widetilde\Lambda\ket{\leftarrow,e}\}=\mathrm{span}\{\ket{\psi_\star(t^e)}\}\subseteq\Xi_{E^e(t^e)}^e & \mbox{if }e\in E^{\leftarrow}(u).
\end{array}\right.$$
Thus
$$\widetilde\Lambda({\cal V}_u) \subseteq \Xi_{E(u)\setminus\overline{E}}\oplus\bigoplus_{e\in\overline{E}\cap E^\rightarrow(u)}\Xi_{E^e(s^e)}^e\oplus\bigoplus_{e\in\overline{E}\cap E^\leftarrow(u)}\Xi_{E^e(t^e)}^e=\Xi_{E^{\circ}(u)}$$
by \eq{E-circ-u}.
\end{proof}

From these definitions, we conclude:
\begin{equation}\label{eq:cal-A-comp}
\begin{split}
    {\cal A}_{G^\circ} &= \bigoplus_{e\in E^\circ}\Xi_e^{\circ\cal A}\oplus \bigoplus_{u\in B^\circ}\Xi_{u}^{\circ\cal A}
    = \bigoplus_{e\in E\setminus\overline{E}}\Xi_e^{\cal A}\oplus \bigoplus_{e\in \overline{E}}\bigoplus_{e'\in E^e}\Xi_{e'}^{e\cal A}
    \oplus \Xi_s^{\cal A}\oplus\Xi_t^{\cal A}\\
    &=\bigoplus_{e\in E\setminus\overline{E}}\Xi_e^{\cal A}\oplus \bigoplus_{e\in \overline{E}}\bigoplus_{e'\in E^e}\Xi_{e'}^{e\cal A}
    \oplus \Xi_s^{\cal A}
 \end{split}
\end{equation}
and similarly
\begin{equation}\label{eq:cal-B-comp}
\begin{split}
    {\cal B}_{G^\circ} &= \bigoplus_{u\in V^\circ}{\cal V}_u^\circ+ \bigoplus_{e\in E\setminus\overline{E}}\Xi_e^{\cal B}\oplus \bigoplus_{e\in \overline{E}}\bigoplus_{e'\in E^e}\Xi_{e'}^{e\cal B}
    \oplus \Xi_s^{\cal B}\oplus\Xi_t^{\cal B}\\
    &= \bigoplus_{u\in V}\widetilde\Lambda({\cal V}_u)\oplus\bigoplus_{e\in\overline{E}}\bigoplus_{u\in V^e\setminus\{s,t\}}{\cal V}_u^e+ \oplus \bigoplus_{e\in \overline{E}}\bigoplus_{e'\in \overline{E}^e}\Xi_{e'}^{e\cal B}.
\end{split}
\end{equation}

\paragraph{Working Bases for $G^\circ$}

By the assumptions of \defin{composable-basis}, we have a working basis $\Psi_{{\cal A}_G}$ for ${\cal A}_G$ of the form:
$$\Psi_{{\cal A}_G}=\bigcup_{e\in E\cup B}\Psi_{{\cal A}_G}(e),$$
where $\Psi_{{\cal A}_G}(e)$ is an orthonormal basis for $\Xi_e^{\cal A}$, and similarly, for each $e\in\overline{E}$, we have a working basis $\Psi_{{\cal A}_{G^e}}$ for ${\cal A}_{G^e}$ of the form:
$$\Psi_{{\cal A}_{G^e}}=\bigcup_{e'\in E^e\cup B^e}\Psi_{{\cal A}_{G^e}}(e').$$
Since the spaces $\Xi_e^{\circ\cal A}$ are simply inherited from the corresponding spaces in $G$ and the $G^e$, their bases can be as well. 
Specifically, we define a working basis 
$$\Psi_{{\cal A}_{G^\circ}}=\bigcup_{e\in E^\circ\cup B^\circ}\Psi_{{\cal A}_{G^\circ}}(e)=\bigcup_{e\in (E\setminus\overline{E})\cup B}\Psi_{{\cal A}_{G^\circ}}(e)\cup \bigcup_{e\in \overline{E}}\bigcup_{e'\in E^e}\Psi_{{\cal A}_{G^\circ}}(e').$$
for ${\cal A}_{G^\circ}$, where for all $e\in E\setminus\overline{E}$, $\Psi_{{\cal A}_{G^\circ}}(e)=\Psi_{{\cal A}_{G}}(e)$, and for all $e\in \overline{E}$ and $e'\in E^e$, $\Psi_{{\cal A}_{G^\circ}}(e')=\Psi_{{\cal A}_{G^e}}(e')$. Then we have the following.

\begin{lemma}[Working Basis for ${\cal A}$]
    $\Psi_{{\cal A}_{G^\circ}}$ can be generated in time $T+T'+O(1)$. 
\end{lemma}
\begin{proof}
Assume without loss of generality that for $e\in \overline{E}$, $\Psi_{{\cal A}_G}(e)$, which is one-dimensional, has label set $L(e)=\{e\}$.
We let 
$$L_{\cal A}^\circ = (\{0\}\times(L\setminus\overline{E}))\sqcup\bigsqcup_{e\in\overline{E}}\{e\}\times L_{\cal A}^e,$$ so to map $\ket{\ell}\mapsto\ket{b_{\ell}^\circ}$ for $\ell=(\ell_0,\ell_1)\in L^\circ$, we simply need to apply the relevant basis generation map depending on the value of $\ell_0\in\{0\}\cup\overline{E}$. We can check if $\ell=(\ell_0,\ell_1)\in L_{\cal A}^\circ$ by membership in the relevant set, controlled on $\ell_0$. For example, if $\ell_0=0$, we check if $\ell_1\in L$, and $\ell_1\not\in\overline{E}$.
\end{proof}

By the assumptions of \defin{composable-basis}, we have a working basis $\Psi_{{\cal B}_G}$ for ${\cal B}_G$ of the form
$$\Psi_{{\cal B}_G}= \Psi_{{\cal B}_G}^-\cup\left\{\frac{1}{\sqrt{2}}(\ket{\rightarrow,e}+\ket{\leftarrow,e}):e\in \overline{E}\right\},$$
and similarly, for each $e\in \overline{E}$, we have a working basis $\Psi_{{\cal B}_{G^e}}$ for ${\cal B}_{G^e}$ of the form
$$\Psi_{{\cal B}_{G^e}}= \Psi_{{\cal B}_{G^e}}^-\cup\left\{\frac{1}{\sqrt{2}}(\ket{\rightarrow,e'}+\ket{\leftarrow,e'}):e'\in \overline{E}^e\right\},$$
where $\Psi_{{\cal B}_{G^e}}^-\subseteq H_{G^e}^-\colonequals\Xi_{B^e\cup E^e\setminus\overline{E}^e}\oplus\mathrm{span}\{\ket{\rightarrow,e}-\ket{\leftarrow,e}:e\in\overline{E}^e\}$ contains two distinguished vectors, $\ket{b_0^e}$ and $\ket{b_1^e}$. From these, we define the following isometry on $H_G^-$:
\begin{equation}\label{eq:Lambda}
\Lambda \colonequals  \Pi_{E\setminus\overline{E}}+\Pi_B+\sum_{e\in\overline{E}}\frac{1}{\sqrt{2}}\ket{\bar{b}_1^e}(\bra{\rightarrow,e}-\bra{\leftarrow,e}).    
\end{equation}
Then we define:
\begin{equation}\label{eq:cal-B-basis}
    \Psi_{{\cal B}_G^\circ}=\underbrace{\Lambda(\Psi_{{\cal B}_G}^-)\cup\bigcup_{e\in\overline{E}}(\Psi_{{\cal B}_{G^e}}^-\setminus\{\ket{b_0^e},\ket{b_1^e}\})}_{\equalscolon \Psi_{{\cal B}_{G^\circ}}^-} \cup \Bigg\{\underbrace{\frac{1}{\sqrt{2}}(\ket{\rightarrow,e'}+\ket{\leftarrow,e'})}_{\Psi_{{\cal B}_{G^\circ}}(e'),\mbox{ spans }\Xi_{e'}^{e\cal B}}:e'\in \overline{E}^\circ=\bigcup_{e\in\overline{E}}\overline{E}^e\Bigg\}.
\end{equation}
    In \app{locality-of-basis}, we prove that $\Psi_{{\cal B}_{G^\circ}}$ is indeed a basis for ${\cal B}_{G^\circ}$. Again, this is actually not so important -- we could also have defined ${\cal B}_{G^\circ}$ from the basis, in which case, we would not need the assumption, from canonical $st$-boundary, that the spaces ${\cal V}_{s^e}^e$ and ${\cal V}_{t^e}^e$ are one-dimensional. 
However, then it would not be obvious that any graph structure remains, which we feel is beneficial for maintaining intuition when working with complicated objects. 

\begin{lemma}
    $\Psi_{{\cal A}_{G^\circ}}$ and $\Psi_{{\cal B}_{G^\circ}}$ are $st$-composable bases as in \defin{composable-basis}.
\end{lemma}
\begin{proof}
The first two properties are clear, so we check properties 3-5. Since $\Psi_{{\cal B}_G}$ is composable, $\Psi_{{\cal B}_G}^-$ contains $\ket{b_0}=\frac{1}{\sqrt{2}}(\ket{\leftarrow,s}+\ket{\rightarrow,t})$, so $\Psi_{{\cal B}_{G^\circ}}^-$ contains 
$$\ket{b_0^\circ}=\Lambda(\ket{b_0}) = \frac{1}{\sqrt{2}}(\ket{\leftarrow,s}+\ket{\rightarrow,t}),$$
establishing property 3 of \defin{composable-basis}. Similarly, $\Psi_{{\cal B}_G}^-$ contains 
$\ket{b_1}=\frac{1}{\sqrt{2+\r}}(\ket{\leftarrow,s}-\ket{\rightarrow,t}+\sqrt{\r}\ket{\bar{b_1}})$, so $\Psi_{{\cal B}_{G^\circ}}^-$ contains 
$$\ket{b_1^\circ}=\Lambda(\ket{b_1}) = \frac{1}{\sqrt{2+\r}}(\ket{\leftarrow,s}-\ket{\rightarrow,t}+\sqrt{\r}\Lambda\ket{\bar{b_1}}),$$
establishing property 4 of \defin{composable-basis}. Finally, it is easy to verify that the remaining vectors of $\Psi_{{\cal B}_{G^\circ}}$ are orthogonal to $\ket{\leftarrow,s}$ and $\ket{\rightarrow,t}$, establishing property 5.
\end{proof}

\begin{lemma}
If the working basis $\Psi_{{\cal B}_G}$ has time complexity $T$, and for each $e\in\overline{E}$, the working basis $\Psi_{{\cal B}_{G^e}}$ has time complexity $T'$, then $\Psi_{{\cal B}_{G^\circ}}$ has time complexity $T+T'+O(1)$.
\end{lemma}
\begin{proof}
We have 
$$L_\circ = \underbrace{L^-\sqcup\bigsqcup_{e\in\overline{E}}L_e^-\setminus\{0,1\}}_{L_\circ^-}\sqcup \overline{E}^{\circ}.$$
As always, for any label $e\in\overline{E}^\circ$, the mapping:
$$\ket{e}\mapsto \ket{b_e^\circ}=\frac{1}{\sqrt{2}}(\ket{\rightarrow,e}+\ket{\leftarrow,e})$$
has time complexity $1$ (just apply a Hadamard). 

For $\ell\in L^-_e\setminus\{0,1\}$ for some $e$, apply the map $\ket{\ell}\mapsto \ket{b_{\ell}^e}$ in cost $T'$.

For $\ell\in L^-$, apply the map $\ket{\ell}\mapsto \ket{b_{\ell}}$ in cost $T$, and then use the direct sum of basis maps for $\bar{\Psi}_{{\cal B}_{G^e}}^-$ to map:
\begin{equation*}
\frac{1}{\sqrt{2}}(\ket{\rightarrow,e}-\ket{\leftarrow,e})\mapsto \ket{1,e} \mapsto \ket{\bar{b}_1^e}.\qedhere\end{equation*}
\end{proof}

\noindent We complete the analysis by establishing positive and negative witnesses.

\begin{lemma}
Let $\ket{w}$ be a positive witness for $G$, so it necessarily only overlaps switches that are on. For every switch $e$ that is on, let $\ket{w^e}$ be a positive witness for $G^e$. Then
$$\ket{\hat{w}^{\circ}}=\sum_{e\in \overline{E}}\sqrt{\frac{\r^e}{2}}\braket{\rightarrow,e}{w}\ket{\hat{w}^e} + \Pi_{E\setminus \overline{E}}\ket{w}$$
is a positive witness for $G^\circ$.
\end{lemma}
\begin{proof}
We need to show that $\ket{w^\circ}=\ket{s}+\ket{\leftarrow,s}+\ket{\hat{w}^{\circ}}+\ket{\rightarrow,t}+\ket{t}$ is in ${\cal A}_{G^\circ}^\bot\cap {\cal B}_{G^\circ}^\bot$. 
We start by showing that $\ket{w^\circ}$ is orthogonal to everything in ${\cal A}_{G^\circ}$. 
Orthogonality with $\Xi_s^{\cal A}=\mathrm{span}\{\ket{s}-\ket{\leftarrow,s}\}$ (see \defin{canonical-boundary}) is obvious. For $e\in E\setminus\overline{E}$, 
$$\Pi_e^{\cal A}\ket{w^\circ}=\Pi_e^{\cal A}\Pi_{E\setminus\overline{E}}\ket{w}=\Pi_e^{\cal A}\ket{w}=0$$
since $\ket{w}\in {\cal A}_G^\bot$ (because it is a positive witness). 
For all $e\in\overline{E}$ and $e'\in E^e$, 
$$\Pi_{e'}^{e\cal A}\ket{w^\circ} = \sqrt{\frac{\r^e}{2}}\braket{\rightarrow,e}{w}{\Pi_{e'}^{e\cal A}\ket{\hat{w}^e}}=\sqrt{\frac{\r^e}{2}}\braket{\rightarrow,e}{w}{\Pi_{e'}^{e\cal A}\ket{w^e}}=0$$
since $\ket{\hat{w}^e}=(I-\Pi_{B^e})\ket{w^e}$, and $\Pi_{e'}^{e\cal A}\ket{w^e}=0$ because $\ket{w^e}$ is a positive witness for $G^e$. Thus $\ket{w^\circ}\in {\cal A}_{G^\circ}^\bot$. 

By an identical proof, we can show that $\ket{w^\circ}$ is orthogonal to $\Xi_e^{\cal B}$ for all $e\in E\setminus\overline{E}$, and $\Xi_{e'}^{e\cal B}$ for all $e\in \overline{E}$ and $e'\in E^e$. Thus, it remains only to show orthogonality with $\Psi_{{\cal B}_{G^\circ}}^-$.

Note that since $\ket{b_1}\propto \ket{\leftarrow,s^e}-\ket{\rightarrow,t^e}+\sqrt{\r^e}\ket{\bar{b}_1^e}\in {\cal B}_{G^e}$, we must have
\begin{align*}
    \left(\bra{\leftarrow,s^e}-\bra{\rightarrow,t^e}+\sqrt{\r^e}\bra{\bar{b}_1^e}\right)\left(\ket{s^e}-\ket{\leftarrow,s^e}+\ket{\hat{w}^e}+\ket{\rightarrow,t^e}-\ket{t^e}\right) &= 0\\
     \sqrt{\r^e}\braket{\bar{b}_1^e}{\hat{w}^e} &= 2.
\end{align*}
Thus:
\begin{equation}
    \begin{split}
        \Lambda^\dagger\ket{{\hat{w}}^\circ} &= \sum_{e\in\overline{E}}\sqrt{\frac{\r^e}{2}}\braket{\rightarrow,e}{w}{\Lambda^\dagger\ket{\hat{w}^e}}+\Pi_{E\setminus\overline{E}}\ket{w}\\
        &= \sum_{e\in\overline{E}}\sqrt{\frac{\r^e}{2}}\braket{\rightarrow,e}{w}\braket{\bar{b}_1^e}{\hat{w}^e}\frac{1}{\sqrt{2}}\left(\ket{\rightarrow,e}-\ket{\leftarrow,e}\right)+\Pi_{E\setminus\overline{E}}\ket{w}\\
        &= \sum_{e\in\overline{E}}\sqrt{\frac{\r^e}{2}}\braket{\rightarrow,e}{w}\frac{2}{\sqrt{\r^e}}\frac{1}{\sqrt{2}}\left(\ket{\rightarrow,e}-\ket{\leftarrow,e}\right)+\Pi_{E\setminus\overline{E}}\ket{w}\\
        &= \sum_{e\in\overline{E}}\braket{\rightarrow,e}{w}\left(\ket{\rightarrow,e}-\ket{\leftarrow,e}\right)+\Pi_{E\setminus\overline{E}}\ket{w}\\
        &= \Pi_{\overline{E}}\ket{w}+\Pi_{E\setminus\overline{E}}\ket{w} = \ket{\hat{w}},
    \end{split}
\end{equation}
where we used the fact that $\braket{\rightarrow,e}{w}=-\braket{\leftarrow,e}{w}$. It follows that $\Lambda^\dagger\ket{w^\circ}=\ket{w}$, and thus
we have, for any $\ket{\psi}\in\Psi_{{\cal B}_G}^-$,
\begin{equation}
    \bra{\psi}\Lambda^\dagger\ket{w^\circ} = \braket{\psi}{w}=0,
\end{equation}
since $\ket{w}$ is a positive witness for $G$. To complete the proof, we show orthogonality with $\Psi_{{\cal B}_{G^e}}^-\setminus\{\ket{b_0^e},\ket{b_1^e}\}$ for any $e\in\overline{E}$. ${\cal B}_{G^e}$ is orthogonal to $\ket{s^e}$ and $\ket{t^e}$, and furthermore, for $i>1$, we have $\braket{\leftarrow,s^e}{b_i^e}=\braket{\rightarrow,t^e}{b_i^e}=0$. We have:
\begin{align*}
    \braket{b_i^e}{w^\circ}=\sqrt{\frac{\r^e}{2}}\braket{\rightarrow,e}{w}\braket{b_i^e}{\hat{w}^e}.
\end{align*}
Since $\ket{\hat{w}^e}$ is a cropped positive witness for $G^e$, we have:
$$0=\braket{b_i^e}{w}=\braket{b_i^e}{s}-\braket{b_i^e}{\leftarrow,s}+\braket{b_i^e}{\hat{w}^e}+\braket{b_i^e}{\rightarrow,t^e}-\braket{b_i^e}{t}=\braket{b_i^e}{\hat{w}^e},$$
completing the proof.
\end{proof}

\begin{lemma}
    Let $\ket{w_{\cal A}}$ be a negative witness for $G$ that only overlaps switches that are off (this is always true for optimal witnesses, because when a switch is on $\Xi_e^{\cal B}=\Xi_e^{\cal A}$, so might as well put any overlap with $\Xi_e$ into $\ket{w_{\cal B}}$). For every switch $e\in \overline{E}$ that is off, let $\ket{w_{\cal A}^e}$ be a negative witness for $G^e$. Then:
    $$\ket{\hat{w}_{\cal A}^\circ}=\sum_{e\in \overline{E}}\braket{\rightarrow,e}{w_{\cal A}}\sqrt{\frac{2}{\r^e}}{\ket{\hat{w}_{\cal A}^e}} + \Pi_{E\setminus \overline{E}}\ket{w_{\cal A}}$$
    is a negative witness for $G^\circ$.
\end{lemma}
\begin{proof}
We will use the fact that if a switch is off, we always have $\Pi_e\ket{w_{\cal A}}=\braket{\rightarrow,e}{w_{\cal A}}(\ket{\rightarrow,e}-\ket{\leftarrow,e})$. This is true because a switch that is off always has $\Xi_e^{\cal A}=\mathrm{span}\{\ket{\rightarrow,e}-\ket{\leftarrow,e}\}$. In particular, this means that $\Pi_{\overline{E}}\ket{\hat{w}_{\cal A}}\in \Xi_{\overline{E}}^-$, so since $\ket{\hat{w}_{\cal A}}+\ket{\leftarrow,s}-\ket{\rightarrow,t}\in {\cal B}_G$, in particular, it is in ${\cal B}_G^-$. We have:
\begin{align*}
    \ket{\hat{w}_{\cal A}}+\ket{\leftarrow,s}-\ket{\rightarrow,t} &\in {\cal B}_G^-\\
    \Pi_{E\setminus\overline{E}}\ket{\hat{w}_{\cal A}} + \sum_{e\in \overline{E}}\braket{\rightarrow,e}{w_{\cal A}}(\ket{\rightarrow,e}-\ket{\leftarrow,e})+\ket{\leftarrow,s}-\ket{\rightarrow,t} &\in {\cal B}_G^-\\
    \Lambda\left(\Pi_{E\setminus\overline{E}}\ket{\hat{w}_{\cal A}} + \sum_{e\in \overline{E}}\braket{\rightarrow,e}{w_{\cal A}}(\ket{\rightarrow,e}-\ket{\leftarrow,e})+\ket{\leftarrow,s}-\ket{\rightarrow,t}\right) &\in {\cal B}_{G^\circ}\\
    \Pi_{E\setminus\overline{E}}\ket{\hat{w}_{\cal A}} + \sum_{e\in \overline{E}}\braket{\rightarrow,e}{w_{\cal A}}\sqrt{2}\ket{\bar{b}_1^e}+\ket{\leftarrow,s}-\ket{\rightarrow,t} &\in {\cal B}_{G^\circ}\\
    \ket{\hat{w}_{\cal A}^\circ}+\sum_{e\in \overline{E}}\braket{\rightarrow,e}{w_{\cal A}}\sqrt{\frac{2}{\r^e}}\left(\sqrt{\r^e}\ket{\bar{b}_1^e}-\ket{\hat{w}_{\cal A}^e}\right)+\ket{\leftarrow,s}-\ket{\rightarrow,t} &\in {\cal B}_{G^\circ}.
\end{align*}
We complete the proof by noting that for all $e\in\overline{E}$, $\sqrt{\r^e}\ket{\bar{b}_1^e}-\ket{\hat{w}_{\cal A}^e}\in {\cal B}_{G^\circ}$. To see this, we note:
    \begin{equation*}
        \ket{\hat{w}_{\cal A}^e} + \ket{\leftarrow,s}-\ket{\rightarrow,t} \in {\cal B}_{G^e}^-\cap \mathrm{span}\{\ket{b_0^e}\}^\bot,
    \end{equation*}
    which follows from the fact that $\ket{\hat{w}_{\cal A}^e}$ is a cropped negative witness, and so it is orthogonal to $\ket{\leftarrow,s^e}$ and $\ket{\rightarrow,t^e}$, and in particular, $\ket{b_0^e}=\frac{1}{\sqrt{2}}(\ket{\leftarrow,s^e}+\ket{\rightarrow,t^e})$. The only other basis vector of ${\cal B}_{G^e}^-$ that overlaps $\ket{\leftarrow,s^e}$ and $\ket{\rightarrow,t^e}$ is $\ket{b_1^e}\propto \ket{\leftarrow,s^e}-\ket{\rightarrow,t^e}+\sqrt{\r^e}\ket{\bar{b}_0^e}$, from which it follows:
    \begin{equation*}
        \ket{\hat{w}_{\cal A}^e} - \sqrt{\r^e}\ket{\bar{b}_0^e} \in {\cal B}_{G^e}^-\cap \mathrm{span}\{\ket{b_0^e},\ket{b_1^e}\}^\bot\subset {\cal B}_{G^\circ}. \qedhere
    \end{equation*}
\end{proof}

\subsection{Boolean Formula Composition}

In this section, we give a subspace graph computing $\varphi\circ (f_\sigma)_\sigma$ when $\varphi$ is a symmetric formula, given subspace graphs for each $f_\sigma$. Subspace graphs for formulas $\varphi$ can be obtained simply by composing the switching networks for OR and AND in \sec{OR-gadget} and \sec{AND-gadget} (see \cite{jeffery2017stConnFormula}). Our composition theorem, \thm{comp}, generalizes this simple switching network composition to more general subspace graphs with switches.

We assume that $\varphi$ is balanced (see \defin{balanced2}), as otherwise we get a factor in the depth. In the general balanced case, we must be able to generate certain superpositions efficiently, which is slightly complicated to describe (beginning with \defin{reflection-cost}). The special case where $\varphi$ is symmetric (see \defin{balanced}) is much simpler, as shown in \cor{boolean-comp}.

\begin{definition}\label{def:reflection-cost}
    Fix a Boolean formula $\varphi$ on $\{0,1\}^\Sigma$, and real values ${\sf C}_\sigma^+$ and ${\sf C}_\sigma^-$ for each $\sigma\in\Sigma$. Let ${\sf C}_\sigma\colonequals\sqrt{{\sf C}_\sigma^+{\sf C}_\sigma^-}$.

    For $\sigma\in\overline{\Sigma}\setminus\Sigma$, let $d_\sigma$ be the degree of $\sigma$, and $v(\sigma)\in\{\vee,\wedge\}$ indicate the gate labelling the node $\sigma$, so for example, if $v(\sigma)=\vee$, then the sub-formula rooted at $\sigma$ is an OR of $d_\sigma$ sub-formulas. Define: 
$${\sf C}_\sigma \colonequals\sum_{i=1}^d {\sf C}_{\sigma i},
\qquad 
{\sf C}_\sigma^+ = \left\{\begin{array}{ll}
\sum_{i=1}^{d_\sigma}{\sf C}_{\sigma i}^+&\mbox{if }v(\sigma)=\vee\\
\frac{{\sf C}_\sigma^2}{\sum_{i=1}^{d_\sigma}{\sf C}_{\sigma i}^+}&\mbox{if }v(\sigma)=\wedge,
\end{array}\right.
\quad\mbox{and}\quad
{\sf C}_\sigma^- = \left\{\begin{array}{ll}
\frac{{\sf C}_\sigma^2}{\sum_{i=1}^{d_\sigma}{\sf C}_{\sigma i}^-}&\mbox{if }v(\sigma)=\vee\\
\sum_{i=1}^{d_\sigma}{\sf C}_{\sigma i}^-&\mbox{if }v(\sigma)=\wedge.
\end{array}\right.$$
Then note that
$${\sf C}_\sigma = \sqrt{{\sf C}_{\sigma}^+{\sf C}_\sigma^-},
\quad\mbox{and also}\quad 
{\sf C}={\sf C}_\emptyset = \sqrt{\sum_{\sigma\in\Sigma}{\sf C}_\sigma^2}.$$

    We say the values $\{{\sf C}_\sigma^+,{\sf C}_\sigma^-\}_{\sigma\in\Sigma}$ incur logarithmic reflection cost if for every $\sigma\in\overline{\Sigma}$ such that $v(\sigma)=\vee$, a state proportional to $\sum_{i\in [d_\sigma]}\sqrt{{\sf C}_{\sigma i}^+}\ket{i}$ can be generated in complexity $O(\log d_\sigma)$, and for all $\sigma\in\overline{\Sigma}$ such that $v(\sigma)=\wedge$, a state proportional to $\sum_{i\in [d_\sigma]}\sqrt{{\sf C}_{\sigma i}^-}\ket{i}$ can be generated in complexity $O(\log d_\sigma)$.
\end{definition}

\begin{lemma}\label{lem:boolean-comp-gen}
    Let $\varphi$ be a balanced formula (\defin{balanced2}) on $\{0,1\}^\Sigma$.    
    Let $\{f_\sigma\}_{\sigma\in\Sigma}$ be Boolean functions, and for each $\sigma\in \Sigma$, let $G_\sigma$ be an $st$-composable subspace graph (\defin{st-composable-subspace-graph}) computing $f_\sigma$, with working bases that can be generated in time at most $T$ with scaling factor $\r^\sigma$, and $\log\dim H_{G_\sigma}\leq S$. 
    Suppose that for each $\sigma\in\Sigma$, ${\sf C}_\sigma^+$ is a known upper bound on $\r^\sigma\hat{W}_+(G_\sigma)$, and ${\sf C}_\sigma^-$ is a known upper bound on $\hat{W}_-(G_\sigma)/\r^\sigma$, so ${\sf C}_\sigma\colonequals\sqrt{{\sf C}_\sigma^+{\sf C}_\sigma^-}$ is a known upper bound on $\hat{\cal C}(G_\sigma)$. Suppose the values $\{{\sf C}_\sigma^+,{\sf C}_\sigma^-\}_{\sigma\in\Sigma}$ incur logarithmic reflection cost (see \defin{reflection-cost}).
    
    Then there is an $st$-composable subspace graph $G^\circ$ computing $\varphi\circ (f_\sigma)_{\sigma\in \Sigma}$ with 
    $\log\dim H_{G^\circ}=S+O(\log|\Sigma|)$ and 
    $$\hat{\cal C}(G^\circ)^2\leq \sum_{\sigma\in\Sigma}{\sf C}_\sigma^2.$$
    Furthermore, the working bases of $G^\circ$ can be generated in time $T+O(\log |\Sigma|)$.
\end{lemma}

\begin{proof}[Proof of \lem{boolean-comp-gen}]
Let $c$ be a constant such that for each node in $\varphi$, if its subtree has $N$ leaves, and it has $d$ children, then the sub-tree of each child has at most $cN/d$ leaves (this exists because $\varphi$ is balanced), and let $a$ be a constant such that the bases referred to in \lem{OR-gadget} and \lem{AND-gadget} can both be generated in complexity at most $a\log d$, and let $a'$ be a constant such that the cost of generating the basis of $G^\circ$ in \thm{comp} is $T+T'+a'$. We will show by induction on the depth $D$ of $\varphi$ that there exists an $st$-composable subspace graph $G^\circ$ computing $\varphi\circ (f_\sigma)_\sigma$ with composable bases with scaling factor $\r^\circ$ such that:
\begin{enumerate}
    \item $\r^\circ\hat{W}_+(G^\circ)\leq {\sf C}^+$
    \item $\frac{\hat{W}_-(G^\circ)}{\r^\circ}\leq {\sf C}^-$
    \item $\log\dim H_{G^\circ} \leq \log|\Sigma|+(2+\log c)D+S$
    \item the bases can be generated in cost at most $T+(a+a')\log(c^D|\Sigma|)$.
\end{enumerate}
Since $\varphi$ is balanced, $D=O(\log |\Sigma|)$, so $$(a+a')\log(c^D|\Sigma|)=(a+a')\log|\Sigma|+(a+a')D\log c = O(\log|\Sigma|).$$

Without loss of generality, we can assume that none of the leaves of $\varphi$ is negated, because if it is, we can just push this negation into $f_\sigma$. The proof will be by induction on the depth of $\varphi$.

    \paragraph{Base Case:} If $D=0$, then $\varphi\circ (f_\sigma)_{\sigma\in \Sigma} = f_\sigma$ for some $\sigma$. Then we just take $G^\circ=G_\sigma$, so $\r^\circ=\r^\sigma$. We have, by assumption, that $\r^\sigma \hat{W}_+(G_\sigma) \leq {\sf C}_\sigma^+={\sf C}^+$ (when the depth is 0) and $\hat{W}_-(G_\sigma)/\r^\sigma\leq {\sf C}_\sigma^-={\sf C}^-$. We also have, by assumption, that the working bases of $G_\sigma$ can be generated in time $T$, and $\log\dim H_{G_\sigma}\leq S$, completing the base case.

\paragraph{OR Case:} First, suppose $\varphi$ has depth $D>0$, with $\varphi=\bigvee_{i=1}^d\varphi_i$ for some $d$ and some formulas $\varphi_i$ of depth $D-1$. Let 
$G=G_{\textsc{or},d}$ be the subspace graph from \lem{OR-gadget}, which is $st$-composable, because it is a switching network, and let
$G_{\varphi_1},\dots,G_{\varphi_d}$ be the $st$-composable subspace graphs whose existence is promised by the induction hypothesis. For all $i\in [d]$, we will compose the graphs $G_{\varphi_i}$ into the switch edge $e_i$ of $G$ using \thm{comp}, to obtain $G^\circ$. If $\varphi_i(x)=1$ for some $i$, then $G_{\varphi_i}$ must have some cropped positive witness $\ket{\hat{w}^i}$, and using the cropped positive witness $\ket{\hat{w}}$ from \lem{OR-gadget}, by \thm{comp}, $G^\circ$ has a cropped positive witness
$$\ket{\hat{w}^\circ}=\sqrt{\frac{\r^i}{2}}\frac{1}{\sqrt{\w_i}}\ket{\hat{w}^i},$$
where $\r^i$ is the scaling factor of the basis for $G_{\varphi_i}$. 
Thus, if we choose 
$$\w_i={{\sf C}_i^+}/2,
\quad\mbox{so}\quad
\r^\circ=2\sum_{i=1}^d\w_i = {\sum_{i=1}^d{\sf C}_i^+}
$$ 
we get
$$\hat{W}_+(G^\circ)\leq \frac{\r^i}{2\w_i}\hat{W}_+(G_{\varphi_i})
\leq \frac{{\sf C}_i^+}{{\sf C}_i^+}
= 1,
\quad\mbox{and so}\quad 
\r^\circ \hat{W}_+(G^\circ) \leq \sum_{i=1}^d{\sf C}_i^+ = {\sf C}^+
$$
using the induction hypothesis, which says that $\r^i\hat{W}_+(G_{\varphi_i})\leq {\sf C}_i^+$. 

On the other hand, suppose for all $i\in [d]$, $\varphi_i(x)=0$ and so there exists a cropped negative witness $\ket{\hat{w}_{\cal A}^i}$ for $G_{\varphi_i}$. Then using the cropped negative witness from \lem{OR-gadget}, by \thm{comp}, we have a cropped negative witness for $G^\circ$:
$$\ket{\hat{w}^\circ_{\cal A}}=\sum_{i=1}^d\sqrt{\frac{2}{\r^i}}\sqrt{\w_i} \ket{\hat{w}^i_{\cal A}}.$$
Thus:
$$\hat{W}_-(G^\circ)\leq \sum_{i=1}^d\frac{2\w_i}{\r^i}\hat{W}_-(G_{\varphi_i})\leq\sum_{i=1}^d{{\sf C}_i^+}{{\sf C}_i^-}=\sum_{i=1}^d{\sf C}_i^2 = {\sf C}^2,$$
by the induction hypothesis, which says that $\frac{\hat{W}_-(G_{\varphi_i})}{\r^i}\leq {\sf C}_i^-$. 

By the induction hypothesis, the cost of generating the bases for any $G_{\varphi}$ is at most:
$$T'\colonequals (a+a')\log\left(c^{D-1} c|\Sigma|/d\right)=(a+a')\log\left(c^{D} |\Sigma|/d\right).$$
Thus, by \thm{comp}, since the bases of $G$ can be generated in cost at most $T=a\log d$ -- assuming the state proportional to $\sum_i\sqrt{\w_i}\ket{i}\propto \sum_i\sqrt{{\sf C}_i^+}\ket{i}$ can be generated in $O(\log d)$ time --
we can generate the bases of $G^\circ$ in cost 
$$T+T'+a' = (a+a')\log(c^{D-1}(c|\Sigma|/d)) + a\log d+a' \leq (a+a')\log(c^D|\Sigma|),$$
as needed. 

Similarly, by \thm{comp}, we have
\begin{align*}
\dim H_{G^\circ}&\leq \dim H_G-2d+\sum_{i=1}^d \dim H_{G_{\varphi_i}}\\
&\leq 2d+4-2d+d2^{\log (|\Sigma_i|)+(2+\log c)(D-1)+S}\leq 4d2^{\log (|\Sigma_i|)+(2+\log c)(D-1)+S}
\end{align*}
by the induction hypothesis and \lem{OR-gadget}. Thus,
$$\log \dim H_{G^\circ}\leq 2+\log d+\log \frac{c|\Sigma|}{d}+(2+\log c)(D-1)+S=\log |\Sigma|+(2+\log c)D+S.$$

\paragraph{AND Case:} Next, suppose $\varphi=\bigwedge_{i=1}^d\varphi_i$ for some formulas of depth $D-1$. Let 
$G=G_{\textsc{and},d}$ be the graph from \lem{AND-gadget}, which is a switching network, and hence $st$-composable. 
Let
$G_{\varphi_1},\dots,G_{\varphi_d}$ be the $st$-composable subspace graphs whose existence is promised by the induction hypothesis. We will compose the graphs $G_{\varphi_i}$ into the switch edge $e_i$ of $G$ to obtain $G^\circ$. If $\varphi_i(x)=1$ for all $i$, then there exist cropped positive witnesses $\ket{\hat{w}^i}$ for each $G_{\varphi_i}$, so using the cropped positive witness for $G$ from \lem{AND-gadget}, by \thm{comp}, $G^\circ$ has cropped positive witness
$$\ket{\hat{w}^\circ}=\sum_{i=1}^d\sqrt{\frac{\r^i}{2}}\frac{1}{\sqrt{\w_i}}\ket{\hat{w}^i}.$$
Thus, setting 
$$\w_i=\frac{1}{2{\sf C}_i^-}, 
\quad\mbox{so}\quad
\frac{1}{\r^\circ} = {\sum_{i=1}^d\frac{1}{2\w_i}}
= \sum_{i=1}^d{{\sf C}_i^-},$$
we have:
$$\hat{W}_+(G^\circ)\leq \sum_{i=1}^d\frac{\r^i}{2\w_i}\hat{W}_+(G_{\varphi_i})\leq\sum_{i=1}^d{\sf C}_i^-{\sf C}_i^+
={\sum_{i=1}^d{\sf C}_i^2}={{\sf C}^2}
,$$
by the induction hypothesis, which says that $\r^i\hat{W}_+(G_{\varphi_i})\leq {\sf C}_i^+$. 

On the other hand, suppose $\varphi_i(x)=0$ for some $i\in [d]$, so there exists a cropped negative witness $\ket{\hat{w}_{\cal A}^i}$ for $G_{\varphi_i}$. Then using the cropped negative witness from \lem{AND-gadget}, by \thm{comp}, we have a cropped
negative witness for $G^\circ$:
$$\ket{\hat{w}^\circ_{\cal A}}=\sqrt{\frac{2}{\r^i}}\sqrt{\w_i} \ket{\hat{w}^i_{\cal A}}.$$
Thus, we have:
$$\hat{W}_-(G^\circ)\leq \frac{2\w_i}{\r^i} \hat{W}_-(G_{\varphi_i}) \leq \frac{{\sf C}_i^-}{{\sf C}_i^-}=1,
\quad\mbox{and so}\quad
\frac{\hat{W}_-(G^\circ)}{\r^\circ} \leq \sum_{i=1}^d{\sf C}_i^- = {\sf C}^-,$$
using the induction hypothesis, which says that $\frac{\hat{W}_-(G_{\varphi_i})}{\r^i}\leq {\sf C}_i^-$.

A similar argument about the basis efficiency  and dimension of $H_{G^\circ}$ as in the OR case completes the proof.
\end{proof}

We will prove a more easily-applied corollary for the case where $\varphi$ is symmetric. We will use the following two lemmas.

\begin{lemma}\label{lem:sym-C}
    Let $\varphi$ be a symmetric formula (see \defin{balanced}) of depth $D$, and for $D'\in [D]$, let $d_{D'}$ be the out-degree of the nodes whose subtrees have depth $D'$. Assume without loss of generality that the nodes whose children are leaves are $\wedge$ gates.
    Suppose for all $\sigma\in\Sigma$, ${\sf C}^-_\sigma = {\sf L}^-$ is a known upper bound on $\hat{W}_-(G_\sigma)/r^\sigma$, independent of $\sigma$; and ${\sf C}_\sigma^+$ is a known upper bound on $\hat{W}_+(G_\sigma)\hat{W}_-(G_\sigma)$. Let $\{{\sf C}_\sigma^+,{\sf C}_\sigma^-,{\sf C}_\sigma\}_{\sigma\in \overline{\Sigma}}$ be as in \defin{reflection-cost}.
    Then for all $\sigma\in\overline{\Sigma}\setminus\Sigma$ whose subtrees have depth $D'$, 
    $${\sf C}_\sigma^+ = \frac{{\sf C}_\sigma^2}{{\sf L}^-\prod_{j=1}^{\lceil D'/2\rceil}d_{2j-1}}
    \quad\mbox{and}\quad
    {\sf C}_\sigma^- = {\sf L}^-{\prod_{j=1}^{\lceil D'/2\rceil}d_{2j-1}}$$
\end{lemma}
\begin{proof}
First note that we can assume without loss of generality that the nodes with $D'$ odd are labelled by $\wedge$, and $D'$ even are labelled by $\vee$.

    We will prove the statement by induction. Since ${\sf C}_\sigma = \sqrt{{\sf C}_\sigma^+{\sf C}_\sigma^-}$, it is sufficient to prove the statement about ${\sf C}_\sigma^-$. 

    \paragraph{Base Case:} Let $\sigma\in\Sigma$, so $D'=0$. Then, indeed, ${\sf C}_\sigma^-={\sf L}^-={\sf L}^-\prod_{j=1}^0d_{2j-1}$. 

    \paragraph{Odd Case:} Let $D'>0$ be odd. By the induction hypothesis, and since $D'-1$ is even, for all $i\in [d_{D'}]$
    $${\sf C}_{\sigma i}^- = {\sf L}^-\prod_{j=1}^{(D'-1)/2}d_{2j-1}.$$
    Since $D'$ is odd, $v(\sigma)=\wedge$, and so 
    \begin{align*}
        {\sf C}_\sigma^- = {\sum_{i=1}^{d_{D'}}{\sf C}_{\sigma i}^-}
        = d_{D'}{\sf L}^-\prod_{j=1}^{(D'-1)/2}d_{2j-1} = {\sf L}^-\prod_{j=1}^{\lceil D'/2\rceil}d_{2j-1}.
    \end{align*}

    \paragraph{Even Case:} Let $D'>0$ be even. By the induction hypothesis, for all $i\in [d_{D'}]$
    $${\sf C}_{\sigma i}^+ = \frac{{\sf C}_{\sigma i}^2}{{\sf C}_{\sigma i}^-}=\frac{{\sf C}_{\sigma i}^2}{{\sf L}^-\prod_{j=1}^{\lceil(D'-1)/2\rceil}d_{2j-1}}=\frac{{\sf C}_{\sigma i}^2}{{\sf L}^-\prod_{j=1}^{\lceil D'/2\rceil}d_{2j-1}}.$$
    Since $D'$ is even, $v(\sigma)=\vee$, and so 
    \begin{align*}
        {\sf C}_\sigma^+ = {\sum_{i=1}^{d_{D'}}{\sf C}_{\sigma i}^+}
        = \frac{\sum_{i=1}^{d_{D'}}{\sf C}_{\sigma,i}^2}{{\sf L}^-\prod_{j=1}^{\lceil D'/2\rceil}d_{2j-1}}
        =\frac{{\sf C}_\sigma^2}{{\sf L}^-\prod_{j=1}^{\lceil D'/2\rceil}d_{2j-1}}.
    \end{align*}
    Thus 
    \begin{equation*}{\sf C}_\sigma^- = \frac{{\sf C}_\sigma^2}{{\sf C}_\sigma^+} = {\sf L}^-{\prod_{j=1}^{\lceil D'/2\rceil}d_{2j-1}}.\qedhere\end{equation*}
\end{proof}

\begin{corollary}\label{cor:boolean-comp}
    Let $\varphi$ be a symmetric formula (\defin{balanced}) on $\{0,1\}^\Sigma$.    
    Let $\{f_\sigma\}_{\sigma\in\Sigma}$ be Boolean functions, and for each $\sigma\in \Sigma$, let $G_\sigma$ be an $st$-composable subspace graph (\defin{st-composable-subspace-graph}) computing $f_\sigma$, with working bases that can be generated in time at most $T$ with scaling factor $\r^\sigma$, and $\log\dim H_{G_\sigma}\leq S$. 
    Suppose that for each $\sigma\in\Sigma$, ${\sf L}^-$ is a known upper bound on $\hat{W}_-(G_\sigma)/\r^\sigma$, and ${\sf C}_\sigma^+$ is a known upper bound on $\r^\sigma \hat{W}_+(G_\sigma)$.
    Let ${\sf C}_\sigma=\sqrt{{\sf C}_\sigma^+{\sf C}^-}$, which is then a known upper bound on $\hat{\cal C}(G_\sigma)$. For any $\sigma\in\overline{\Sigma}\setminus{\Sigma}$ (i.e. proper prefixes of the strings in $\Sigma$), if the children of $\sigma$ in $\varphi$ are labelled by $[d]$, define 
    $${\sf C}_{\sigma}=\sqrt{\sum_{i\in [d]}{\sf C}_{\sigma i}^2},$$
    and suppose a state proportional to $\sum_{i\in [d]}{\sf C}_{\sigma i}\ket{i}$ can be generating in cost $O(\log d)$.     
    Then there is an $st$-composable subspace graph $G^\circ$ computing $\varphi\circ (f_\sigma)_{\sigma\in \Sigma}$ with 
    $\log\dim H_{G^\circ}=S+O(\log|\Sigma|)$ and
    $$\r^\circ \hat{W}_+(G^\circ)\leq \frac{\sum_{\sigma\in\Sigma}{\sf C}_\sigma^2}{{\sf L}^-\prod_{j=1}^{\lceil D/2\rceil}d_{2j-1}}
    \quad\mbox{and}\quad
    \hat{W}_-(G^\circ)/\r^\circ \leq {\sf L}^-\prod_{j=1}^{\lceil D/2\rceil}d_{2j-1}$$
    so
    $$\hat{\cal C}(G^\circ)^2\leq \sum_{\sigma\in\Sigma}{\sf C}_\sigma^2.$$
    Furthermore, the working bases of $G^\circ$ can be generated in time $T+O(\log |\Sigma|)$.
\end{corollary}
\begin{proof}
By \lem{sym-C}, since we can generate $\sum_{i\in [d]}{\sf C}_{\sigma i}\ket{i}$ in $\log d$ complexity, our upper bounds ${\sf C}_\sigma^{\pm}$ incur logarithmic reflection cost, so we can apply \lem{boolean-comp-gen}.    
\end{proof}

\subsection{Quantum Divide \& Conquer}\label{sec:D-and-C}

In this section, we state our time-efficient quantum divide-\&-conquer results. For an instance $x$ of a function $f_{\ell,n}$, we say that a sub-instance $x'$ of $f_{\ell',n'}$ for some $\ell'\leq \ell$ and $n'\leq n$ is \emph{unit-time computable} if 
any bit $x_i'$ of $x'$ can be computed as a function of $x$ and $i$ in unit time. 
The following is a time-efficient version of Strategy 1 in \cite{childs2022Divide}. Note that it is fairly natural to assume $\varphi = \varphi'(z_1,\dots,z_a)\vee z_{a+1}$ for some formula $\varphi'$, since often $z_{a+1}=f_{\mathrm{aux}}(x)$ simply handles the base case of the recursion.
We prove the statement for symmetric formulas, however, using \lem{boolean-comp-gen} instead of \cor{boolean-comp}, a similar statement, possibly with extra time overhead, holds for any balanced formula $\varphi$, as long as we can account for the possibly more complicated reflection costs.

\begin{theorem}\label{thm:strategy_1}
    Fix a function family $f_{\ell,n}:D_{\ell,n}\rightarrow\{0,1\}$. Fix unit-time-computable functions $\lambda_1,\lambda_2:\mathbb{N}\rightarrow\mathbb{N}$, and a formula $\varphi$ on $\{0,1\}^{a+1}$ such that $\varphi=\varphi'(z_1,\dots,z_a)\vee z_{a+1}$ for some symmetric formula $\varphi'$.
    Suppose $\{{\cal P}_{\mathrm{aux},\ell,n}\}_{\ell,n\in\mathbb{N}}$ is a family of quantum algorithms such that:
    \begin{itemize}
        \item ${\cal P}_{\mathrm{aux},\ell,n}$ decides some $f_{\mathrm{aux},\ell,n}:D_{\ell,n}\rightarrow\{0,1\}$ with time and space complexities $T_{\mathrm{aux}}(\ell,n)$ and  $S_{\mathrm{aux}}(\ell,n)$;
        \item if $\ell\leq \ell_0$, $f_{\ell,n}(x)=f_{\mathrm{aux},\ell,n}(x)$;
        \item if $\ell>\ell_0$, $f_{\ell,n}=\varphi\circ (f_i)_{i\in [a+1]}$ 
        where each $f_i$ for $i\in [a]$ is such that $f_i(x)=f_{\lambda_1(\ell),\lambda_2(n)}(x^i)$ for some unit-time-computable instance $x^i$ of $f_{\lambda_1(\ell),\lambda_2(n)}$, and $f_{a+1}(x)=f_{\mathrm{aux},\ell,n}(x^{a+1})$ for some unit-time computable instance $x^{a+1}$ of $f_{\mathrm{aux},\ell,n}$.
    \end{itemize}
    Then there is a bounded-error quantum algorithm that compute $f_{\ell,n}$ with time complexity $\widetilde{O}(T(\ell,n))$,
    and space complexity $O(S_{\mathrm{aux}}(\ell,n)+\log T(\ell,n))$, 
    where for all $\ell>\ell_0$:
    $$T(\ell,n)\colonequals\sqrt{aT(\lambda_1(\ell),\lambda_2(n))^2+4T_{\mathrm{aux}}(\ell,n)^2}\leq \sqrt{a}T(\lambda_1(\ell),\lambda_2(n))+2T_{\mathrm{aux}}(\ell,n)$$
and for $\ell\leq\ell_0$, $T(\ell,n)=2T_{\mathrm{aux}}(\ell,n)$. 
\end{theorem}
We remark that, while we assume the subroutine ${\cal P}_{\mathrm{aux},\ell,n}$ has no error, if it has sufficiently small error inversely proportional to the number of times it is called, the algorithm must still work, as this cannot be distinguished from having no error. Amplifying a bounded-error quantum algorithm to have such small error incurs factors logarithmic in the number of times it is called, so $2T_{\mathrm{aux}}(\ell,n)$ becomes $O(T_{\mathrm{aux}}(\ell,n)\log T(\ell,n))$. We stress that these log factors are only acceptable because we do not \emph{recursively} call ${\cal P}_{\mathrm{aux},\ell,n}$.
\begin{proof} We will use the shorthand $\lambda(\ell,n)=(\lambda_1(\ell),\lambda_2(n))$.
Let $D_{\varphi'}$ be the depth of $\varphi'$, and for each $j\in [D_{\varphi'}]$, let $d_j$ be the number of children of any node at distance $j$ from the leaves. Define $\bar{d}=\prod_{j=1}^{\lceil D_\varphi/2\rceil}d_j$ for $D_{\varphi'}$. 
Let $T_{\mathrm{aux},0}$ be an upper bound on $T_{\mathrm{aux}}(\ell,n)$ when $\ell\leq \ell_0$. For all $(\ell,n)$, let ${\sf C}^-(\ell,n)$ be defined recursively, as follows.
$${\sf C}^-(\ell,n)\colonequals\left\{\begin{array}{ll}
    T_{\mathrm{aux},0} & \mbox{if } \ell\leq \ell_0\\
    \frac{T(\ell,n)^2}{\frac{aT(\lambda(\ell,n))^2}{\bar{d}{\sf C}^-(\lambda(\ell,n))}+4T_{\mathrm{aux}}(\ell,n)} & \mbox{else.}
\end{array}\right.$$

The proof will be by induction. Specifically, we will show that there is an $st$-composable subspace graph $G_{\ell,n}$ that computes $f_{\ell,n}$ with 
\begin{enumerate}
    \item $\dim H_G \leq 2^{S(\ell,n)}$;
    \item $\hat{W}_-(G_{\ell,n})/r_{\ell,n} \leq {\sf C}^-(\ell,n)$, and $\hat{\cal C}(G_{\ell,n})\leq T({\ell,n})\equalscolon {\sf C}(\ell,n)$;
    \item $\log\dim H_{G_{\ell,n}}=S_{\ell,n}\leq c D\log(2a)+cS_{\mathrm{aux}}(\ell,n)+c\log T_{\mathrm{aux}}(\ell,n)$;
    \item and basis generation cost $L_{\ell,n}\colonequals Dc'\log(2a)+c\log T_{\mathrm{aux}}(\ell,n)$,
\end{enumerate}
where $D$ is the depth of recursion, and $c$ and $c'$ are sufficiently large constants.

\paragraph{Base case:} For the base case, suppose $\ell\leq\ell_0$. Then we let $G_{\ell,n}$ be the subspace graph from \lem{alg-subspace-graph}
derived from the algorithm ${\cal P}_{\mathrm{aux},\ell,n}$ and using $\alpha_r=1$ for all $r$. Then $G_{\ell,n}$ computes $f_{\mathrm{aux},\ell,n}$ with
$\dim H_{G_{\ell,n}}\leq c(S_{\mathrm{aux}}(\ell,n)+\log T_{\mathrm{aux}}(\ell,n))$,
and basis generation cost at most $c\log T_{\mathrm{aux}}(\ell,n)$, for sufficiently large constant $c$. By \lem{alg-subspace-graph}, we also have
$$\hat{W}_-(G_{\ell,n})/\r_{\ell,n} \leq 2{T}_{\mathrm{aux}}(\ell,n)/2 \leq T_{\mathrm{aux},0},
\quad\mbox{and}\quad
\hat{\cal C}(G_{\ell,n})\leq 2{T}_{\mathrm{aux}}(\ell,n).$$

\paragraph{Induction Case:} 
For the induction step, let $G_{\lambda(\ell,n)}$ be the subspace graph whose existence is guaranteed by the induction hypothesis. First, we will use \cor{boolean-comp}, where $G_1,\dots,G_a$ are each copies of $G_{\lambda(\ell,n)}$ -- so ${\sf C}_i\colonequals{\sf C}({\lambda(\ell,n)})=T(\lambda(\ell,n))$ and ${\sf L}^-\colonequals {\sf C}^-(\lambda(\ell,n))$ are known upper bounds on $\hat{\cal C}(G_i)$ and $\hat{W}_-(G_i)/\r^i$, by the induction hypothesis -- to get a subspace graph $G_{\ell,n}^{\varphi'}$. In order to apply \cor{boolean-comp}, we have used the fact that ${\sf L}^-$ does not depend on $i$. By \cor{boolean-comp}, we have
\begin{align*}
\hat{W}_-(G_{\ell,n}^{\varphi'})/\r &\leq {\sf L}^- \bar{d}
= {\sf C}^-(\lambda(\ell,n)) \bar{d}
\equalscolon {\sf C}_0^-(\ell,n)\\
\quad\mbox{and}\quad
\hat{\cal C}(G_{\ell,n}^{\varphi'})^2 &\leq \sum_{i=1}^a {\sf C}_i^2 = a T(\lambda(\ell,n))^2\equalscolon {\sf C}_0({\ell,n})^2. 
\end{align*}
Next, we will use \lem{boolean-comp-gen}, with the formula $x_0\vee x_1$, and subspace graphs $G_0'=G^{\varphi'}_{\ell,n}$ and $G_1'$ a subspace graph computing $f_{\mathrm{aux},\ell,n}$ from \lem{alg-subspace-graph}, using the known upper bounds ${\sf C}_0(\ell,n)$ and ${\sf C}_0^-(\ell,n)$ defined above, and ${\sf C}_1^-(\ell,n)\colonequals T_{\mathrm{aux}}(\ell,n)\geq \hat{W}_-(G_1')\r_1'$ and 
${\sf C}_1(\ell,n)\colonequals 2T_{\mathrm{aux}}(\ell,n) \geq \hat{\cal C}(G_1')$ from \lem{alg-subspace-graph}. 
This gives a subspace graph $G_{\ell,n}$ computing $f_{\ell,n}$
with
\begin{align*}
\hat{\cal C}(G_{\ell,n})^2 &\leq {\sf C}_0(\ell,n)^2+{\sf C}_1(\ell,n)^2
\leq aT(\lambda(\ell,n))^2+4T_{\mathrm{aux}}(\ell,n)^2 = T(\ell,n)^2\\
\hat{W}_-(G_{\ell,n})/\r_{\ell,n}&\leq \frac{{\sf C}_0(\ell,n)^2+{\sf C}_1(\ell,n)^2}{\frac{{\sf C}_0(\ell,n)^2}{{\sf C}_0^-(\ell,n)}+\frac{{\sf C}_1(\ell,n)^2}{{\sf C}_1^-(\ell,n)}}
 \leq \frac{T(\ell,n)^2}{\frac{aT(\lambda(\ell,n))}{\bar{d}{\sf C}^-(\lambda(\ell,n))}+\frac{4T_{\mathrm{aux}}(\ell,n)^2}{T_{\mathrm{aux}}(\ell,n)}}={\sf C}^-(\ell,n).
\end{align*}
In order to apply \lem{boolean-comp-gen}, we must be able to generate a state proportional to 
$$\sqrt{\frac{{\sf C}_0(\ell,n)^2}{{\sf C}_0^-(\ell,n)}}\ket{0}+\sqrt{\frac{{\sf C}_1(\ell,n)^2}{{\sf C}_1^-(\ell,n)}}\ket{1}$$
in constant complexity. Since the amplitudes are fixed, input-independent values, we can do this with a single 2-local unitary. 

The cost of generating the bases for each of the graphs $G_1,\dots,G_{a}$ is at most 
$$L_{\lambda(\ell,n)}=(D-1)c'\log(2a)+c\log T_{\mathrm{aux}}(\ell,n)$$
by the induction hypothesis, and so by \cor{boolean-comp}, the cost to generate the bases for $G^{\varphi'}_{\ell,n}$ is at most
$$c'\log(a)+L_{\lambda(\ell,n)}.$$
The cost to generate the basis for $G_1'$ is at most $c\log T_{\mathrm{aux}}(\ell,n)\leq c'\log(a)+L_{\lambda(\ell,n)}$, by \lem{alg-subspace-graph}. Then by  \lem{boolean-comp-gen}  the cost to generate the basis for $G_{\ell,n}$ is at most:
\begin{align*}
c'\log(2) +c'\log(a)+L_{\lambda(\ell,n)} &\leq c'\log(2a)+(D-1)c'\log(2a)+c\log T_{\mathrm{aux}}(\ell,n)\\
&= Dc'\log(2a)+c\log T_{\mathrm{aux}}(\ell,n)=L_{\ell,n},
\end{align*}
as needed. 
Finally, we note that for all $i\in [a]$, by the induction hypothesis
\begin{align*}
\log\dim H_{G_i}&\leq c (D-1)\log(2a)+cS_{\mathrm{aux}}(\lambda(\ell,n))+c\log T_{\mathrm{aux}}(\lambda(\ell,n))\\
&\leq c( D-1)\log(2a)+cS_{\mathrm{aux}}(\ell,n)+c\log T_{\mathrm{aux}}(\ell,n)
\end{align*}
and so by \cor{boolean-comp}, 
$$\log\dim H_{G^{\varphi'}_{\ell,n}} \leq c( D-1)\log(2a)+cS_{\mathrm{aux}}(\ell,n)+c\log T_{\mathrm{aux}}(\ell,n)+c\log(a).$$
By \lem{alg-subspace-graph}, 
$$\log\dim H_{G_{1}'}\leq cS_{\mathrm{aux}}(\ell,n)+c\log T_{\mathrm{aux}}(\ell,n).$$
Thus, by \lem{boolean-comp-gen}, if $c$ is a sufficiently large constant, 
\begin{align*}
\log\dim H_{G_{\ell,n}} & \leq c( D-1)\log(2a)+cS_{\mathrm{aux}}(\ell,n)+c\log T_{\mathrm{aux}}(\ell,n)+ c\log(a)+c\log(2)\\
&= cD\log(2a)+cS_{\mathrm{aux}}(\ell,n)+c\log T_{\mathrm{aux}}(\ell,n)=S_{\ell,n}.
\end{align*}

This completes the induction. To complete the proof of the theorem, we first apply \cor{scaling} to get a subspace graph $G_{\ell,n}'$ with $\hat{W}_+\leq 1$, and then apply \thm{subspace-graph-to-alg}, which turns $G_{\ell,n}'$ into a bounded-error quantum algorithm for $f_{\ell,n}$ with time complexity 
$$O\left(L_{\ell,n}\sqrt{T(\ell,n)}\right),$$
and space complexity 
$O(S_{\ell,n}+\log T(\ell,n))$.
We complete the proof by noting that
$$\log T(\ell,n) \geq \log \sqrt{a}^D = \frac{D}{2}\log a \geq \frac{D}{4}\log(2a),$$
so $L_{\ell,n}=O(\log(T(\ell,n)))$, and $S_{\ell,n}=O(S_{\mathrm{aux}}(\ell,n)+\log T(\ell,n))$.
\end{proof}

\subsection{Switching Networks with Subroutines}

The following is similar to \lem{boolean-comp-gen}, except that we generalize to any switching network, not just those arising from Boolean formulas (series-parallel switching networks), and we specialize to the case where the composed subspace graphs are of the form in \sec{alg-subspace-graph} (this specialization is not necessary, but it is not clear what we gain from the more general result). 
\begin{lemma}\label{lem:switching-networks-alg}
    Let $G$ be any switching network computing a Boolean function $f$,  
    with working bases that can be generated in time $T_{\mathrm{basis}}$.
    For each $e\in E$, let $\{U_r^e\}_{r=1}^{T_e}$ be a quantum algorithm computing a Boolean function $f_e$ with bounded error. 
    Let $T_{\max}$ be an upper bound on $T_e$ for all $e$, and $F_x$ an $st$-cut-set of $G(x)$ whenever $f(x)=0$. 
    Then there is a subspace graph $G^\circ$ computing $f\circ (f_e)_{e\in E}$ with complexities $\hat{W}_+(G^\circ)=O\left(\max_{x\in f^{-1}(1)}{\cal R}_{s,t}(G(x)) \log T_{\max}\right)$ and $\hat{W}_-(G^\circ)=O\left(\max_{x\in f^{-1}(0)}\sum_{e\in F_x}\w_e T_e^2\right)$, and its bases can be generated in time $T_{\mathrm{basis}}+O(\log T_{\max})$.
\end{lemma}
We remark that sometimes a better negative witness for $G$ can be obtained by taking a superposition of different $st$-cut-sets. This also gives a better negative witness for $G^\circ$, but its form is more complicated, so we omit describing this. 
\begin{proof}
    For each $e\in E$, let $G^e$ be the subspace graph from \lem{alg-subspace-graph} derived from the algorithm $\{U_{r}^e\}_{r}$, using weights $\alpha_r=r+1$. Thus $G^e$ computes $f_e$. 
    We will apply \thm{comp} to $G$, and the $G^e$.

    \paragraph{Positive Analysis:} A positive input for $f\circ (f_e)_e$ gives rise to a positive input for $f$ -- let $\ket{\hat{w}}$ be a positive witness for $G$ under that input. This only uses edges that are turned on, so for each $e$ that is on (i.e. $f_e=1$), let $\ket{\hat{w}^e}$ be a cropped positive witness for $G^e$. By \thm{comp}, there is a cropped positive witness for $G^\circ$:
    $$\ket{\hat{w}^\circ}\colonequals\sum_{e\in E}\sqrt{\frac{\r^e}{2}}\braket{\leftarrow,e}{\hat{w}}\ket{\hat{w}^e},$$
    since $\overline{E}=E$, where $\r^e=2$ is the scaling factor of $G^e$'s basis. Since $G$ is a switching network, its positive witness is a unit $st$-flow, in the sense of \eq{flow-witness}, so $\braket{\leftarrow,e}{\hat{w}}=\frac{\theta(e)}{\sqrt{\w_e}}$, so we have:
    $$\ket{\hat{w}^\circ}=\sum_{e\in E}\frac{\theta(e)}{\sqrt{\w_e}}\ket{\hat{w}^e},$$
    and
    $$\hat{W}_+(G^\circ)\leq \norm{\ket{\hat{w}^\circ}}^2 = \sum_{e\in E}\frac{\theta(e)^2}{\w_e}\norm{\ket{\hat{w}^e}}^2 = \sum_{e\in E}\frac{\theta(e)^2}{\w_e}\hat{W}_+(G^e).$$
    Because we have set $\alpha_r=r+1$, by \lem{alg-subspace-graph}, we have 
    $$\hat{W}_+(G^e)\leq \sum_{r=1}^{T_e}\frac{1}{r+1}=O(\log T_e),$$
    so 
    $$\hat{W}_+(G^\circ)\leq  \sum_{e\in E}\frac{\theta(e)^2}{\w_e}\log T_{\max} = {\cal R}_{s,t}(G(x)) \log T_{\max}.$$

    \paragraph{Negative Analysis:} Similarly, for a negative input to $f\circ (f_e)_e$, let $x$ be its corresponding negative input to $f$, and let $\ket{\hat{w}_{\cal A}}$ be a negative witness for $G$ on input $x$. For each $e$ such that $f_e=0$, let $\ket{\hat{w}_{\cal A}^e}$ be a negative witness for $G^e$. By \thm{comp}, there is a cropped negative witness for $G^\circ$:
    $$\ket{\hat{w}_{\cal A}^\circ}=\sum_{e\in E}\braket{\leftarrow,e}{\hat{w}_{\cal A}}\ket{\hat{w}_{\cal A}^e}.$$
    Since $G$ is a switching network, we can always choose a negative witness of the form 
    $$\ket{\hat{w}_{\cal A}} = \sum_{e\in F_x}\sqrt{\w_e}(\ket{\rightarrow,e}-\ket{\leftarrow,e}),$$
    where $F_x$ is an $st$-cut-set of $G(x)$. Then:
    \begin{align*}
    \hat{W}_-(G^\circ) &\leq \max_{x\in f^{-1}(0)}\sum_{e\in F_x}\w_e\norm{\ket{\hat{w}_{\cal A}^e}}^2= \max_{x\in f^{-1}(0)}\sum_{e\in F_x}\w_e\hat{W}_-(G^e)\\
    &\leq \max_{x\in f^{-1}(0)}\sum_{e\in F_x}\w_e\sum_{r=1}^{T_e}(r+1) = O\left(\max_{x\in f^{-1}(0)}\sum_{e\in F_x}\w_e T_e^2\right).
    \end{align*}

    \paragraph{Basis Generation:} By \thm{comp}, since the bases of $G^e$ can be generated in time $O(\log T_e)$, the bases of $G^\circ$ can be generated in time $T_{\text{basis}}+O(\log T_{\max})$, completing the proof.    
\end{proof}

\noindent Then the main theorem of this section follows directly.
\begin{theorem}\label{thm:switching-network-comp}
    Let $G$ be any switching network computing a Boolean function $f$,  
    with working bases that can be generated in time $T_{\mathrm{basis}}$.
    For each $e\in E$, let $\{U_r^e\}_{r=1}^{T_e}$ be a quantum algorithm computing a Boolean function $f_e$ with bounded error. 
    Let $T_{\max}$ be an upper bound on $T_e$ for all $e$, and $F_x$ an $st$-cut-set of $G(x)$ whenever $f(x)=0$. 
Then there is a bounded error quantum algorithm for $f$ with complexity
$$\widetilde{O}\left(T_{\mathrm{basis}}\sqrt{\max_{x\in f^{-1}(1)}{\cal R}_{s,t}(G(x)) \max_{x\in f^{-1}(0)}\sum_{e\in F_x}\w_eT_e^2}\right).$$
\end{theorem}
\begin{proof}
    Let $G^\circ$ be the subspace graph from \lem{switching-networks-alg}. 
    By \cor{scaling}, there is a subspace graph $G'$, also computing $f\circ (f_e)_e$, whose bases can also be generated in time $T_{\text{basis}}+O(\log T_{\max})$, and such that $\hat{W}_+(G')\leq 1$ and $\hat{W}_-(G')\leq \hat{W}_+(G^\circ)\hat{W}_-(G^\circ)$. Then we can apply \thm{subspace-graph-to-alg} to $G'$ to get a bounded error quantum algorithm for $f\circ (f_e)_e$ with complexity:
    \begin{align*}
    &O\left((T_{\text{basis}}+\log T_{\max})\sqrt{\hat{W}_+(G^\circ)\hat{W}_-(G^\circ)} \right)\\
    ={}&O\left((T_{\text{basis}}+\log T_{\max})\sqrt{\max_{x\in f^{-1}(1)}{\cal R}_{s,t}(G(x)) \max_{x\in f^{-1}(0)}\sum_{e\in F_x}\w_eT_e^2\log T_{\max}} \right).\qedhere
    \end{align*}
\end{proof}

\section{Application to DSTCON}\label{sec:dstcon}

In this section we consider the \textit{directed st-connectivity} problem. First, we specify the graph access model. We will work in the adjacency matrix model, where we assume that for a directed graph $G=(V,E)$ the input is given as an oracle $\mathcal{O}_G$ that can be queried in unit cost, where for any $u,v\in V$ 
$$\mathcal{O}_G: \ket{u}\ket{v}\ket{0}\mapsto \begin{cases}
   \ket{u}\ket{v}\ket{1} & \mbox{if } (u,v)\in E\\
    \ket{u}\ket{v}\ket{0} & \mbox{otherwise.}
\end{cases}$$
Without loss of generality, we assume $(u,u)\in E$ for any $u\in V$. This doesn't change the presence or absence of paths in $G$.\\
However, we note that the time complexity $2^{\frac{1}{2}\log^2n+O(\log n)}$ in \thm{Savitch_quantum} also holds for the edge list model, in which the algorithm can query the $i$-th out- or in-neighbour of a vertex. That is because given such access, we can implement a query to $\mathcal{O}_G$ in poly$(n)$ time and $O(\log n)$ space, and poly$(n)$ factors are suppressed in $2^{O(\log n)}$.

\begin{problem}[\textsc{dstcon}]\label{prob:dstcon}
Given access to a directed graph $G=(V,E)$ via the oracle $\mathcal{O}_G$, and two vertices $s,t\in V$, decide whether there is a directed  path from $s$ to $t$ in $G$.
\end{problem}

There is a classical recursive algorithm for \textsc{dstcon} that operates in the low-space regime due to Savitch \cite{savitch1970relationships}. We first describe the subroutine $\textsc{path}_{\ell}({\cal O}_G,u,v)$ that will be called recursively in the algorithm. This subroutine decides whether there is a path of length at most $\ell$ from $u$ to $v$ in $G$. For $\ell\geq 2$, the subroutine searches over all vertices for some $w$ such that there are paths of length at most $\ell/2$ from $s$ to $w$ and $w$ to $t$:

\begin{subroutine}
  \caption{$\textsc{path}_{\ell}({\cal O}_G,u,v)$ for $\ell\geq 2$}
  \KwIn{Oracle ${\cal O}_G$, $u,v\in V$}
  \KwOut{$1$ if there is a path of length $\leq \ell$ from $u$ to $v$ in $G$; 0 otherwise.}
  \For{$w\in V$}{
    $b_w \colonequals  \textsc{path}(\ell/2,u,w)\land\textsc{path}(\ell/2,w,v)$;
  }
  \Return{$\lor_{w\in V}b_w$}\;
\end{subroutine}

\noindent Next, we describe the base case. When $\ell=1$, the subroutine simply performs a single query ${\cal O}_G(u,v)$ and outputs:
$$\textsc{path}_1({\cal O}_G,u,v)=\begin{cases}
   1 & \mbox{if } (u,v)\in E\\
    0 & \mbox{otherwise.}
\end{cases}$$

\noindent Then Savitch's algorithm simply outputs $\textsc{path}_n({\cal O}_G,s,t)$, where $n=|V|$, to decide if there is a path from $s$ to $t$.

\begin{algorithm}
  \caption{Savitch's algorithm for \textsc{dstcon} \label{alg:Savitch}}
  \KwIn{Oracle ${\cal O}_G$ for $G=(V,E)$ with $ \abs{V}=n$, $s,t\in V$}
  \KwOut{$1$ if there is a path from $s$ to $t$ in $G$; 0 otherwise.}
  \Return{$\textsc{path}_n({\cal O}_G,s,t)$}\;
\end{algorithm}

\noindent The following result of Savitch is easy to verify.
\begin{theorem}\label{thm:Savitch}
    Algorithm \ref{alg:Savitch} decides \textsc{dstcon} in time $O((2n)^{\log n}=2^{\log^2+O(\log n)})$ and space $O(\log^2 n)$.
\end{theorem}

\noindent Next, we show a quantum speedup for Savitch's algorithm via application of \thm{strategy_1}.

\begin{theorem}\label{thm:Savitch_quantum}
    Let $G=(V,E)$ be a directed graph, $\abs{V}=n$. Then there exists a recursive quantum algorithm that decides \textsc{dstcon} on $G$ with bounded error in time $\widetilde{O}((\sqrt{2n})^{\log n})=2^{\frac{1}{2}\log^2n+O(\log n)}$ and space $O(\log^2 n)$.
\end{theorem}

\begin{proof}
    To show existence of such a quantum algorithm, we will rephrase algorithm \ref{alg:Savitch} in terms of the condition of \thm{strategy_1} and analyze its complexity. We start with defining a function corresponding to \textsc{dstcon}.
    $$f_{\ell,n}:\{0,1\}^{n^2}\times\{0,1\}^{\log n}\times\{0,1\}^{\log n}\to \{0,1\}$$
    $$(G,u,v)\mapsto \begin{cases}
   1 & \mbox{if there is a path from $u$ to $v$ in $G$ of length $\le \ell$} \\
    0 & \mbox{otherwise,}
\end{cases}$$
where $G$ encodes a directed graph on $n$ vertices and $u$,$v$ encode two vertices in $G$.\\

Next, we define $\lambda_1(\ell)=\ell/2$ and $\lambda_2(n)=n$. This specifies the recursion. Finally, we set $\ell_0=1$ and define 
$$f_{aux,1,n}(G,u,v)=\begin{cases}
   1 & \mbox{if } (u,v)\in E\\
    0 & \mbox{otherwise.}
\end{cases}$$
and $f_{aux,\ell,n}= 0$ for all $\ell>1$. Then ${\cal P}_{\mathrm{aux},\ell,n}$ is a single edge query when $\ell=1$ and does nothing otherwise. This implies $T_{\mathrm{aux}}(\ell,n)=O(1)$ and $S_{\mathrm{aux}}(\ell,n)=O(\log n)$ as this is the space needed to write down a vertex name. We have everything to write down $f_{\ell,n}$ recursively, in terms of a symmetric formula $\varphi'$ in $a=2n$ variables:

$$f_{\ell,n}(G,u,v)=\underbrace{\bigvee_{w\in V}(f_{\ell/2,n}(G,u,w)\land f_{\ell/2,n}(G,w,v))}_{\varphi'}\lor f_{\mathrm{aux},\ell,n}(G,u,v).$$

Now we apply \thm{strategy_1} to everything described above and conclude that there is a recursive bounded error quantum algorithm that computes $f_{\ell,n}$ with time complexity $\widetilde{O}(T(\ell,n))$ and space complexity $O(S_{\mathrm{aux}}(\ell,n)+\log T(\ell,n))$, such that we have for all $\ell>1$:
$$T(\ell,n)=\sqrt{2n}T(\ell/2,n)+O(1),$$
and $T(1,n)=O(1)$.
By solving the recursion, we obtain $$T(n,n)=O((\sqrt{2n})^{\log n})=O\left(2^{\frac{1}{2}(\log^2n+\log n)}\right).$$
The space complexity of the algorithm for $f_{n,n}$ is thus
$$O(S_{\mathrm{aux}}(n,n)+\log T(n,n))=O(\log^2 n).$$ Since $f_{n,n}$ precisely describes the problem \textsc{dstcon}, this concludes the proof.
\end{proof}

\section{Acknowledgements}

We thank Simon Apers for useful discussions about this work, and Troy Lee for useful comments on an early draft.  

This work is supported by NWO Klein project number OCENW.Klein.061, and ARO contract no W911NF2010327. This work is funded/co-funded by the European Union (ERC, ASC-Q, 101040624). Views and opinions expressed are however those of the author(s) only and do not necessarily reflect those of the European Union or the European Research Council. Neither the European Union nor the granting authority can be held responsible for them. SJ is a CIFAR Fellow in the Quantum Information Science Program. 

This publication is part of the project Divide \& Quantum  (with project number 1389.20.241) of the research programme NWA-ORC which is (partly) financed by the Dutch Research Council (NWO).

This work is supported by the Dutch National Growth Fund (NGF), as part of the Quantum Delta NL programme.

\bibliographystyle{alpha}
\bibliography{refs}

\newcommand{\etalchar}[1]{$^{#1}$}
\newcommand{\lName}{1}\newcommand{\arxiv}[1]{arXiv: \href{https://arxiv.org/abs/#1}{\ttfamily{#1}}\removefirstdot}\newcommand{\arXiv}[1]{arXiv: \href{https://arxiv.org/abs/#1}{\ttfamily{#1}}\removefirstdot}\def\removefirstdot#1{\if.#1{}\else#1\fi}\providecommand{\multiletter}[1]{#1}\renewcommand{\multiletter}[1]{#1}\DeclareRobustCommand{\dutchPrefix}[2]{#2}\providecommand{\dutchPrefix}[2]{#2}\renewcommand{\dutchPrefix}[2]{#2}\newcommand{\skp}[3]{#2}\newcommand{\focs }[1]{\if\lName1\skp{ }{Proceedings of the #1 {IEEE} Symposium on Foundations of Computer Science ({FOCS})}{ }\else{FOCS}\fi}\newcommand{\stoc }[1]{\if\lName1\skp{ }{Proceedings of the #1 {ACM} Symposium on the Theory of Computing ({STOC})}{ }\else{STOC}\fi}\newcommand{\soda }[1]{\if\lName1\skp{ }{Proceedings of the #1 {ACM-SIAM} Symposium on Discrete Algorithms ({SODA})}{ }\else{SODA}\fi}\newcommand{\stacs }[1]{\if\lName1\skp{ }{Proceedings of the #1 Symposium on Theoretical Aspects of Computer Science ({STACS})}{
  }\else{STACS}\fi}\newcommand{\itcs }[1]{\if\lName1\skp{ }{Proceedings of the #1 Innovations in Theoretical Computer Science Conference (ITCS)}{ }\else{ITCS}\fi}\newcommand{\fsttcs }[1]{\if\lName1\skp{ }{Proceedings of the #1 International Conference on Foundations of Software Technology and Theoretical Computer Science (FSTTCS)}{ }\else{FSTTCS}\fi}\newcommand{\mfcs }[1]{\if\lName1\skp{ }{Proceedings of the #1 International Symposium on Mathematical Foundations of Computer Science ({MFCS})}{ }\else{MFCS}\fi}\newcommand{\ccc }[1]{\if\lName1\skp{ }{Proceedings of the #1 {IEEE} Conference on Computational Complexity ({CCC})}{ }\else{CCC}\fi}\newcommand{\isit }[1]{\if\lName1\skp{ }{Proceedings of the #1 {IEEE} International Symposium on Information Theory ({ISIT})}{ }\else{ISIT}\fi}\newcommand{\colt }[1]{\if\lName1\skp{ }{Proceedings of the #1 Conference On Learning Theory (COLT)}{ }\else{COLT}\fi}\newcommand{\nips }[1]{\if\lName1\skp{ }{Advances in Neural Information Processing Systems #1 ({NIPS})}{
  }\else{NIPS}\fi}\newcommand{\aistats }[1]{\if\lName1\skp{ }{Proceedings of the #1 International Conference on Artificial Intelligence and Statistics ({AISTATS})}{ }\else{AISTATS}\fi}\newcommand{\icml }[1]{\if\lName1\skp{ }{Proceedings of the #1 International Conference on Machine Learning (ICML)}{ }\else{ICML}\fi}\newcommand{\icalp }[1]{\if\lName1\skp{ }{Proceedings of the #1 International Colloquium on Automata, Languages, and Programming (ICALP)}{ }\else{ICALP}\fi}\newcommand{\esa }[1]{\if\lName1\skp{ }{Proceedings of the #1 Annual European Symposium on Algorithms (ESA)}{ }\else{ESA}\fi}\newcommand{\tqc }[1]{\if\lName1\skp{ }{Proceedings of the #1 Conference on the Theory of Quantum Computation, Communication, and Cryptography (TQC)}{}\else{TQC}\fi}\newcommand{\jacm }{\if\lName1\skp{ }{Journal of the ACM}{ }\else{J. ACM}\fi}\newcommand{\acmta }{\if\lName1\skp{ }{ACM Transactions on Algorithms}{ }\else{{ACM} Tr. Alg}\fi}\newcommand{\acmtct }{\if\lName1\skp{ }{ACM Transactions on Computation Theory}{
  }\else{ACM Tr. Comp. Th.}\fi}\newcommand{\jams }{\if\lName1\skp{ }{Journal of the AMS}{ }\else{J. AMS}\fi}\newcommand{\pams }{\if\lName1\skp{ }{Proceedings of the AMS}{ }\else{Proc. AMS}\fi}\newcommand{\linalgappl }{\if\lName1\skp{ }{Linear Algebra and its Applications}{ }\else{Lin. Alg. \& App.}\fi}\newcommand{\jalgo }{\if\lName1\skp{ }{Journal of Algorithms}{ }\else{J. Alg.}\fi}\newcommand{\jcss }{\if\lName1\skp{ }{Journal of Computer and System Sciences}{ }\else{J. Comp. Sys. Sci.}\fi}\newcommand{\cc }{\if\lName1\skp{ }{Computational Complexity}{ }\else{Comp. Comp.}\fi}\newcommand{\algor }{\if\lName1\skp{ }{Algorithmica}{ }\else{Alg.}\fi}\newcommand{\comb }{\if\lName1\skp{ }{Combinatorica}{ }\else{Comb.}\fi}\newcommand{\cacm }{\if\lName1\skp{ }{Communications of the ACM}{ }\else{Comm. ACM}\fi}\newcommand{\sigart }{\if\lName1\skp{ }{SIGART Bulletin}{ }\else{SIGART Bull.}\fi}\newcommand{\sigactn }{\if\lName1\skp{ }{SIGACT News}{ }\else{SIGACT News}\fi}\newcommand{\eatcsbul }{\if\lName1\skp{ }{Bulletin of
  the {EATCS}}{ }\else{Bull. {EATCS}}\fi}\newcommand{\siamrev }{\if\lName1\skp{ }{SIAM Review}{ }\else{SIAM Rev.}\fi}\newcommand{\siamjc }{\if\lName1\skp{ }{SIAM Journal on Computing}{ }\else{SIAM J. Comp.}\fi}\newcommand{\siamjo }{\if\lName1\skp{ }{SIAM Journal on Optimization}{ }\else{SIAM J. Opt.}\fi}\newcommand{\siamjdm }{\if\lName1\skp{ }{SIAM Journal on Discrete Mathematics}{ }\else{SIAM J. Disc. Math.}\fi}\newcommand{\siamjnum }{\if\lName1\skp{ }{SIAM Journal on Numerical Analysis}{ }\else{SIAM J. Num. Anal.}\fi}\newcommand{\siamjmathanal }{\if\lName1\skp{ }{SIAM Journal on Mathematical Analysis}{ }\else{SIAM J. Math. Anal.}\fi}\newcommand{\discmath }{\if\lName1\skp{ }{Discrete Mathematics}{ }\else{Disc. Math.}\fi}\newcommand{\das }{\if\lName1\skp{ }{Discrete Applied Mathematics}{ }\else{Disc. App. Math.}\fi}\newcommand{\amatstat }{\if\lName1\skp{ }{Annals of Mathematical Statistics}{ }\else{Ann. Math. Stat.}\fi}\newcommand{\rms }{\if\lName1\skp{ }{Russian Mathematical Surveys}{ }\else{Russ. Math.
  Surv.}\fi}\newcommand{\invmath }{\if\lName1\skp{ }{Inventiones Mathematicae}{ }\else{Inv. Math.}\fi}\newcommand{\jnumber }{\if\lName1\skp{ }{Journal of Number Theory}{ }\else{J. Num. Th.}\fi}\newcommand{\toc }{\if\lName1\skp{ }{Theory of Computing}{ }\else{Th. Comp.}\fi}\newcommand{\cjtcs }{\if\lName1\skp{ }{Chicago Journal of Theoretical Computer Science}{}\else{Chic. J. Th. Comp. Sci.}\fi}\newcommand{\tocsys }{\if\lName1\skp{ }{Theory of Computing Systems}{}\else{Theory Comput. Syst.}\fi}\newcommand{\quantum }{\if\lName1\skp{ }{{Quantum}}{ }\else{Quant.}\fi}\newcommand{\cmp }{\if\lName1\skp{ }{Communications in Mathematical Physics}{ }\else{Comm. Math. Phys.}\fi}\newcommand{\jmp }{\if\lName1\skp{ }{Journal of Mathematical Physics}{ }\else{J. Math. Phys.}\fi}\newcommand{\rspa }{\if\lName1\skp{ }{Proceedings of the Royal Society A}{ }\else{Proc. Roy. Soc. A}\fi}\newcommand{\qic }{\if\lName1\skp{ }{Quantum Information and Computation}{ }\else{Quant. Inf. \& Comp.}\fi}\newcommand{\physrev }{\if\lName1\skp{
  }{Physical Review}{ }\else{Phys. Rev.}\fi}\newcommand{\pra }{\if\lName1\skp{ }{Physical Review A}{ }\else{Phys. Rev. A}\fi}\newcommand{\prb }{\if\lName1\skp{ }{Physical Review B}{ }\else{Phys. Rev. B}\fi}\newcommand{\pre }{\if\lName1\skp{ }{Physical Review E}{ }\else{Phys. Rev. E}\fi}\newcommand{\prr }{\if\lName1\skp{ }{Physical Review Research}{ }\else{Phys. Rev. Research}\fi}\newcommand{\prx }{\if\lName1\skp{ }{Physical Review X}{ }\else{Phys. Rev. X}\fi}\newcommand{\prl }{\if\lName1\skp{ }{Physical Review Letters}{ }\else{Phys. Rev. Lett.}\fi}\newcommand{\njp }{\if\lName1\skp{ }{New Journal of Physics}{ }\else{New J. Phys.}\fi}\newcommand{\prapp }{\if\lName1\skp{ }{Physical Review Applied}{ }\else{Phys. Rev. Appl.}\fi}\newcommand{\physrep }{\if\lName1\skp{ }{Physics Reports}{ }\else{Phys. Rep.}\fi}\newcommand{\rmp }{\if\lName1\skp{ }{Reviews of Modern Physics}{ }\else{Rev. Mod. Phys. }\fi}\newcommand{\phystoday }{\if\lName1\skp{ }{Physics Today}{ }\else{Phys. Today}\fi}\newcommand{\physics
  }{\if\lName1\skp{ }{Physics}{ }\else{Phys.}\fi}\newcommand{\nature }{\if\lName1\skp{ }{Nature}{ }\else{Nat.}\fi}\newcommand{\natcomm }{\if\lName1\skp{ }{Nature Communications}{ }\else{Nat. Comm.}\fi}\newcommand{\natphys }{\if\lName1\skp{ }{Nature Physics}{ }\else{Nat. Phys.}\fi}\newcommand{\npjqi }{\if\lName1\skp{ }{npj Quantum Information}{ }\else{npj Quant. Inf.}\fi}\newcommand{\scirep }{\if\lName1\skp{ }{Scientific Reports}{ }\else{Sci. Rep.}\fi}\newcommand{\science }{\if\lName1\skp{ }{Science}{ }\else{Sci.}\fi}\newcommand{\jpa }{\if\lName1\skp{ }{Journal of Physics A: Mathematical and Theoretical}{ }\else{J. Phys. A}\fi}\newcommand{\ijtp }{\if\lName1\skp{ }{International Journal of Theoretical Physics}{ }\else{Int. J. Th. Phys.}\fi}\newcommand{\jmo }{\if\lName1\skp{ }{Journal of Modern Optics}{ }\else{J. Mod. Opt.}\fi}\newcommand{\jstatph }{\if\lName1\skp{ }{Journal of Statistical Physics}{ }\else{J. Stat. Phys.}\fi}\newcommand{\pnas }{\if\lName1\skp{ }{Proceedings of the National Academy of Sciences}{
  }\else{PNAS}\fi}\newcommand{\lncs }{\if\lName1\skp{ }{Lecture Notes in Computer Science}{ }\else{L. Notes Comp. Sci.}\fi}\newcommand{\lnai }{\if\lName1\skp{ }{Lecture Notes in Artificial Intelligence}{ }\else{L. Notes Art. Int.}\fi}\newcommand{\lnm }{\if\lName1\skp{ }{Lecture Notes in Mathematics}{ }\else{L. Notes Math.}\fi}\newcommand{\tams }{\if\lName1\skp{ }{Transactions of the American Mathematical Society}{ }\else{Trans. AMS}\fi}\newcommand{\ieeetit }{\if\lName1\skp{ }{{IEEE} Transactions on Information Theory}{ }\else{{IEEE} Trans. Inf. Th.}\fi}\newcommand{\iscs }{\if\lName1\skp{ }{International Series in Computer Science}{ }\else{Int. Ser. Comp. Sci.}\fi}\newcommand{\tocl }{\if\lName1\skp{ }{Theory of Computing Library}{ }\else{Th. Comp. Lib.}\fi}
\begin{thebibliography}{CKKD{\etalchar{+}}22}

\bibitem[ABB{\etalchar{+}}23]{allcock2023divide}
Jonathan Allcock, Jinge Bao, Aleksandrs Belovs, Troy Lee, and Miklos Santha.
\newblock On the quantum time complexity of divide and conquer.
\newblock \arxiv{2311.16401}, 2023.

\bibitem[AJPW22]{apers2022ustcon}
Simon Apers, Stacey Jeffery, Galina Pass, and Michael Walter.
\newblock Quantum walks for {USTCON}.
\newblock In {\em \esa{31st}}, pages 10:1--10:17, 2022.
\newblock \arxiv{2212.00094}.

\bibitem[Amb10]{ambainis2010VTSearch}
Andris Ambainis.
\newblock Quantum search with variable times.
\newblock {\em \tocsys}, 47:786--807, 2010.
\newblock \arxiv{quant-ph/0609168}.

\bibitem[BB94]{bonet1994balancing}
M.~L. Bonet and S.~R. Buss.
\newblock Size-depth tradeoffs for boolean formulae.
\newblock {\em Information Processing Letters}, 49:151--155, 1994.

\bibitem[BCJ{\etalchar{+}}13]{belovs2013TimeEfficientQW3Distintness}
Aleksandrs Belovs, Andrew~M. Childs, Stacey Jeffery, Robin Kothari, and Fr\'{e}d\'{e}ric Magniez.
\newblock Time-efficient quantum walks for 3-distinctness.
\newblock In {\em \icalp{40th}}, pages 105--122, 2013.

\bibitem[Bel13]{belovs2013ElectricWalks}
Aleksandrs Belovs.
\newblock Quantum walks and electric networks.
\newblock \arxiv{1302.3143}, 2013.

\bibitem[BJY23]{belovs2023LasVegasTime}
Aleksandrs Belovs, Stacey Jeffery, and Duyal Yolcu.
\newblock Taming quantum time complexity.
\newblock \arxiv{2311.15873}, 2023.

\bibitem[BR12]{belovs2012stConn}
Aleksandrs Blovs and Ben~W. Reichardt.
\newblock Span programs and quantum algorithms for st-connectivity and claw detection.
\newblock In {\em \esa{20th}}, pages 193--204, 2012.

\bibitem[CKKD{\etalchar{+}}22]{childs2022Divide}
Andrew~M. Childs, Robin Kothari, Matt Kovacs-Deak, Aarthi Sundaram, and Daochen Wang.
\newblock Quantum divide and conquer.
\newblock \arxiv{2210.06419}, 2022.

\bibitem[Cor23]{cornelissen2023thesis}
Arjan Cornelissen.
\newblock {\em Quantum multivariate estimation and span program algorithms}.
\newblock PhD thesis, University of Amsterdam, 2023.

\bibitem[DHHM06]{durr2004QQueryCompGraph}
Christoph Dürr, Mark Heiligman, Peter H\o{}yer, and Mehdi Mhalla.
\newblock Quantum query complexity of some graph problems.
\newblock {\em \siamjc}, 35(6):1310--1328, 2006.
\newblock Earlier version in ICALP'04. \arxiv{quant-ph/0401091}.

\bibitem[IJ19]{ito2015approxSpan}
Tsuyoshi Ito and Stacey Jeffery.
\newblock Approximate span programs.
\newblock {\em Algorithmica}, 79:2158--2195, 2019.

\bibitem[Jef22]{jeffery2022subroutines}
Stacey Jeffery.
\newblock Quantum subroutine composition.
\newblock \arxiv{2209.14146}, 2022.

\bibitem[JJKP18]{jarret2018connectivity}
Michael Jarret, Stacey Jeffery, Shelby Kimmel, and Alvaro Piedrafita.
\newblock Quantum algorithms for connectivity and related problems.
\newblock In {\em \esa{26th}}, pages 49:1--49:13, 2018.

\bibitem[JK17]{jeffery2017stConnFormula}
Stacey Jeffery and Shelby Kimmel.
\newblock Quantum algorithms for graph connectivity and formula evaluation.
\newblock {\em Quantum}, 1(26), 2017.

\bibitem[JZ23]{jeffery2022kDist}
Stacey Jeffery and Sebastian Zur.
\newblock Multidimensional quantum walks and application to $k$-distinctness.
\newblock In {\em \stoc{55th}}, pages 1125--1130, 2023.
\newblock \arxiv{2208.13492}.

\bibitem[Kit96]{kitaev1996PhaseEst}
Alexei~Y. Kitaev.
\newblock Quantum measurements and the {A}belian stabilizer problem.
\newblock {\em ECCC}, TR96-003, 1996.
\newblock \arxiv{quant-ph/9511026}.

\bibitem[Lee59]{lee1959switchingnetworks}
C.~Y. Lee.
\newblock Representation of switching functions by binary decision programs.
\newblock {\em Bell Systems Technical Journal}, 38(4):985--999, 1959.

\bibitem[MNRS11]{magniez2006SearchQuantumWalk}
Frédéric Magniez, Ashwin Nayak, Jérémie Roland, and Miklos Santha.
\newblock Search via quantum walk.
\newblock {\em \siamjc}, 40(1):142--164, 2011.
\newblock Earlier version in STOC'07. \arxiv{quant-ph/0608026}.

\bibitem[Pot15]{potechin2015thesis}
Aaron~H. Potechin.
\newblock {\em Analyzing monotone space complexity via the switching network model}.
\newblock PhD thesis, Massachusetts Institute of Technology, 2015.

\bibitem[Rei09]{reichardt2009span}
Ben~W. Reichardt.
\newblock Span programs and quantum query complexity: {T}he general adversary bound is nearly tight for every {B}oolean function.
\newblock In {\em \focs{50th}}, pages 544--551, 2009.
\newblock \arxiv{0904.2759}.

\bibitem[Rei11]{reichardt2009GameTree}
Ben~W. Reichardt.
\newblock Faster quantum algorithm for evaluating game trees.
\newblock In {\em \soda{22nd}}, pages 546--559, 2011.

\bibitem[R{\v{S}}12]{reichardt2011spanformulas}
Ben~W. Reichardt and Robert {\v{S}}palek.
\newblock Span-program-based quantum algorithm for evaluating formulas.
\newblock {\em \toc}, 8(13):291--319, 2012.

\bibitem[Sav70]{savitch1970relationships}
Walter~J Savitch.
\newblock Relationships between nondeterministic and deterministic tape complexities.
\newblock {\em Journal of computer and system sciences}, 4(2):177--192, 1970.

\bibitem[Sha38]{shannon1938switchingNetworks}
Claude~E. Shannon.
\newblock A symbolic analysis of relay and switching networks.
\newblock {\em Transactions of the American Institute of Electrical Engineers}, 57(12):713--723, 1938.

\bibitem[Sha49]{shannon1949switchingNetworks}
Claude~E. Shannon.
\newblock The synthesis of two-terminal switching circuits.
\newblock {\em Bell System Technical Journal}, 28(1):59--98, 1949.

\bibitem[Sze04]{szegedy2004QMarkovChainSearch}
Mario Szegedy.
\newblock Quantum speed-up of {M}arkov chain based algorithms.
\newblock In {\em \focs{45th}}, pages 32--41, 2004.
\newblock \arxiv{quant-ph/0401053}.

\end{thebibliography}

\appendix

\section{Recovering the Local Structure}\label{app:locality-of-basis}

In this appendix, we prove that the following spaces are equal.

\begin{equation}\label{eq:cal-B-G-circ-prime}
    {\cal B}_{G^\circ} = \bigoplus_{u\in V}\widetilde\Lambda({\cal V}_u)\oplus\bigoplus_{e\in\overline{E}}\bigoplus_{u\in V^e\setminus\{s,t\}}{\cal V}_u^e+ \bigoplus_{e\in \overline{E}}\bigoplus_{e'\in E^e}\Xi_{e'}^{e\cal B}
\end{equation}

\begin{equation}\label{eq:cal-B-G-circ}
    \overline{\cal B}_{G^\circ} = \mathrm{span}\underbrace{\Lambda(\Psi_{{\cal B}_G}^-)\cup\bigcup_{e\in\overline{E}}(\Psi_{{\cal B}_{G^e}}^-\setminus\{\ket{b_0^e},\ket{b_1^e}\})}_{\equalscolon \Psi_{{\cal B}_{G^\circ}}^-} \cup \underbrace{\left\{\frac{1}{\sqrt{2}}(\ket{\rightarrow,e'}+\ket{\leftarrow,e'}):e'\in \overline{E}^\circ=\bigcup_{e\in\overline{E}}\overline{E}^e\right\}}_{\mbox{spans }\Xi_{\overline{E}^\circ}^{\circ\cal B}=\bigoplus_{e\in\overline{E}}\bigoplus_{e'\in\overline{E}^e}\Xi_{e'}^{e\cal B}}.
\end{equation}

\begin{lemma}\label{lem:b0-b1-orthog}
    For any $e\in\overline{E}$, 
    $$\mathrm{span}\Psi_{{\cal B}_{G^e}}^-\setminus\{\ket{b_0^e},\ket{b_1^e}\}\oplus \Xi_{\overline{E}^e}^{e\cal B} = {\cal B}_{G^e}\cap \{\ket{\leftarrow,s^e},\ket{\rightarrow,t^e}\}^\bot=\bigoplus_{u\in V^e\setminus\{s,t\}}{\cal V}_u^e + \Xi_{\overline{E}^e}^{e\cal B}.$$
\end{lemma}
\begin{proof}
It is clear that 
\begin{align*}
    \mathrm{span}\Psi_{{\cal B}_{G^e}}^-\setminus\{\ket{b_0^e},\ket{b_1^e}\}\oplus \Xi_{\overline{E}^e}^{e\cal B} &\subseteq {\cal B}_{G^e}=\mathrm{span}\Psi_{{\cal B}_{G^e}}^-\oplus \Xi_{\overline{E}^e}^{e\cal B}\\
    \mbox{and also }     \bigoplus_{u\in V^e\setminus\{s,t\}}{\cal V}_u^e + \Xi_{\overline{E}^e}^{e\cal B} &\subseteq {\cal B}_{G^e}=\bigoplus_{u\in V^e}{\cal V}_u^e + \Xi_{\overline{E}^e}^{e\cal B}.
\end{align*}
Since we also have everything in $\mathrm{span}\Psi_{{\cal B}_{G^e}}^-\setminus\{\ket{b_0^e},\ket{b_1^e}\}$, $\bigoplus_{u\in V^e\setminus\{s,t\}}{\cal V}_u^e$ and $ \Xi_{\overline{E}^e}^{e\cal B}$ orthogonal to both $\ket{\leftarrow,s^e}$ and $\ket{\rightarrow,t^e}$, we have:
\begin{align*}
    \mathrm{span}\Psi_{{\cal B}_{G^e}}^-\setminus\{\ket{b_0^e},\ket{b_1^e}\}\oplus \Xi_{\overline{E}^e}^{e\cal B} &\subseteq {\cal B}_{G^e}\cap \{\ket{\leftarrow,s^e},\ket{\rightarrow,t^e}\}^\bot\\
    \mbox{and }     \bigoplus_{u\in V^e\setminus\{s,t\}}{\cal V}_u^e + \Xi_{\overline{E}^e}^{e\cal B} &\subseteq {\cal B}_{G^e}\cap \{\ket{\leftarrow,s^e},\ket{\rightarrow,t^e}\}^\bot.
\end{align*}

For the other direction, suppose $\ket{\phi}\in {\cal B}_{G^e}\cap \{\ket{\leftarrow,s^e},\ket{\rightarrow,t^e}\}^\bot$. Since it is in ${\cal B}_{G^e}$, for some $\ket{\phi_-}\in \mathrm{span}\Psi_{{\cal B}_{G^e}}^-\setminus\{\ket{b_0^e},\ket{b_1^e}\}$ and $\ket{\phi_+}\in\Xi_{\overline{E}^e}^{e\cal B}$, and scalars $a_0$ and $a_1$, we can express it as:
\begin{align*}
\ket{\phi} &= a_0\ket{b_0^e}+a_1\ket{b_1^e} + \ket{\phi_-}+\ket{\phi_+} \\
&= a_0 \frac{1}{\sqrt{2}}(\ket{\leftarrow,s^e}+\ket{\rightarrow,t^e}) + a_1 \frac{\ket{\leftarrow,s^e}-\ket{\rightarrow,t^e}+\sqrt{\r^e}\ket{\bar{b}_1}}{\sqrt{2+\r^e}} + \ket{\phi_-}+\ket{\phi_+}.
\end{align*}
Since $\ket{\phi_-}$, $\ket{\phi_+}$,  and $\ket{\bar{b}_1}$ are all orthogonal to both $\ket{\leftarrow,s^e}$ and $\ket{\rightarrow,t^e}$, since $\ket{\phi}$ is as well, we must have $a_0=a_1=0$. Thus $\ket{\phi}\in \mathrm{span}\Psi_{{\cal B}_{G^e}}^-\setminus\{\ket{b_0^e},\ket{b_1^e}\}\oplus \Xi_{\overline{E}^e}^{e\cal B}$.

Similarly, 
we can express $\ket{\phi}$ as 
\begin{align*}
\ket{\phi} &= a_0'(\ket{\leftarrow,s}+\ket{\psi_\star(s^e)})+a_1'(\ket{\rightarrow,t}+\ket{\psi_\star(t^e)}) + \ket{\phi_V}+\ket{\phi_+'} 
\end{align*}
for $\ket{\phi_V}\in \bigoplus_{u\in V^e\setminus\{s,t\}}{\cal V}_u^e$ and $\ket{\phi_+'}\in \Xi_{\overline{E}^e}^{e\cal B}$. By the same reasoning as above, we must have $a_0'=a_1'=0$, so $\ket{\phi}\in \bigoplus_{u\in V^e\setminus\{s,t\}}{\cal V}_u^e + \Xi_{\overline{E}^e}^{e\cal B}$.
\end{proof}

\begin{lemma}\label{lem:tilde-Lambda-e-plus}
    For any $e\in\overline{E}$, 
    $$\widetilde\Lambda(\ket{\rightarrow,e}+\ket{\leftarrow,e})\in {\cal B}_{G^e}\cap \{\ket{\leftarrow,s^e},\ket{\rightarrow,t^e}\}^\bot.$$
\end{lemma}
\begin{proof}
Since ${\cal V}_s^e+{\cal V}_t^e\subset {\cal B}_{G^e}$, we have
$$\ket{\leftarrow,s^e}+\ket{\psi_\star(s^e)},  \ket{\rightarrow,t^e}+\ket{\psi_\star(t^e)} \in {\cal B}_{G^e}.
$$
Since $\ket{\leftarrow,s}+\ket{\rightarrow,t}\in {\cal B}_{G^e}$ (because $G^e$ is $st$-composable), 
\begin{align*}
    {\sqrt{\r^e}}\widetilde\Lambda(\ket{\rightarrow,e}+\ket{\leftarrow,e}) &= \ket{\psi_\star(s^e)}+\ket{\psi_\star(t^e)}\\
    &= \ket{\leftarrow,s^e}+\ket{\psi_\star(s^e)} +  \ket{\rightarrow,t^e}+\ket{\psi_\star(t^e)} - (\ket{\leftarrow,s}+\ket{\rightarrow,t}),
\end{align*}
so $\widetilde\Lambda(\ket{\rightarrow,e}+\ket{\leftarrow,e})\in {\cal B}_{G^e}$. Clearly $\widetilde\Lambda(\ket{\rightarrow,e}+\ket{\leftarrow,e})$ is also orthogonal to $\ket{\leftarrow,s^e}$ and $\ket{\rightarrow,t^e}$, concluding the proof.
\end{proof}

\begin{lemma}\label{lem:cal-B-tilde}
    Let $\widetilde\Lambda$ be as in \eq{tilde-Lambda}, and define
    \begin{equation*}
    \widetilde{\cal B}_{G^\circ} = \mathrm{span}\underbrace{\widetilde\Lambda(\Psi_{{\cal B}_G}^-)\cup\bigcup_{e\in\overline{E}}(\Psi_{{\cal B}_{G^e}}^-\setminus\{\ket{b_0^e},\ket{b_1^e}\})}_{\equalscolon \widetilde\Psi_{{\cal B}_{G^\circ}}^-} \oplus\bigoplus_{e\in\overline{E}}\bigoplus_{e'\in\overline{E}^e}\Xi_{e'}^{e\cal B}.
\end{equation*}
Then $\widetilde{\cal B}_{G^\circ} = \overline{\cal B}_{G^\circ}$.
\end{lemma}
\begin{proof}

  Let $\ket{\tilde b^e} = \frac{1}{\sqrt{\r^e}}(\ket{\psi_\star(s^e)}-\ket{\psi_\star(t^e)})$. Then 
    $$\ket{\leftarrow,s^e}-\ket{\rightarrow,t^e}+\sqrt{\r^e}\ket{\tilde{b}^e} \in {\cal V}_s^e+{\cal V}_t^e\subseteq \mathrm{span}\{\ket{b_0^e},\ket{b_1^e},\dots,\ket{b_{\ell}^e}\}\oplus\Xi_{\overline{E}^e}^{e\cal B}.$$
    Then since the only vector that overlaps $\ket{\leftarrow,s^e}-\ket{\rightarrow,t^e}$ is $\ket{b_1^e}=\frac{1}{\sqrt{2+\r^e}}(\ket{\leftarrow,s^e}-\ket{\rightarrow,t^e})+\sqrt{\frac{\r^e}{2+\r^e}}\ket{\bar{b}_1^e}$, we must have: 
    $$\bra{b_1^e}(\ket{\leftarrow,s^e}-\ket{\rightarrow,t^e}+\sqrt{\r^e}\ket{\tilde{b}^e}) = \sqrt{2+\r^e},$$
from which it follows that $\braket{\bar{b}_1^e}{\tilde{b}^e}=1$, and we can write:
\begin{equation}
    \ket{\tilde{b}^e} = \ket{\bar{b}_1^e} + \ket{d^e}
\end{equation}
for some $\ket{d^e}\in {\cal B}_{G^e}$ that is orthogonal to both $\ket{\leftarrow,s^e}-\ket{\rightarrow,t^e}$ and $\ket{\leftarrow,s^e}+\ket{\rightarrow,t^e}$, so by \lem{b0-b1-orthog}, $\ket{d^e}\in\mathrm{span}\Psi_{{\cal B}_{G^e}}^-\setminus\{\ket{b_0^e},\ket{b_1^e}\}\oplus\Xi_{\overline{E}^e}^{e\cal B}$. 

Suppose $\ket{\tilde\psi}\in \widetilde{\cal B}_{G^\circ}$, so we can express it, for some scalars $a_e$, $\ket{\psi_{E\setminus\overline{E}}}\in \Xi_{E\setminus\overline{E}}$ and $\ket{\psi'}\in \mathrm{span}\Psi_{{\cal B}_{G^e}}^-\setminus\{\ket{b_0^e},\ket{b_1^e}\}\oplus\Xi_{\overline{E}^e}^{e\cal B}$ as:
\begin{equation}
\begin{split}
    \ket{\tilde\psi} &= \widetilde\Lambda\Bigg(\underbrace{\sum_{e\in\overline{E}}a_e\frac{1}{\sqrt{2}}(\ket{\rightarrow,e}-\ket{\leftarrow,e})+\ket{\psi_{E\setminus\overline{E}}}}_{\in\mathrm{span}\Psi_{{\cal B}_G}^-}\Bigg) + \ket{\psi'}\\
    &= \sum_{e\in\overline{E}}a_e\ket{\tilde{b}^e}+\ket{\psi_{E\setminus\overline{E}}} + \ket{\psi'}\\
    &= \sum_{e\in\overline{E}}a_e\ket{\bar{b}_1^e}+\ket{\psi_{E\setminus\overline{E}}} + \sum_{e\in\overline{E}}a_e\ket{d^e}+\ket{\psi'}\\
    &= \Lambda\Bigg(\underbrace{\sum_{e\in\overline{E}}a_e\frac{1}{\sqrt{2}}(\ket{\rightarrow,e}-\ket{\leftarrow,e})+\ket{\psi_{E\setminus\overline{E}}}}_{\in\mathrm{span}\Psi_{{\cal B}_G}^-}\Bigg) + \underbrace{\sum_{e\in\overline{E}}a_e\ket{d^e}+\ket{\psi'}}_{\in\mathrm{span}\Psi_{{\cal B}_{G^e}}^-\setminus\{\ket{b_0^e},\ket{b_1^e}\}\oplus\Xi_{\overline{E}^e}^{e\cal B}}
    \in \overline{\cal B}_{G^\circ}.
\end{split}
\end{equation}
Thus, $\widetilde{\cal B}_{G^\circ}\subseteq\overline{\cal B}_{G^\circ}$. 

For the other direction, suppose $\ket{\psi}\in\overline{\cal B}_{G^\circ}$. Then similar to above, we can express it as
\begin{equation}
\begin{split}
    \ket{\psi} &= \Lambda\Bigg(\underbrace{\sum_{e\in\overline{E}}a_e\frac{1}{\sqrt{2}}(\ket{\rightarrow,e}-\ket{\leftarrow,e})+\ket{\psi_{E\setminus\overline{E}}}}_{\in\mathrm{span}\Psi_{{\cal B}_G}^-}\Bigg) + \ket{\psi'}\\
    &= \sum_{e\in\overline{E}}a_e\ket{\bar{b}_1^e}+\ket{\psi_{E\setminus\overline{E}}} + \ket{\psi'}\\
    &= \sum_{e\in\overline{E}}a_e\ket{\tilde{b}_1^e}+\ket{\psi_{E\setminus\overline{E}}} - \sum_{e\in\overline{E}}a_e\ket{d^e}+\ket{\psi'}\\
    &= \widetilde\Lambda\Bigg(\underbrace{\sum_{e\in\overline{E}}a_e\frac{1}{\sqrt{2}}(\ket{\rightarrow,e}-\ket{\leftarrow,e})+\ket{\psi_{E\setminus\overline{E}}}}_{\in\mathrm{span}\Psi_{{\cal B}_G}^-}\Bigg) - \underbrace{\sum_{e\in\overline{E}}a_e\ket{d^e}+\ket{\psi'}}_{\in\mathrm{span}\Psi_{{\cal B}_{G^e}}^-\setminus\{\ket{b_0^e},\ket{b_1^e}\}\oplus\Xi_{\overline{E}^e}^{e\cal B}}
    \in \widetilde{\cal B}_{G^\circ}.
\end{split}
\end{equation}
Thus $\overline{\cal B}{G^\circ}\subseteq\widetilde{\cal B}_{G^\circ}$.
\end{proof}

\begin{theorem}
Let     $\overline{\cal B}_{G^\circ}=\mathrm{span}\Psi_{{\cal B}_{G^\circ}}$ where $\Psi_{{\cal B}_{G^\circ}}$ is as in \eq{cal-B-basis} (see also \eq{cal-B-G-circ}), and ${\cal B}_{G^\circ}$ as in \eq{cal-B-comp} (see also \eq{cal-B-G-circ-prime}). Then
    $\overline{\cal B}_{G^\circ}={\cal B}_{G^\circ}$.
\end{theorem}
\begin{proof}
    We will show that ${\cal B}_{G^\circ}=\widetilde{\cal B}_{G^\circ}$, which is sufficient, by \lem{cal-B-tilde}. We first show that $\widetilde{\cal B}_{G^\circ}\subseteq {\cal B}_{G^\circ}$. First, for any $e\in \overline{E}$,
    it follows from \lem{b0-b1-orthog} that
    \begin{equation}\label{eq:locality1}
       \Psi_{{\cal B}_{G^e}}^-\setminus \{\ket{b_0^e},\ket{b_1^e}\}\oplus \Xi_{\overline{E}^e}^{e\cal B} 
       = \bigoplus_{u\in V^e\setminus\{s,t\}}{\cal V}_u^e+\Xi_{\overline{E}^e}^{e\cal B}\subseteq {\cal B}_{G^\circ}'.
    \end{equation}

    Next, for any $\ket{\psi}\in\Psi_{{\cal B}_G}^-\subseteq \bigoplus_{u\in V}{\cal V}_u$, we have 
    \begin{equation}\label{eq:locality2} 
    \widetilde\Lambda (\ket{\psi}) \in \widetilde\Lambda \left(\bigoplus_{u\in V}{\cal V}_u\right) = \bigoplus_{u\in V}\widetilde\Lambda({\cal V}_u)\subseteq {\cal B}_{G^\circ}'.
    \end{equation}
From \eq{locality1} and \eq{locality2}, it follows that $\widetilde{\cal B}_{G^\circ}\subseteq{\cal B}_{G^\circ}$.
    
    We now show the other direction, ${\cal B}_{G^\circ}\subseteq\widetilde{\cal B}_{G^\circ}$. As in \eq{locality1}, for any $e\in\overline{E}$, 
        \begin{equation}\label{eq:locality3}
        \bigoplus_{u\in V^e\setminus\{s,t\}}{\cal V}_u^e+\Xi_{\overline{E}^e}^{e\cal B}=\Psi_{{\cal B}_{G^e}}^-\setminus \{\ket{b_0^e},\ket{b_1^e}\}\oplus \Xi_{\overline{E}^e}^{e\cal B} 
       \subseteq \widetilde{\cal B}_{G^\circ}.
    \end{equation}

    Next, for any $u\in V$ and $\ket{\psi}\in {\cal V}_u\subseteq{\cal B}_G=\mathrm{span}\Psi_{{\cal B}_G}^-\oplus\Xi_{\overline{E}}^{\cal B}$, we can write 
    $$\ket{\psi}=\ket{\psi_-}+\sum_{e\in\overline{E}}a_e(\ket{\rightarrow,e}+\ket{\leftarrow,e})$$
    for some $\ket{\psi_-}\in \mathrm{span}\Psi_{{\cal B}_G}^-$ and scalars $a_e$. Then:
    \begin{equation}
        \widetilde\Lambda(\ket{\psi}) = \underbrace{\widetilde\Lambda\ket{\psi_-}}_{\in\mathrm{span}\widetilde\Lambda(\Psi_{{\cal B}_G}^-)\subseteq\widetilde{\cal B}_{G^\circ}}+\sum_{e\in\overline{E}}a_e\widetilde\Lambda(\ket{\rightarrow,e}+\ket{\leftarrow,e}).
    \end{equation}
    By \lem{tilde-Lambda-e-plus}, for each $e$, $\widetilde\Lambda(\ket{\rightarrow,e}+\ket{\leftarrow,e})\in {\cal B}_{G^e}\cap \{\ket{\leftarrow,s},\ket{\rightarrow,t}\}^\bot$, which is equal to $\mathrm{span}\Psi_{{\cal B}_{G^e}}^-\setminus\{\ket{b_0^e},\ket{b_1^e}\}\oplus\Xi_{\overline{E}}^{e\cal B}$ by \lem{b0-b1-orthog}. Since $\mathrm{span}\Psi_{{\cal B}_{G^e}}^-\setminus\{\ket{b_0^e},\ket{b_1^e}\}\oplus\Xi_{\overline{E}}^{e\cal B}\subseteq \widetilde{\cal B}_{G^\circ}$, we have $\widetilde\Lambda\ket{\psi}\in \widetilde{\cal B}_{G^\circ}$, from which it follows that
    \begin{equation}\label{eq:locality4}
        \widetilde\Lambda({\cal V}_u) \subseteq \widetilde{\cal B}_{G^\circ}.
    \end{equation}
    Combining \eq{locality3} and \eq{locality4} establishes ${\cal B}_{G^\circ}\subseteq\widetilde{\cal B}_{G^\circ}$, completing the proof.
\end{proof}

\end{document}